\documentclass[sigconf, nonacm]{acmart}

\newcommand\vldbdoi{XX.XX/XXX.XX}
\newcommand\vldbpages{XXX-XXX}
\newcommand\vldbvolume{14}
\newcommand\vldbissue{1}
\newcommand\vldbyear{2020}
\newcommand\vldbauthors{\authors}
\newcommand\vldbtitle{\shorttitle} 
\newcommand\vldbavailabilityurl{URL_TO_YOUR_ARTIFACTS}
\newcommand\vldbpagestyle{plain} 
\usepackage{pgffor}
\usepackage{multirow}
\usepackage{booktabs}
\usepackage{adjustbox}
\usepackage{weiwAlgorithm}
\usepackage{graphicx}
\usepackage{algorithmic}
\usepackage{multirow}
\usepackage{subcaption}
\usepackage{tabularx}
\usepackage{diagbox}
\usepackage{url}
\usepackage{enumitem}
\usepackage{caption}
\usepackage{amsmath}
\usepackage{array}
\newcommand{\myparagraph}[1]{\vspace{1mm} \noindent \textbf{#1}.}

\newcommand{\ie}{{i.e.,}\xspace}
\newcommand{\eg}{{e.g.,}\xspace}
\newcommand{\etal}{et al.\xspace}

\newcommand{\fixme}[1]{{\color{red}#1}}
\newcommand{\myblue}[1]{{\color{black} #1}\xspace}
\newcommand{\nonl}{\renewcommand{\nl}{\let\nl\oldnl}}

\newcommand{\Cmt}[1]{
\vspace{0.6mm} \tcp*[h]{\underline{#1}}\\ \vspace{0.6mm}
}

\newcommand{\sand}{{SandIMIN}\xspace}
\newcommand{\sandl}{{SandIMIN-}\xspace}

\DeclareMathOperator*{\argmax}{arg\,max}

\begin{document}
\title{Efficient Influence Minimization via Node Blocking}

\author{Jinghao Wang}
\affiliation{%
  \institution{University of Technology Sydney}
}
\email{jinghaow.au@gmail.com}

\author{Yanping Wu}
\affiliation{%
  \institution{University of Technology Sydney}
}
\email{yanping.wu@student.uts.edu.au}

\author{Xiaoyang Wang}
\affiliation{%
  \institution{The University of New South Wales}
}
\email{xiaoyang.wang1@unsw.edu.au}

\author{Ying Zhang, Lu Qin}
\affiliation{%
  \institution{University of Technology Sydney}
}
\email{{ying.zhang, lu.qin}@uts.edu.au}

\author{Wenjie Zhang}
\affiliation{%
  \institution{The University of New South Wales}
}
\email{wenjie.zhang@unsw.edu.au}

\author{Xuemin Lin}
\affiliation{%
  \institution{Shanghai Jiaotong University}
}
\email{xuemin.lin@sjtu.edu.cn}

\begin{abstract}
Given a graph $G$, a budget $k$ and a misinformation seed set $S$, \textit{Influence Minimization} (IMIN) via node blocking aims to find a set of $k$ nodes to be blocked such that the expected spread of $S$ is minimized. This problem finds important applications in suppressing the spread of misinformation and has been extensively studied in the literature. However, existing solutions for IMIN still incur significant computation overhead, especially when $k$ becomes large. In addition, there is still no approximation solution with non-trivial theoretical guarantee for IMIN via node blocking prior to our work. In this paper, we conduct the first attempt to propose algorithms that yield data-dependent approximation guarantees. Based on the Sandwich framework, we first develop submodular and monotonic lower and upper bounds for our non-submodular objective function and prove the computation of proposed bounds is \#P-hard. In addition, two advanced sampling methods are proposed to estimate the value of bounding functions. Moreover, we develop two novel martingale-based concentration bounds to reduce the sample complexity and design two non-trivial algorithms that provide $(1-1/e-\epsilon)$-approximate solutions to our bounding functions. Comprehensive experiments on 9 real-world datasets are conducted to validate the efficiency and effectiveness of the proposed techniques. Compared with the state-of-the-art methods, our solutions can achieve up to two orders of magnitude speedup and provide theoretical guarantees for the quality of returned results.

\end{abstract}

\maketitle

\pagestyle{\vldbpagestyle}
\begingroup\small\noindent\raggedright\textbf{PVLDB Reference Format:}\\
\vldbauthors. \vldbtitle. PVLDB, \vldbvolume(\vldbissue): \vldbpages, \vldbyear.\\
\href{https://doi.org/\vldbdoi}{doi:\vldbdoi}
\endgroup
\begingroup
\renewcommand\thefootnote{}\footnote{\noindent
This work is licensed under the Creative Commons BY-NC-ND 4.0 International License. Visit \url{https://creativecommons.org/licenses/by-nc-nd/4.0/} to view a copy of this license. For any use beyond those covered by this license, obtain permission by emailing \href{mailto:info@vldb.org}{info@vldb.org}. Copyright is held by the owner/author(s). Publication rights licensed to the VLDB Endowment. \\
\raggedright Proceedings of the VLDB Endowment, Vol. \vldbvolume, No. \vldbissue\ %
ISSN 2150-8097. \\
\href{https://doi.org/\vldbdoi}{doi:\vldbdoi} \\
}\addtocounter{footnote}{-1}\endgroup

\ifdefempty{\vldbavailabilityurl}{}{
\vspace{.3cm}
\begingroup\small\noindent\raggedright\textbf{PVLDB Artifact Availability:}\\
The source code, data, and/or other artifacts have been made available at \url{https://github.com/wjh0116/IMIN}.
\endgroup
}

\section{Introduction}
With the rapid development of the Internet, various online social networks (OSNs) have thrived, immensely satisfying and facilitating the need for individuals to share their perspectives and acquire information.
Leveraging the established connections between individuals, information and opinions can spread through word-of-mouth effects across OSNs ~\cite{DBLP:journals/tkde/WangZZL17,DBLP:journals/tkde/WangZZLC17,DBLP:conf/icde/WangZZL16}.
However, immense user bases and rapid sharing abilities also make OSNs effective channels for spreading misinformation, which leads to significant harm, such as economic damages and societal unrest~\cite{morozov2009swine,allcott2017social}.
Therefore, it is necessary to implement a series of strategies to minimize the spread of misinformation.
In the literature, strategies for addressing this problem can be categorized into three types: $i)$ positive information spreading~\cite{DBLP:conf/www/BudakAA11,DBLP:conf/infocom/TongWGLLLD17}, which selects a set of nodes to trigger the spread of positive information to fight against the spread of misinformation; $ii)$ edge blocking~\cite{DBLP:conf/pricai/KimuraSM08,DBLP:conf/kdd/KhalilDS14}, which removes a set of edges to decrease the spread of misinformation; $iii)$ node blocking~\cite{DBLP:conf/aaai/WangZCLZX13,DBLP:conf/icde/0002ZW0023}, which removes a set of critical nodes to limit the spread of misinformation.

In this paper, we consider the problem of \textit{Influence Minimization} (IMIN) via node blocking~\cite{DBLP:conf/aaai/WangZCLZX13,DBLP:conf/icde/0002ZW0023}.
Specifically, given a graph $G$, a seed set $S$ of misinformation and a budget $k$, IMIN aims to find a set of $k$ nodes to be removed such that the expected spread of $S$ is minimized. These removed nodes are called blockers. Note that, removing a node causes some nodes previously reachable by misinformation to become unreachable.
We call that these nodes are \textit{protected} by the blocker.
In such a condition, IMIN equals to identify a blocker set with at most $k$ nodes so that the expected number of protected nodes is maximized.

The IMIN problem is NP-hard and APX-hard unless P=NP~\cite{DBLP:conf/icde/0002ZW0023}.
Wang \etal~\cite{DBLP:conf/aaai/WangZCLZX13} first study IMIN via node blocking under the IC model. 
They use Monte-Carlo simulations to estimate the expected decreased spread of misinformation seed set and provide a greedy algorithm to select blockers, \ie iteratively select the node that leads to the largest decreased spread. 
However, the proposed solution is prohibitively expensive on large social networks due to the inefficiency of Monte-Carlo simulations. 
Recently, Xie \etal~\cite{DBLP:conf/icde/0002ZW0023} propose a novel approach based on the dominator tree (formal definition can be found in Section \ref{pre:revisit}) that can effectively and efficiently estimate the decreased spread of misinformation seed set. They also adopt the greedy algorithm, but the difference is that they use the newly proposed estimation method. They observe that the greedy method may miss selecting some critical nodes. 
Therefore, they propose a more effective modified greedy algorithm, which prioritizes the outgoing neighbors of misinformation seed nodes.

However, the solutions in \cite{DBLP:conf/icde/0002ZW0023} require re-estimating the value of decreased spread of misinformation seed set after selecting a blocker in each iteration, which results in prohibitive computation overhead, especially when the budget becomes large.
\myblue{
Besides, since the objective function of IMIN is non-submodular~\cite{DBLP:conf/icde/0002ZW0023}, directly adopting the greedy method cannot provide $(1-1/e)$-approximate solutions~\cite{DBLP:journals/mp/NemhauserWF78}.
Moreover, based on curvature and submodularity ratio~\cite{DBLP:conf/icml/BianB0T17}, we show that the result returned by greedy does not provide any non-trivial guarantees.
Therefore, it remains an open problem to devise efficient algorithms for IMIN with non-trivial theoretical guarantees.}

In this paper, we address this problem based on the Sandwich approximation strategy~\cite{DBLP:journals/pvldb/LuCL15, DBLP:journals/tkde/WangYPCC17}, which is a widely used framework for non-submodular maximization problems.
The general idea of 
Sandwich framework is to first develop monotone nondecreasing and submodular lower and upper bounds for the objective function studied, and then produce solutions with approximation guarantees (breads of the Sandwich) for \myblue{bounding functions maximization (i.e., maximize the lower and upper bounding functions).
The actual effectiveness of the Sandwich-based approach relies on how close the proposed bounds are to the objective function. 
In other words, a loose submodular bound w.r.t. the objective function can also be applied to solve our problem, but it cannot produce satisfactory results in terms of effectiveness, and would only yield trivial data-dependent approximation factor.
For example, a constant function can serve as a trivial submodular upper bound (e.g., an upper bound that equals the number of nodes in the graph). However, it is apparent that this bounding function may adversely affect our results since the solution to the constant function can be arbitrary.
Thus, to provide high-quality results, 
tight bounds with submodularity property for our objective function are required.}

\myblue{Following the Sandwich framework, we first propose appropriate lower and upper bounds for our functions, and prove that the computation of them is \#P-hard.}
\myblue{Additionally, the widely used technique for influence estimation, \textit{Reverse Influence Sampling} (RIS)~\cite{DBLP:conf/soda/BorgsBCL14}, cannot directly extend to the proposed bounding functions estimation.
This is because, different from the classic RIS, which treats all the nodes equally and samples the node uniformly, we need to focus on the nodes who are prone to be affected by the misinformation, since only those nodes can contribute to the bounding functions.
That is, if we estimate the bounding functions using a similar manner of RIS, we need to sample the node based on its probability of being activated by the misinformation.
However, this probability is $\#$P-hard to compute~\cite{DBLP:conf/kdd/ChenWW10}.
To overcome this issue, we propose two novel unbiased estimators based on two new proposed sample sets, \ie CP sequence and LRR set, to estimate the value of lower and upper bounding functions, respectively.

For maximizing the bounding functions with theoretical guarantees, a straightforward approach is to employ OPIM-C~\cite{DBLP:conf/sigmod/TangTXY18}, which is RIS-based and the state-of-the-art method for~\textit{Influence Maximization} (IM).
However, given that RIS cannot be applied to estimate the bounding functions, we cannot inherit the sample complexity from OPIM-C. To tackle this challenge, we first design two novel martingale-based concentration bounds tailored to the new proposed unbiased estimators.
By utilizing these, the sample complexity required to make an unbiased estimate of the bounding functions is significantly reduced, in comparison to the previous concentration bounds used in OPIM-C.
Moreover, to avoid the new derived sample complexity depending on the expected spread of the misinformation $\mathbb{E}[I_G(S)]$, whose computation is $\#$P-hard, we resort to the generalized stopping rule algorithm in~\cite{DBLP:journals/tkde/ZhuTTWL23}, to obtain the value of estimated $\mathbb{E}[I_G(S)]$. Based on the above analysis, we design two non-trivial algorithms, LSBM and GSBM, to maximize lower and upper bounding functions with a provable approximation guarantee of $(1-1/e-\epsilon)$ with high probability, respectively.} 
Finally, we propose a lightweight heuristic LHGA for IMIN, to serve as the filling of the Sandwich. 
By instantiating the Sandwich with LSBM, GSBM and LHGA, our proposed solution \sand can offer a strong theoretical guarantee for the IMIN problem \myblue{(details can be found in Section \ref{sec:appro})}. 
Experiments over 9 real-world graphs are conducted to verify the efficiency and effectiveness of proposed techniques compared with the state-of-the-art solutions~\cite{DBLP:conf/icde/0002ZW0023}.
The main contributions of the paper are summarized as follows.

\begin{itemize}[leftmargin=*,topsep=0pt]

\item 
In this paper, based on the Sandwich search framework, we propose a novel solution \sand for the influence minimization problem via node blocking. 
To the best of our knowledge, we are the first to propose algorithms that yield approximation guarantees for the problem. Submodular and monotonic lower and upper bounds are designed for the objective function, and we prove the computation of bounding functions is \#P-hard.

\item 
To estimate the bounds proposed, two novel sample sets and the corresponding sampling techniques are proposed. 
In addition, new martingale-based concentration bounds are developed to reduce the sample complexity and improve the overall performance. 
Furthermore, we propose two non-trivial algorithms to maximize lower and upper \myblue{bounding functions}, which provide $(1-1/e-\epsilon)$ approximation guarantee with high probability.

\item 
We conduct extensive experiments on 9 real-world graphs to verify the efficiency and effectiveness of proposed techniques. 
Compared with the state-of-the-art algorithms~\cite{DBLP:conf/icde/0002ZW0023}, our solutions show better scalability in terms of dataset size and parameters, and can achieve up to two orders of magnitude speedup.

\end{itemize}

\vspace{1mm}
\myblue{\textit{Note that, due to the limited space, all the proofs are omitted and can be found in our appendix.}}

\section{Preliminaries}\label{sec:pre}

In this section, we first formally define the \textit{Influence Minimization} (IMIN) problem, and then we present an overview of existing solutions for the IMIN problem.

\subsection{Problem Definition}

We consider a directed graph $G=(V,E)$ with a node set $V$ and a directed edge set $E$, where $|V|=n$ and $|E|=m$. 
Given an edge $ \langle u, v \rangle \in E $, we refer to $u$ as an incoming neighbor of $v$ and $v$ as an outgoing neighbor of $u$. 
Each edge $\langle u, v \rangle$ is associated with a propagation probability $p(u,v) \in [0,1]$, representing the probability that $u$ influences $v$. 
\myblue{Table \ref{tab:notation} summarizes the notations frequently used.}

\begin{table}
\small

\caption{\myblue{Frequently used notations}}
\label{tab:notation}
\centering
\begin{tabularx}{\linewidth}{|m{1.5cm}<{\centering}|m{6.245cm}|}
    \hline
    \centering \textbf{Notation} & \multicolumn{1}{c|}{\textbf{Description}} \\
    \hline
    \hline
    \centering $G=(V,E)$ & a social network with node set $V$ and edge set $E$\\
    \hline
    $S,B$ &  the seed set of misinformation and blocker set\\
    \hline
    $\mathbb{E}[I_G(S)]$ &  the expected spread of seed set $S$\\
    \hline
    $G[V^{\prime}]$ &  the subgraph in $G$ induced by node set $V^{\prime}$\\
    \hline
    $\phi,\Omega$ & a realization and the set of all possible realizations \\
    \hline
    $D_S(B)$ &  the expected decreased spread of seed set $S$ after blocking nodes in $B$\\
    \hline
    $D_S^L(\cdot),D_S^U(\cdot)$ &  the submodular and monotonic lower bound and upper bound of $D_S(\cdot)$\\
    \hline
    $B_L^o,B_U^o,B^o$ &  the optimal solution to the lower bounding function, upper bounding function and objective function\\
    \hline
    $C^s,\mathbb{C}^s$ &  a CP sequence and the set of CP sequences\\
    \hline
    $L(v),\mathbb{L}$ &  a LRR set of $v$ and the set of LRR sets\\
    \hline
\end{tabularx}
\end{table}

\myparagraph{Diffusion model} In this paper, we focus on the \textit{independent cascade (IC)} model, which is widely used to simulate the information diffusion in the literature~\cite{DBLP:conf/kdd/KempeKT03,DBLP:conf/soda/BorgsBCL14,DBLP:journals/tkde/WangZZL17,DBLP:journals/tkde/WangZZLC17,DBLP:conf/sigmod/TangTXY18,wu2024targeted}. 
Given a seed set $S \subseteq V$, the diffusion process of $S$ under the IC model unfolds in discrete timestamps, whose details are shown in the following.

\begin{itemize}[leftmargin=*]
\item At timestamp 0, the nodes in the seed set $S$ are activated, while all other nodes are inactive. Each activated node will remain active in the subsequent timestamps.

\item If a node $u$ is activated at timestamp $t$, for each of its inactive outgoing neighbor $v$, $u$ has a single chance to activate $v$ with probability $p(u,v)$ at timestamp $t+1$.
\item The propagation process stops when no more nodes can be activated in the graph $G$.
\end{itemize}

\noindent Given a seed set $S\subseteq V$, let $I_G(S)$ be the number of active nodes in $G$ when the propagation process stops. 
Alternatively, the diffusion process can also be characterized as the \textit{live edge} procedure~\cite{DBLP:conf/kdd/KempeKT03}. 
Specifically, by removing each edge $\langle u, v \rangle \in E$ with $1-p(u,v)$ probability, the remaining graph is referred to as a \textit{realization}, denoted as $\phi$.
Let $I_{\phi}(S)$ denote the number of nodes that are reachable from $S$ in $\phi$.
For any seed set $S$, its expected spread $\mathbb{E}[I_G(S)]$ can be defined as follows.
\begin{align}
    \mathbb{E}[I_G(S)]=\mathbb{E}_{\Phi \sim \Omega}[I_{\Phi}(S)]=\sum_{\phi\in \Omega}I_{\phi}(S)\cdot p(\phi),
\end{align}
where $\Omega$ is the set of all possible realizations of $G$, $\Phi \sim \Omega$ denotes that $\Phi$ is a random realization sampled from $\Omega$ and $p(\phi)$ is the probability for realization $\phi$ to occur.

In this paper, we study the problem of minimizing the spread of misinformation.
One strategy for influence minimization problem is to \textit{block critical nodes} on social networks \cite{DBLP:conf/aaai/WangZCLZX13,DBLP:conf/icde/0002ZW0023}.
When a node $u$ is blocked, we set the probability of all edges pointing to $u$ as 0 and refer to $u$ as a blocker.
We can obtain that the activation probability of a blocker is 0. 
Additionally, we assume that a blocker cannot be a seed node for propagating misinformation. 
Note that, after blocking a node $v$, the status of some nodes changes from active to inactive.
We call these nodes are \textit{protected} by $v$.
Given a seed set $S\subseteq V$ and a blocker set $B\subseteq (V\backslash S)$, we denote $D_S(B)=\mathbb{E}[I_G(S)]-\mathbb{E}[I_{G[V\backslash B]}(S)]$ as the expected decreased spread of seed set $S$ after blocking nodes in $B$, where $G[V\backslash B]$ denotes the subgraph in $G$ induced by node set $V\backslash B$. 

\myparagraph{Problem statement} Given a directed social network $G=(V,E)$, a seed set $S$ for propagating misinformation and a budget $k$, \textit{Influence Minimization} (IMIN) via node blocking is to find a blocker set $B^*$ with at most $k$ nodes such that the influence (\ie expected spread) of seed set $S$ is minimized after blocking nodes in $B^*$. In other words, IMIN aims to identify a blocker set $B^*$ with at most $k$ nodes such that the expected number of protected nodes is maximized, \ie
\begin{align*}
    B^*=\argmax_{B \subseteq (V \backslash S),|B|\leq k} D_S(B).
\end{align*}


As shown in~\cite{DBLP:conf/icde/0002ZW0023}, IMIN is proved to be NP-hard and APX-hard unless P=NP. In addition, given a seed set $S$, $D_S(\cdot)$ is monotonic, but not submodular.
Due to the non-submodularity property of the IMIN objective, the direct use of the greedy framework cannot return a result with an approximation ratio of $(1-1/e)$~\cite{DBLP:journals/mp/NemhauserWF78}.

\subsection{Existing Solutions Revisited}\label{pre:revisit}
Here we first abstract the greedy framework from recent studies \cite{DBLP:conf/aaai/WangZCLZX13,DBLP:conf/icde/0002ZW0023}, and then we introduce the state-of-the-art approaches for addressing the IMIN problem and show their limitations.

\myparagraph{Greedy framework for IMIN}
Suppose $D_S(u|B)=D_S(B\cup \{u\})-D_S(B)$ as the marginal gain of adding $u$ to the set $B$.
In a nutshell, the greedy framework~\cite{DBLP:conf/aaai/WangZCLZX13} starts from an empty blocker set $B=\emptyset$.
The subsequent part of the algorithm consists of $k$ iterations.
At each iteration, it iteratively selects the node $v$ from $V\backslash S$ that leads to the largest ${D}_S(v|B)$ and adds it into $B$.
After selecting one blocker, the probability of all edges pointing to it is set as 0. 

Due to $D_S(B)=\mathbb{E}[I_G(S)]-\mathbb{E}[I_{G[V\backslash B]}(S)]$, one feasible way for calculating $D_S(B)$ is to compute $\mathbb{E}[I_G(S)]$.
However, the computation of $\mathbb{E}[I_G(S)]$ is proved as $\#$P-hard \cite{DBLP:conf/kdd/ChenWW10}, which means that the computation of $D_S(B)$ is also $\#$P-hard.
In \cite{DBLP:conf/aaai/WangZCLZX13}, Wang \etal use Monte-Carlo simulations to estimate the influence spread $\mathbb{E}[I_G(S)]$ and adopt the greedy framework. It can solve the IMIN problem effectively but due to the inefficiency of Monte-Carlo simulations, it incurs significant computation overhead.

\myparagraph{The state-of-the-art approach}
Compared to the Monte-Carlo based estimation method under the greedy framework, the state-of-the-art approach \cite{DBLP:conf/icde/0002ZW0023} optimizes the estimation for $D_S(\cdot)$.
Instead of estimating the expected spread (\ie $\mathbb{E}[I_G(S)]$), Xie \etal~\cite{DBLP:conf/icde/0002ZW0023} directly estimate the expected decreased spread (\ie $D_S(\cdot)$) based on the \textit{dominator tree} (DT)~\cite{DBLP:books/lib/AhoU72,DBLP:journals/cacm/LowryM69}.
To explain how this estimation algorithm works, we first introduce three concepts as follows.

\begin{definition}[Dominator]
    Given a realization $\phi$ and a source $s$, a node $u$ is called a dominator of a node $v$ if and only if all paths from $s$ to $v$ pass through $u$.
\end{definition}

\begin{definition}[Immediate Dominator]
    Given a realization $\phi$ and a source $s$, a node $u$ is said to be an immediate dominator of a node $v$, denoted as $idom(v)=u$, if and only if $u$ dominates $v$ and every other dominator of $v$ dominates $u$.
\end{definition}

\begin{definition}[Dominator Tree (DT)]
\label{def:DT}
    Given a realization $\phi$ and a source $s$, the dominator tree of $\phi$ is induced by the edge set $\{ \langle idom(u),u\rangle:u \in V\backslash \{s\}\}$ with root $s$.
\end{definition}

According to the concept of DT, each DT has one root. 
Xie \etal~\cite{DBLP:conf/icde/0002ZW0023} first propose to create a unified seed node $s$ to replace the given seed set $S$ under the IC model. 
For each node $u \in (V \backslash S)$, if there are $h$ distinct seed nodes pointing to node $u$, and each edge has a probability $p_i$ ($1\leq i \leq h$), they will remove all edges from the seed nodes to node $u$, and add an edge from node $s$ to node $u$ with the probability $(1-\prod_{i=1}^h(1-p_i))$. 
Correspondingly, $D_S(\cdot)$ is replaced by $D_s(\cdot)$.
In addition, Xie \etal~\cite{DBLP:conf/icde/0002ZW0023} prove that for each node $u\in (V \backslash S)$, the expected number of protected nodes by $u$ (\ie $D_s(u)$) equals the expected size of the subtree rooted at $u$ in the DT. 
Based on these, the estimation algorithm of \cite{DBLP:conf/icde/0002ZW0023} runs in the following steps, where $\hat{D}_{s}(\cdot)$ is the estimated value of $D_s(\cdot)$.

\begin{itemize}[leftmargin=*]
    \item Generate a certain number of random realizations $\mathcal{G}$ from $G$. 
    \item For each generated realization $\phi \in \mathcal{G}$, apply Lengauer-Tarjan algorithm~\cite{DBLP:journals/toplas/LengauerT79} to construct the DT of $\phi$, which roots at $s$.
    For each node $v \in (V\backslash S)$, measure the size of subtree with the root $v$ in DT, which is denoted as $c_{\phi}(v)$.
    \item For each node $v \in (V\backslash S)$, the estimated value of $D_{s}(v)$ is the average value of $c_{\phi}(v)$ in all realizations that are generated, \ie $\hat{D}_{s}(v)=(\sum_{\phi \in \mathcal{G}}c_{\phi}(v))/|\mathcal{G}|$.
\end{itemize}


\begin{figure}[t]
\centering
    \begin{minipage}{0.14\textwidth}
        \centering
        \includegraphics[width=\textwidth]{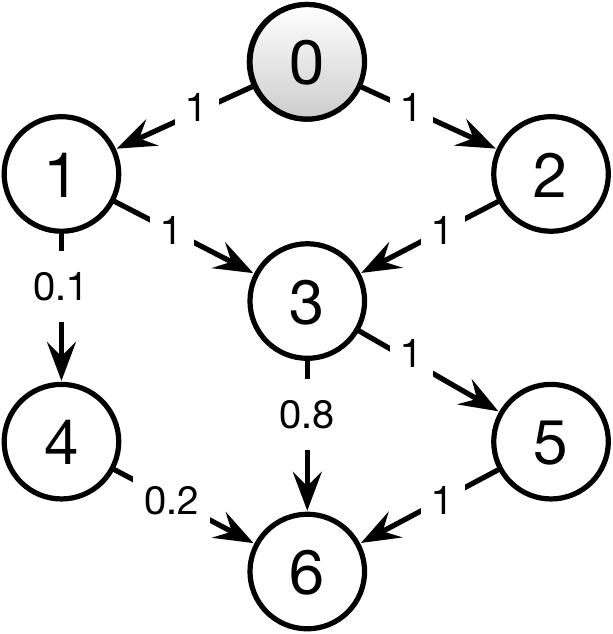}
        \subcaption{\small Social network $G$}
    \end{minipage} 
    \hspace{0.02\textwidth}
    \begin{minipage}{0.14\textwidth}
        \centering
        \includegraphics[width=\textwidth]{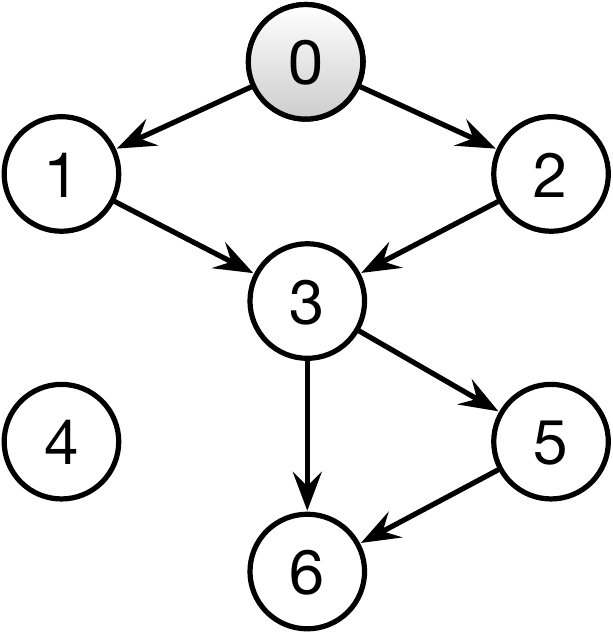}
        \subcaption{A realization $\phi$}
    \end{minipage} 
    \hspace{0.02\textwidth}
    \begin{minipage}{0.14\textwidth}
        \centering
        \begin{subfigure}{\textwidth}
            \includegraphics[width=\textwidth]{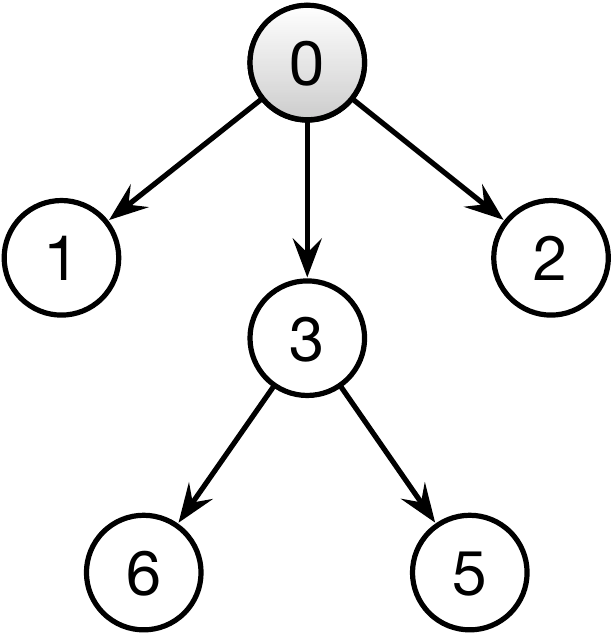}
            \subcaption{DT of $\phi$}
        \end{subfigure}
    \end{minipage}
	\caption{Example of estimating $D_{s}(\cdot)$}
	\label{example_DT}
\end{figure}

\begin{example}\label{exam:dt}
Here we illustrate an example of the above estimation procedure with one realization.
Figure \ref{example_DT}(a) shows a social network $G=(V,E)$, where $v_0$ is the misinformation seed node.
The number associated with each edge is its corresponding propagation probability.
Figure \ref{example_DT}(b) shows a realization $\phi$ obtained from $G$.
In Figure \ref{example_DT}(c), a DT rooted at $s$ of $\phi$ is constructed by Lengauer-Tarjan algorithm~\cite{DBLP:journals/toplas/LengauerT79}. 
Then, we can get the estimated value of $D_{s}(v)$ for each node $v \in (V\backslash \{v_0\})$ by measuring the size of subtree with root $v$ in DT, \ie $\hat{D}_{s}(v_1)=\hat{D}_{s}(v_2)=\hat{D}_{s}(v_5)=\hat{D}_{s}(v_6)=1$, $\hat{D}_{s}(v_3)=3$, and $\hat{D}_{s}(v_4)=0$ since $v_4$ cannot be activated by $v_0$ in $\phi$. 
\end{example}


By utilizing the greedy framework but with one difference, Xie \etal~\cite{DBLP:conf/icde/0002ZW0023} propose AdvancedGreedy (AG).
That is, they use the above estimation method to obtain the estimated value of $D_s(\cdot|B)$.
However, some critical nodes may be missed by AG, \eg some outgoing neighbors of the seed nodes. 
Reconsider Figure \ref{example_DT}(a), suppose $k=2$. $v_0$ has two outgoing neighbors and if we directly select these two nodes as blockers, $v_0$ cannot activate any node. 
However, if we block the blockers returned by AG (\ie $v_3$ and $v_1$), $v_2$ cannot be protected.
To address this issue, Xie \etal~\cite{DBLP:conf/icde/0002ZW0023} further propose GreedyReplace (GR), which consists of two stages. 
In the first stage, the outgoing neighbors of seed nodes are stored in the candidate set.
They iteratively select the node $v$ from the candidate set with the largest $\hat{D}_s(v|B)$ and add it into $B$, until $|B|=\min\{d_{s}^{out},k\}$, where $d_{s}^{out}$ denotes the number of nodes in the candidate set. 
In the second stage, they consider processing the blockers in $B$ according to the reverse order of their insertion order. 
For each blocker in $B$, they first remove it from $B$, called the replaced node.
Then they select the node $v$ from $V\backslash S$ with the largest $\hat{D}_s(v|B)$ and add it into $B$, called the current best blocker.
If the replaced node is the current best blocker, they return $B$ directly. Otherwise, continue the replacement process. 
Compared with AG, GR can achieve a better result quality.
However, in most cases, AG is more efficient since GR requires two stages to select nodes, and often requires multiple rounds of replacement process in the second stage before it terminates, resulting in additional time cost for GR. 

\myparagraph{Limitations}
Despite the efficiency of dominator tree based estimation method, AG and GR still incur significant computation overhead in practice.
This is because when AG/GR selects a node as a blocker, it needs to remove that node from the graph. 
Therefore, AG/GR cannot reuse the realizations generated in the last round.
Correspondingly, Xie \etal need to regenerate realizations and construct the corresponding DTs, which incurs significant time cost for large values of $k$.
\myblue{Moreover, there is no theoretical analysis provided for AG and GR, which are both based on the greedy framework.
    Based on the curvature and submodularity ratio~\cite{DBLP:conf/icml/BianB0T17}, in this paper, we analyze the approximation guarantee for the greedy strategy on our non-submodular objective.
    The submodularity ratio serves as a metric to assess how closely the objective approximates being submodular.
    Formally, for all $A, B \subset V$, the submodularity ratio of $D_s(\cdot)$ is the largest scalar $\psi$ such that,
\begin{align}\label{eq:sub_ratio}
    \scalebox{1}{$\sum_{\omega \in A \backslash B}[D_s(B\cup \{\omega\})-D_s(B)]\geq \psi[D_s(B\cup A)-D_s(B)].$}
\end{align}
Considering an example with the graph $G=(V,E)$, where $V=\{v_0,v_1,\ldots,v_{n-1}\}$ and $E=\{\langle v_0, v_1 \rangle,\langle v_0, v_2 \rangle,\langle v_1, v_3 \rangle,\langle v_2, v_3 \rangle,\langle v_3, v_4 \rangle\\, \langle v_3, v_5 \rangle,\ldots,\langle v_3, v_{n-1} \rangle\}$.
The probability on each edge is set to 1 and $v_0$ is the misinformation seed node.
When $B=\emptyset$ and $A=\{v_1,v_2\}$, the left side of Eq. (\ref{eq:sub_ratio}) is equal to 2, and the right side of Eq. (\ref{eq:sub_ratio}) is equal to $n-1$, \ie $2 \geq \psi \cdot(n-1)$.
As $n$ gradually becomes larger, the submodularity ratio $\psi$ approaches 0 infinitely, consequently 
using the greedy strategy yields approximation guarantee that also approaches 0 infinitely~\cite{DBLP:conf/icml/BianB0T17}.
Thus, the result returned by the state-of-the-art algorithms for IMIN does not provide any non-trivial guarantees.}
To fill these gaps, in this paper, we design efficient approximation algorithms with theoretical guarantees for IMIN.

\section{Sandwich Approximation Strategy}
\label{sec:bounds}
To solve the IMIN problem, we propose efficient approximation algorithms based on the Sandwich framework~\cite{DBLP:journals/pvldb/LuCL15, DBLP:journals/tkde/WangYPCC17}, which is widely used for non-submodular maximization.
In Section \ref{sandwich:framework}, we first give the general framework of our algorithm.
Then we propose the lower and upper bounds in Section \ref{bound:lower} and \ref{bound:upper}, respectively.

\subsection{Overview of \sand}
\label{sandwich:framework}
Generally, our \sand algorithm first finds the $\alpha_1$-approximate solution and $\alpha_2$-approximate solution (breads of Sandwich) to the lower bound and the upper bound of the objective function, respectively.
Then, it finds a solution (filling of Sandwich) to the original problem with a heuristic method.
Finally, it returns the best solution among these three results.
The pseudocode of the above process is shown in Algorithm \ref{alg:Sandwich} and has the following result,
\begin{algorithm}[t]
{
    \SetVline
    \footnotesize
    \caption{\sand}
    \label{alg:Sandwich}
    \Input{The graph $G=(V,E)$, the seed set $S$, the unified seed node $s$, the budget $k$ and the error parameters $\alpha_1$, $\alpha_2$, $\delta$, $\gamma$.}
    \Output{The blocker set $B$.}
    \Cmt{LSBM Algorithm in Section \ref{boundsmax:lower}}
    \State{$B_L\leftarrow $ the $\alpha_1$-approximate solution for \myblue{lower bounding function maximization}}
    \Cmt{GSBM Algorithm in Section \ref{boundsmax:upper}}
    \State{$B_U\leftarrow $ the $\alpha_2$-approximate solution for \myblue{upper bounding function maximization}}
    \Cmt{LHGA Algorithm in Section \ref{boundsmax:original}}
    \State{$B_R\leftarrow $ the heuristic solution for original problem}
    \State{$\hat{I}_{G[V\backslash \cdot]}(s)\leftarrow $ the $(\gamma,\delta)$-estimate of $\mathbb{E}[I_{G[V\backslash \cdot]}(s)]$}
    \State{$B \leftarrow \arg \min_{B^* \in \{B_L,B_U,B_R\}}\hat{I}_{G[V\backslash B^*]}(s)$}

    \State{\textbf{return} $B$}
}
\end{algorithm}


\begin{align}\label{eq:sandwich}
    D_s(B) \geq \max \left\{\frac{D_s^L(B_L^o)}{D_s(B^o)}\alpha_1,\frac{D_s(B_U)}{D_s^U(B_U)}\alpha_2\right\}\frac{1-\gamma}{1+\gamma}D_s(B^o)\text{,}
\end{align}
where $D_s^L$, $D_s^U$ are the non-negative, monotonic and submodular set functions defined on $V$, \ie $D_s^L: 2^V \rightarrow \mathbb{R}_{\geq 0}$ and $D_s^U: 2^V \rightarrow \mathbb{R}_{\geq 0}$, such that $\forall B \subseteq (V\backslash S)$, $D_s^L(B) \leq D_s(B) \leq D_s^U(B)$.
$B_L^o$, $B_U^o$ and $B^o$ are the optimal solutions to maximize the \myblue{lower bounding function, upper bounding function} and objective function, respectively. 
In addition, we call $\hat{\mu}$ is the $(\gamma,\delta)$-estimate of $\mu$ if $\hat{\mu}$ satisfies:
\begin{align}
    \Pr[(1-\gamma)\mu\leq\hat{\mu}\leq(1+\gamma)\mu]\geq 1-\delta.
\end{align}

As observed, the key to \sand is to find the lower and upper bounds of the objective function, which are both monotonic and submodular. 
Before presenting our bounds, we first introduce the technique of how to transfer multiple seeds to one seed for presentation simplicity.
We create a unified seed node $s$ and then introduce the edges with the propagation probability of 1 from $s$ to every seed node $v \in S$. 
Note that, $s$ is the virtual node and it does not belong to $V$.
Under such a setting, we can guarantee that $s$ and $S$ have the same spread under the IC model and there is no need to pre-compute $1-\prod_{i=1}^h(1-p_i)$ for each node $u \in (V \backslash S)$ as stated in Section \ref{pre:revisit}. 
In the following, we use $D_s(\cdot)$ to denote $D_S(\cdot)$.

\myparagraph{Roadmap of \sand} In Section \ref{section:LB} and Section \ref{section:UB}, we propose submodular and monotonic lower and upper bounds of the objective function, respectively. 
We then design sampling methods to estimate the value of lower and upper bounds in Section \ref{section:est_LB} and Section \ref{section:est_UB}, respectively. In Section \ref{section:LSBM} and \ref{section:GSBM}, we devise two approximation algorithms that provide  $(1-1/e-\epsilon)$-approximate solutions for lower and upper \myblue{bounding functions} maximization, respectively (Lines 1-2 of Algorithm \ref{alg:Sandwich}). In Section \ref{section:LHGA}, we devise a heuristic method for IMIN (Line 3 of Algorithm \ref{alg:Sandwich}) and show that \sand yields a data-dependent approximation guarantee.

\subsection{Lower Bound}\label{section:LB}
\label{bound:lower}
A  function $f:2^V \rightarrow \mathbb{R}_{\geq 0}$ is submodular if for any $S \subseteq T \subseteq V$ and any $x \in (V \backslash T)$, $f(\cdot)$ satisfies  $f(S \cup \{x\})-f(S)\geq f(T \cup \{x\})-f(T)$.
The reason for the non-submodularity of function $D_s(\cdot)$ is due to the \textit{combination effect} of nodes in the blocker set. 
That is, to prevent a node from being activated by $s$, we need to simultaneously block two or more nodes.
For example, reconsidering Figure \ref{example_DT}, to protect $v_3$, we need to simultaneously block $v_1$ and $v_2$. 
Motivated by this, we disregard the combination effect to obtain a lower bound that is submodular. 
Specifically, we only consider nodes that can be protected by blocking only one node in $G$.
Accordingly, given a blocker set $B$, the lower bound of $D_s(B)$ can be defined as:
\begin{align}
    D_s^L(B)=\mathbb{E}_{\Phi \sim \Omega}[|\cup_{v \in B}{N}_{\Phi}(v)|]=\sum_{\phi\in \Omega}p(\phi)\cdot |\cup_{v \in B}N_{\phi}(v)|\text{,}
\end{align}
where $N_{\phi}(v)$ denotes the set of nodes whose status changes from active to inactive under $\phi$ after $v$ is blocked. 
\begin{lemma}
\label{lemma:lower_submodular}
    Given a seed set $S$ and its unified seed node $s$,
    $D_s^L(\cdot)$ is monotone nondecreasing and submodular under the IC model.
\end{lemma}
\myblue{Due to the space limitation, the detailed proof for Lemma \ref{lemma:lower_submodular} and other omitted proofs can be found in our appendix.}
In addition, since $D_s^L(\emptyset)=\mathbb{E}[I_G(s)]$ and computing $\mathbb{E}[I_G(s)]$ is $\#$P-hard~\cite{DBLP:conf/kdd/ChenWW10}, the computation of $D_s^L(\cdot)$ is $\#$P-hard.

\subsection{Upper Bound}\label{section:UB}
\label{bound:upper}
For each node $v$ activated by $s$, $v$ can be protected if all the paths from $s$ to $v$ are blocked. 
Intuitively, we can obtain the upper bound of the objective function by relaxing the condition for nodes to be protected.
We call that $v$ can be \textit{alternative-protected} if there exists one path from $s$ to $v$ is blocked. 
Given a realization $\phi=\left(V({\phi}),E({\phi})\right)$, let ${\phi}^s=\left(V({\phi}^s),E({\phi}^s)\right)$ be the graph of misinformation receivers under $\phi$.
Specifically, $V({\phi}^s)=R_{\phi}(s) \backslash S$, where $R_{\phi}(s)$ is the set of nodes reachable from $s$ in $\phi$ and 
$E({\phi}^s)=\{\langle u, v \rangle \in E({\phi}):u \in V({\phi}^s) , v \in V({\phi}^s) \}$. 
In such a setting, for any node $v \in V({\phi}^s)$, $v$ can alternative-protect those nodes that are reachable from $v$ in~${\phi}^s$.
Given a blocker set $B$, the upper bound of $D_s(B)$ can be defined as:
\begin{align}
    D_s^U(B)=E_{\Phi \sim \Omega}[|M_{\phi}(B)|]=\sum_{\phi\in \Omega}p(\phi)\cdot |M_{\phi}(B)|\text{,}
\end{align}
where $M_{\phi}(B)$ is the set of nodes reachable from $B$ in ${\phi}^s$. 
Upper \myblue{bounding function maximization} is essentially the influence maximization problem~\cite{DBLP:conf/kdd/KempeKT03}, thus it also possesses monotonicity, submodularity and NP-hardness, and the computation of $D_S^U(\cdot)$ is $\#$P-hard.

    

    

\begin{figure}[t]
	\centering
\includegraphics[width=0.7\linewidth]{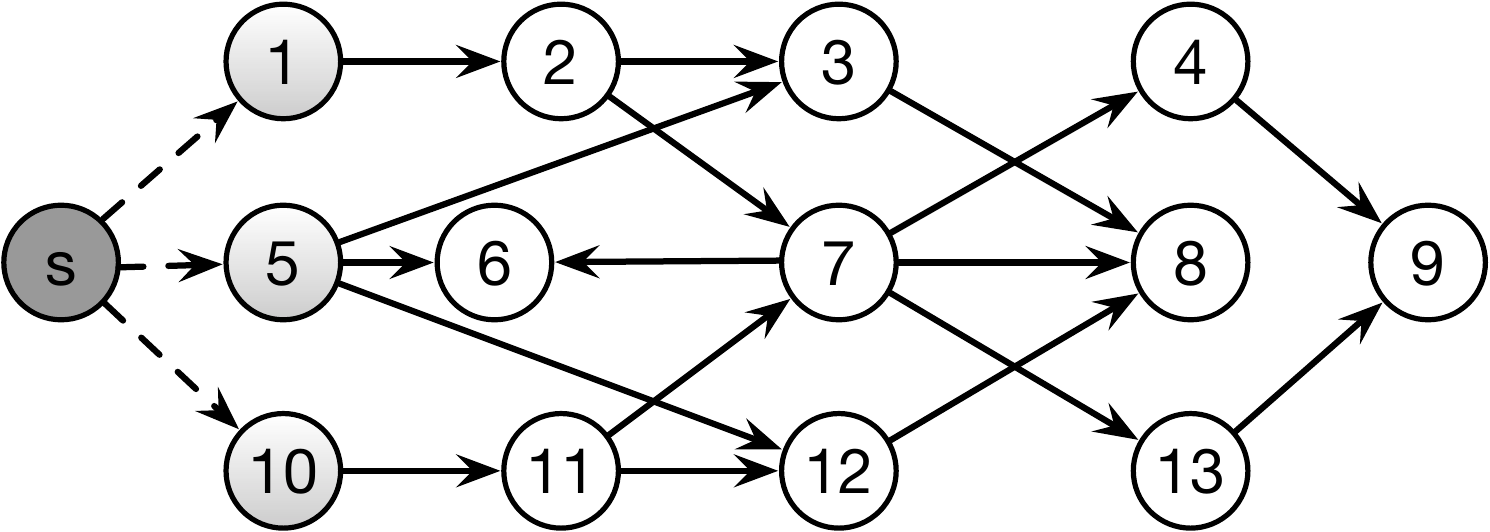}
	\caption{Example of the bounds}
	\label{example:bounds}
\end{figure}

\begin{example}
Here we first illustrate how to transfer multiple seeds to one seed.
As shown in Figure \ref{example:bounds}, the misinformation seed set $S=\{v_1,v_5,v_{10}\}$. We create a unified seed node $s$ and three edges, $\langle s, v_1 \rangle,\langle s, v_5 \rangle,\langle s, v_{10} \rangle$, with the propagation probability of 1.

Then we illustrate an example of our proposed bounds.
Suppose the blocker set is $B=\{v_3,v_7,v_{12} \}$.
For the objective function of IMIN, $D_s(B)=|\{v_3,v_4,v_7,v_8,v_9,v_{12},v_{13}\}|=7$.
For the lower bound, $D_s^L(B)=|\{v_3,v_4,v_7,v_9,v_{12},v_{13}\}|=6$.
For the upper bound, $D_s^U(B)=|\{v_3,v_4,v_6,v_7,v_8,v_9,v_{12},v_{13}\}|=8$.
\end{example}

\section{Bounding Functions Estimation}
\label{sec:boundesti}
To achieve Lines 1-2 in Algorithm \ref{alg:Sandwich}, we first need to compute $D_s^L(\cdot)$ and $D_s^U(\cdot)$.
However, the computation of them is $\#$P-hard.
\myblue{Additionally, the state-of-the-art technique for influence estimation (\ie RIS~\cite{DBLP:conf/soda/BorgsBCL14}) cannot be applied to estimate our bounding functions.
Specifically, we only need to consider the nodes that can be reached by misinformation, since only these nodes need to be protected and can contribute to the bounding functions.
Besides, the number of these nodes is unknown due to the randomness of propagation.}
To address this issue, in Section \ref{section:est_LB} and Section \ref{section:est_UB}, 
we devise two sampling methods named Local Sampling and Global Sampling to estimate the lower bound and upper bound, respectively.

\subsection{\myblue{Lower Bounding Function Estimation}}\label{section:est_LB}

In the following, we first clarify some concepts mentioned in the sampling technique.
\begin{definition} [Common Path (CP) Set \& Common Path (CP) Sequence]
Given a graph $G=(V, E)$, a seed set $S\subseteq V$, a unified source node $s$ and a realization $\phi$ obtained from $G$, 
for any node $v\in (V\backslash S)$, a Common Path (CP) set of $v$, denoted by $C_{\phi}(s,v)$, is the set of common nodes on all paths from $s$ to $v$ in $\phi$ (exclude $S$), 
\ie $C_{\phi}(s,v)=\{u\in (V\backslash S) :u \in \bigcap_{i=1}^j P_i(s,v)$\}, where $P_i(s,v)$ ($1 \leq i \leq j$) denotes the set of nodes on a path from $s$ to $v$ and $j$ is the number of paths from $s$ to $v$. 
Let $R_{\phi}(s)$ be the set of nodes reachable from $s$ in $\phi$.
A Common Path (CP) sequence in $\phi$, denoted by $C_{\phi}^s$, is the set of CP sets of the nodes
in $R_{\phi}(s)$
(exclude $S$), 
\ie $C^s_{\phi}=\{C_{\phi}(s,v):v\in (R_{\phi}(s)\backslash S)\}$.
\end{definition}

\begin{example}
Reconsider the social network $G$ in Figure \ref{example_DT}.
Given the realization $\phi$ of $G$ in Figure \ref{example_DT}(b) and the misinformation seed node $v_0$, we first illustrate the CP set of node $v_6$ as follows.
We can find that there are four paths between $v_0$ and $v_6$, \ie $P_1(v_0,v_6)=\{v_1,v_3,v_6\}$, $P_2(v_0,v_6)=\{v_2,v_3,v_6\}$, $P_3(v_0,v_6)=\{v_1,v_3,v_5,v_6\}$ and $P_4(v_0,v_6)=\{v_2,v_3,v_5,v_6\}$.
Note that, $P_i(v_0,v_6) (1\leq i\leq 4)$ does not include $v_0$.
The CP set of $v_6$ in $\phi$ is the common nodes on four paths, \ie $C_{\phi}(v_0,v_6)=\{v_3,v_6\}$.
Due to $R_{\phi}(v_0)=\{v_1,v_2,v_3,v_5,v_6\}$, a CP sequence in $\phi$ is the set of CP sets of nodes in $R_{\phi}(v_0)$.
\end{example}

\textbf{In this paper, $\phi$ can be dropped when it is clear from the context. }
Given a blocker set $B \subseteq (V \backslash S)$ and a 
set of CP sequences $\mathbb{C}^s$, we use $Cov_{\mathbb{C}^s}(B)$ denote the coverage of $B$ in $\mathbb{C}^s$, \ie $Cov_{\mathbb{C}^s}(B)=\sum_{C^s \in \mathbb{C}^s}\sum_{C(s,v) \in C^s}\min\{|B \cap C(s,v)|,1\}$.
We can estimate $D_s^L(B)$ by generating a certain number of CP sequences. 
Lemma \ref{lem:lowerest} shows that $\frac{Cov_{\mathbb{C}^s}(B)}{|\mathbb{C}^s|}$ is an unbiased estimator of $D_s^L(B)$.

\begin{lemma}
\label{lem:lowerest}
Given a misinformation seed set $S\subseteq V$, a unified seed node $s$ and the set of CP sequences $\mathbb{C}^s$, for any blocker set $B \subseteq (V \backslash S)$,
    \begin{align}
        D_s^L(B)=\mathbb{E}[\frac{Cov_{\mathbb{C}^s}(B)}{|\mathbb{C}^s|}]\text{,}
    \end{align}
    where the expectation is taken over the random choices of $\mathbb{C}^s$.
\end{lemma}

\begin{algorithm}[t]
{
    \SetVline
    \footnotesize
    \caption{Local Sampling}
    \label{alg:LS}
    \Input{The graph $G=(V,E)$, the seed set $S$ and the unified seed node $s$.}
    \Output{The CP sequence $C_{\phi}^s$.}
    \State{generate a realization $\phi$ from $G$}
    \State{obtain all reachable nodes of $s$ in $\phi$ by DFS and store them into $R_{\phi}(s)$}
    \State{record the immediate dominator of each node $u \in R_{\phi}(s)$ as $\textit{idom}[u]$ and construct the DT roots at $s$ of $\phi$}
    \State{$C_{\phi}^s\leftarrow \emptyset$, $M[\cdot]\leftarrow 0$}
    \State{\textbf{for each} $v \in S$ \textbf{do} $M[v]\leftarrow 1$}
    \ForEach{$u \in R_{\phi}(s)$ with the order of DFS traversal from $s$}
    {
        \If{$M[u]=0$}
        {
            \If{$M[\textit{idom}[u]]=0$}
            {
                \State{$C_{\phi}(s,u)\leftarrow C_{\phi}(s,u) \cup C_{\phi}(s,\textit{idom}[u]) \cup \{u\}$}
            }
            \Else
            {
                \State{$C_{\phi}(s,u)\leftarrow C_{\phi}(s,u) \cup \{u\}$}
            }
            \State{$C_{\phi}^s \leftarrow C_{\phi}^s \cup C_{\phi}(s,u)$}
        }
        
    }
    \State{\textbf{return} $C_{\phi}^s$}
}
\end{algorithm}

Based on the above analysis, to accurately estimate $D^L_s(B)$, we need to generate sufficient CP sequences, which consist of numerous CP sets.
A straightforward method to generate a CP set is first to find all paths between two nodes and then identify the common nodes on these paths. 
However, finding all paths between two nodes is time-consuming. 
To address this issue, we devise a scalable implementation to generate a CP set in polynomial time based on the DT (Definition \ref{def:DT} in Section \ref{pre:revisit}). 
Given the source node $s$ and a node $v \in (V\backslash S)$, the construction of the CP set of $v$ is as follows. 

\begin{itemize}[leftmargin=*]
    \item Generate a realization $\phi$ from $G$. 
    \item Construct the DT of $\phi$ by Lengauer-Tarjan algorithm~\cite{DBLP:journals/toplas/LengauerT79}, obtain all the nodes on the path from $s$ to $v$ (exclude $s$) in DT and store them into $C_{\phi}(s,v)$.
\end{itemize}

We can observe that the time complexity of generating a CP set is the same as that of the Lengauer-Tarjan algorithm~\cite{DBLP:journals/toplas/LengauerT79}, which is $\mathcal{O}(m \cdot\alpha(m,n))$, $\alpha$ is the inverse function of Ackerman’s function~\cite{ackermann1928hilbertschen}.
Since DT is a tree, each node in DT has only one incoming neighbor. 
Based on this property, we further propose an efficient algorithm for constructing a CP sequence, whose details are shown in Algorithm \ref{alg:LS}.
We first generate a random realization $\phi$ from $G$ (Line 1).
Then we obtain all reachable nodes of $s$ in $\phi$ by applying DFS, and store them into $R_{\phi}(s)$ (Line 2).
By applying Lengauer-Tarjan algorithm~\cite{DBLP:journals/toplas/LengauerT79}, the immediate dominator of each node $u \in R_{\phi}(s)$ is recorded as \textit{idom}$[u]$ and we construct the corresponding DT roots at $s$ of $\phi$ (Line 3). 
In Line 4, we initialize $C_{\phi}^s$ as $\emptyset$ to store the CP sequence and $M[\cdot]$ as 0. 
If $v\in S$, we set $M[v]=1$ (Line 5) and we only construct the CP set for the nodes in $R_{\phi}(s)\backslash S$ (Line 7).
Note that, we prioritize constructing the CP set for nodes with the order of DFS traversal from $s$ (Line 6). 
In such a setting, for each node $u \in (R_{\phi}(s)\backslash S)$, the CP set of $\textit{idom}[u]$ will be generated earlier than that of $u$ since $\textit{idom}[u]$ will be traversed earlier.
When constructing the CP set for node $u$, if $M[\textit{idom}[u]]=0$ (\ie $\textit{idom}[u] \notin S$), $u$ can inherit the CP set of $\textit{idom}[u]$ (Lines 8-9).
Otherwise, the CP set for $u$ will only consist of $u$ itself (Lines 10-11). 
It is clear that in Lines 6-12, we only need to process a DFS traversal starting from $s$.
Therefore, the time complexity of Algorithm \ref{alg:LS} is $\mathcal{O}(m \cdot \alpha(m,n))$.

\begin{example}
Here is an example to illustrate the process of Local Sampling. 
As shown in Figure \ref{example_DT}(a), $v_0$ is the misinformation seed node. 
A realization obtained from Figure \ref{example_DT}(a) and the corresponding DT are shown in Figure \ref{example_DT}(b) and Figure \ref{example_DT}(c), respectively. 
Suppose the DFS order is $(v_1,v_3,v_6,v_5,v_2)$. 
Since the immediate dominator of $v_1$ and $v_3$ is both $v_0$, \ie $\textit{idom}[v_1]=\textit{idom}[v_3]=v_0$, the CP set of them only consists themselves, \ie $C_{\phi}(v_0,v_1)=\{v_1\}$ and $C_{\phi}(v_0,v_3)=~\{v_3\}$. 
Since $\textit{idom}[v_6] = \textit{idom}[v_5] = v_3$, we can obtain that $C_{\phi}(v_0,v_6)=\{v_3,v_6\}$ and~$C_{\phi}(v_0,v_5)=\{v_3,v_5\}$. 
Similarly, $C_{\phi}(v_0,v_2)=\{v_2\}$.
\end{example}

\subsection{\myblue{Upper Bounding Function Estimation}}\label{section:est_UB}

\begin{algorithm}[t]
{
    \SetVline
    \footnotesize
    \caption{Global Sampling}
    \label{alg:GS}
    \Input{The graph $G=(V,E)$, the seed set $S$ and the unified seed node $s$.}
    \Output{The random LRR set $L_{\phi}(v)$.}
    \State{generate a realization $\phi$ obtained from $G$ and a node $v$ is randomly selected from $V_s^{\prime}$ with the probability of $\frac{1}{|V_s^{\prime}|}$}
    \If{$v \in R_{\phi}(s)$}
    {
        \State{$L_{\phi}(v) \leftarrow$ the set of nodes that can reach $v$ in ${\phi}^s$}
    }
    \State{\textbf{else} $L_{\phi}(v) \leftarrow \emptyset$}
    
    \State{\textbf{return} $L_{\phi}(v)$}
}
\end{algorithm}

\myblue{To estimate the upper bounding function, in the following, we first propose the concept of LRR set.}
\begin{definition} [Local Reverse Reachable (LRR) Set]
Given a graph $G=(V, E)$, a seed set $S\subseteq V$, a unified seed node $s$, a node $v\in (V\backslash S)$ and a realization $\phi=(V({\phi}),E({\phi}))$, a Local Reverse Reachable (LRR) set of $v$, denoted by $L_{\phi}(v)$, is the set of nodes that can reach $v$ in ${\phi}^s$.
\end{definition}

By generating a certain number of random LRR sets $\mathbb{L}$, we can obtain that $|V'_s|\cdot \frac{Cov_{\mathbb{L}}(B)}{|\mathbb{L}|}$ is an unbiased estimate of $D_s^U(B)$, where~$V'_s$ denotes the set of nodes that can be reached by $s$ in $G$ (exclude $S$).
\begin{lemma}\label{lem:upper esti}
    Given a seed set $S\subseteq V$, a unified seed node $s$ and the set of random LRR sets $\mathbb{L}$, for any blocker set $B \subseteq (V \backslash S)$,
    \begin{align}
        D_s^U(B)=|V_s^{\prime}|\cdot \mathbb{E}[\frac{Cov_{\mathbb{L}}(B)}{|\mathbb{L}|}]\text{,}
    \end{align}
    where the expectation is taken over the random choices of $\mathbb{L}$, $Cov_{\mathbb{L}}(B)$ is the coverage of $B$ in $\mathbb{L}$,\ie $Cov_{\mathbb{L}}(B)=\sum_{L(v) \in \mathbb{L}}\min\{|B \cap L(v)|,1\}$.
\end{lemma}

\begin{table}[t]
  \caption{\myblue{Time cost of generating 10K LRR sequences on DBLP}}
  \label{tab:LRRsequence}
\begin{tabular}{|c|c|c|c|c|c|c|}
  \hline
  \textbf{$|S|$} & 10 & 20 & 30 & 40 & 50 \\ 
  \hline
  \textsf{time (s)} & 26438.5 & 48907.3 & 67038.2 & 83011.5 & 99088.0 \\   
  \hline
\end{tabular}
\end{table}
Based on the above lemma, we propose Global Sampling to enable an accurate estimation of $D_s^U(\cdot)$.
As shown in Algorithm \ref{alg:GS}, we first generate a realization $\phi$ and randomly select a 
node $v$ from ~$V'_s$ (Line 1).
If the selected node $v$ cannot be influenced by $s$ in $\phi$, which implies that there are no nodes that can alternative-protect $v$ within this realization, we set the LRR set to an empty set (Line 4).

\myblue{\myparagraph{Discussion}
Actually, the upper bounding function can be estimated in a similar manner to Local Sampling based on the concept of Local Reverse Reachable (LRR) sequence for $\phi$, denoted by $L_{\phi}^s$. $L_{\phi}^s$ is the sequence of the LRR sets of the nodes that can be reached by $s$ in $\phi$ (exclude $S$), \ie $L_{\phi}^s=\{L_{\phi}(v):v\in R_{\phi}(s)\backslash S \}$.
By generating a certain number of LRR sequences $\mathbb{L}^s$, $D_s^U(B)$ can be unbiasedly estimated via $\frac{Cov_{\mathbb{L}^s}(B)}{|\mathbb{L}^s|}$, where $Cov_{\mathbb{L}^s}(B)$ denotes the coverage of $B$ in $\mathbb{L}^s$, \ie $Cov_{\mathbb{L}^s}(B)=\sum_{L^s \in \mathbb{L}^s}\sum_{L(v) \in L^s}\min\{|B \cap L(v)|,1\}$.

However, since ${\phi}^s$ is not necessarily a tree structure, we cannot generate an LRR sequence with only one single DFS traversal, as we do when generating a CP sequence. 
Generally, we need to conduct $|R_{\phi}(s)\backslash S|$ DFS traversals for each LRR sequence generation, which is rather time-consuming, especially when a substantial number of users can receive the misinformation.
Thus, we propose the Global Sampling method to tackle the problem. 
In Table \ref{tab:LRRsequence}, we present the time cost of generating 10K LRR sequences on the DBLP network, which has more than one million edges (the dataset details can be seen in Section \ref{sec:exp}).
As observed, with the increase of the number of misinformation seed nodes, the time cost of LRR sequence generation grows.
In particular, when $|S|=50$, the time required to produce the samples even exceeds one day, which is practically infeasible.
In contrast, based on Global Sampling, our solution for maximizing the upper bounding function only needs 7693.55 seconds when $|S|=50$ on DBLP, and the sample size is larger than 10K to ensure the desired approximation guarantee.
}

\section{Bounding Functions Maximization}
\label{sec:boundmaxi}

Based on the proposed estimation methods in Section \ref{sec:boundesti}, we first design two algorithms for \myblue{bounding functions maximization} in Section \ref{section:LSBM} and Section \ref{section:GSBM}, \myblue{which aim to find a set of blockers to maximize the bounding functions with approximation guarantees.
We then propose a heuristic algorithm to solve IMIN in Section \ref{section:LHGA}.}

\begin{table}[t]
  \caption{The number of CP sequences generated by OPBM and LSBM on DBLP ($k=100$, $|S|=10$, $\beta=0.1$, $\delta=1/n$)}
  \label{tab:nocps}
\begin{tabular}{|c|c|c|c|c|c|c|}
  \hline
  \textbf{$\epsilon$} & 0.1 & 0.2 & 0.3 & 0.4 & 0.5 \\ 
  \hline
  \textsf{OPBM} & 2,871,296 & 1,435,648 & 717,824 & 358,912 & 179,456 \\   
  \hline
  \textsf{LSBM} & 22,544 & 11,272 & 5,636 & 2,818 & 2,818 \\
  \hline
\end{tabular}
\vspace{1.5mm}
\end{table}

\subsection{\myblue{Lower Bounding Function Maximization}}\label{section:LSBM}
\label{boundsmax:lower}

In general, to get a solution for lower \myblue{bounding function maximization}, we need to obtain the value of $D_s^L(\cdot)$.
As we stated in Section \ref{section:est_LB}, $Cov_{\mathbb{C}^s}(\cdot)/|\mathbb{C}^s|$ is an unbiased estimator of $D_s^L(\cdot)$.
That is, for estimating $D_s^L(\cdot)$ accurately, we need to generate a certain number of samples (\ie CP sequences).
However, the generation of numerous samples may lead to significant computation overhead for the algorithm.
Therefore, it is imperative for us to determine an appropriate sample size, so as to strike a balance between efficiency and accuracy.
A straightforward approach is to employ the state-of-the-art algorithm for influence maximization (\ie OPIM-C~\cite{DBLP:conf/sigmod/TangTXY18}). Specifically, OPIM-C first generates two independent collections of samples $\mathcal{R}_1$ and $\mathcal{R}_2$ and then generates a solution using the greedy algorithm on $\mathcal{R}_1$. Afterward, OPIM-C assesses whether the solution can meet the stopping criterion. 
If the criterion is met, the solution is returned; otherwise, the above-mentioned steps are repeated until the algorithm terminates. In particular, the stopping criterion of OPIM-C is $\mathcal{M}(S) \geq1-1/e-\epsilon$, where $S$ is the solution in the current round and $\mathcal{M}(\cdot)$ is the function determined by two martingale-based concentration bounds~\cite{DBLP:conf/sigmod/TangSX15} on $\mathcal{R}_1$ and $\mathcal{R}_2$.

\myblue{However, OPIM-C cannot be directly applied for our problem. 
The main reason is that OPIM-C relies on RIS, which is not suitable for the estimation of bounding function as stated before.
Besides, since only the nodes eligible to be reached by misinformation can contribute to the bounding function, the concentration bounds leveraged by OPIM-C are also not tailored to the unbiased estimator proposed in Section \ref{section:est_LB}, which results in the stopping criterion of OPIM-C being infeasible for our problem.}
This motivates us to develop an efficient algorithm for lower \myblue{bounding function maximization}. 
In this section, we first propose two martingale-based concentration bounds, based on which, we then derive a novel stopping criterion.
Combining this with Local Sampling, we propose \myblue{Local Sampling based Bounding Function Maximization (LSBM).}
In Table \ref{tab:nocps}, \myblue{we report the number of CP sequences generated by LSBM and the algorithm, which employs OPIM-C for our problem directly but with the proposed unbiased estimator of $D_s^L(\cdot)$ (termed as OPBM) 
on DBLP.}
As shown, in most cases, the number of samples generated by OPBM is more than 100x larger than that of LSBM.
In addition, on the larger datasets, OPBM cannot even finish due to the memory overflow. 
In the following, we first propose two novel concentration bounds based on the concept of \textit{martingale}~\cite{DBLP:journals/im/ChungL06}.
\begin{lemma}[Concentration Bounds]\label{new bound}
    Given a blocker set $B$, a seed node $s$ and a set of $\theta$ random CP sequences $\mathbb{C}^s$. For any $\lambda >0$,
    \begin{align}
        &\Pr[\frac{Cov_{\mathbb{C}^s}(B)}{\mathbb{E}[I_G(s)]}-\frac{D^L_s(B)\cdot \theta}{\mathbb{E}[I_G(s)]} \geq \lambda] \leq \exp({-\frac{\lambda^2}{\frac{2D_s^L(B)}{\mathbb{E}[I_G(s)]}\cdot \theta+\frac{2}{3}\lambda}}),\label{new bound1}\\
        &\Pr[\frac{Cov_{\mathbb{C}^s}(B)}{\mathbb{E}[I_G(s)]}-\frac{D^L_s(B)\cdot \theta}{\mathbb{E}[I_G(s)]} \leq -\lambda] \leq \exp({-\frac{\lambda^2}{\frac{2D^L_s(B)}{\mathbb{E}[I_G(s)]}\cdot \theta}}).\label{new bound2}
    \end{align}
\end{lemma}

\begin{algorithm}[t]
{
    \SetVline
    \footnotesize
    \caption{\myblue{LSBM}}
    \label{alg:LSBM}
    \Input{The graph $G=(V,E)$, the seed set $S$, the unified seed node $s$, the budget $k$ and parameter $\beta,\epsilon,\delta$.}
    \Output{The blocker set $B_L$ with $(1-1/e-\epsilon)$-approximation.}
    \State{$\textit{ON} \leftarrow $ the set of nodes that can be directly activated by $S$}
    \If{$|\textit{ON}|\leq k$}
    {
        \State{\textbf{return} $B_L=\textit{ON}$}
    }
    \State{$\hat{I}_G(s)\leftarrow$ the estimated value of $\mathbb{E}[I_G(s)]$ with $(\beta,\frac{\delta}{6})$-approximation}
    \State{$\text{OPT}^L\leftarrow$ the lower bound of $D_s^L(B^o_L)$}
    \State{$\theta_{\max}\leftarrow \frac{2\hat{I}_G(s)\left( (1-1/e)\sqrt{\ln\frac{12}{\delta}}+\sqrt{(1-1/e)(\ln\binom{n-|S|}{k}+\ln\frac{12}{\delta})} \right)^2}{(1-\beta)\epsilon^2 \text{OPT}^L}$}
    \State{$\theta_0\leftarrow \theta_{\max}\cdot(1-\beta)\epsilon^2 \text{OPT}^L/\hat{I}_G(s)$}
    \State{$i_{\max}\leftarrow \lceil \log_2\frac{\theta_{\max}}{\theta_0} \rceil$}
    \State{generate two sets $\mathbb{C}_1^s$, $\mathbb{C}_2^s$ of $\theta_0$ random CP sequences, respectively}
    \State{$a_1\leftarrow \ln\frac{3i_{max}}{\delta},a_2\leftarrow \ln\frac{3i_{max}}{\delta}$}
    \For{$i\leftarrow1$ to $i_{\max}$}
    {

        \State{$B_L\leftarrow \text{Max-Coverage}(\mathbb{C}_1^s,k)$}
        \State{$\sigma^L(B_L) \leftarrow 0$}
        \If{$Cov_{\mathbb{C}^s_2}(B_L)\cdot(1-\beta)/\hat{I}_G(s)\geq 5a_1/18$}
        {
            \State{$\sigma^L(B_L)\leftarrow\left((\sqrt{\frac{Cov_{\mathbb{C}^s_2}(B_L)\cdot(1-\beta)}{\hat{I}_G(s)}+\frac{2a_1}{9}}-\sqrt{\frac{a_1}{2}})^2-\frac{a_1}{18}\right)\cdot\frac{1}{|\mathbb{C}_2^s|}$}
        }
        \ElseIf{$Cov_{\mathbb{C}^s_2}(B_L)\cdot(1+\beta)/\hat{I}_G(s)\leq 5a_1/18$}
        {
            \State{$\sigma^L(B_L)\leftarrow\left((\sqrt{\frac{Cov_{\mathbb{C}^s_2}(B_L)\cdot(1+\beta)}{\hat{I}_G(s)}+\frac{2a_1}{9}}-\sqrt{\frac{a_1}{2}})^2-\frac{a_1}{18}\right)\cdot\frac{1}{|\mathbb{C}_2^s|}$}
        }
        \State{$Cov_{\mathbb{C}^s_1}^u(B_L^o) \leftarrow \underset{0\leq i \leq k}{\min}\left(Cov_{\mathbb{C}^s_1}(B_i) + \underset{v \in maxMC(B_i,k)}{\sum}Cov_{\mathbb{C}^s_1}(v~|~B_i)\right)$}
        \State{$\sigma^U(B_L^o)\leftarrow\left(\sqrt{\frac{Cov_{\mathbb{C}^s_1}^u(B_L^o)\cdot(1+\beta)}{\hat{I}_G(s)}+\frac{a_2}{2}}+\sqrt{\frac{a_2}{2}}\right)^2\cdot\frac{1}{|\mathbb{C}_1^s|}$}
        \If{$\sigma^L(B_L)/ \sigma^U(B_L^o)\geq 1-1/e-\epsilon$~\textbf{or}~$i=i_{\max}$}
        {
            \State{\textbf{return} $B_L$}
        }
        \State{double the sizes of $\mathbb{C}_1^s$ and $\mathbb{C}_2^s$ with new CP sequences}
    }
}

\footnotesize
\vspace{2mm}
		\State{\textbf{Procedure} 
  {Max-Coverage}$(\mathbb{C}^s,k)$}
            \State{$B\leftarrow \emptyset$}
    \For{$i\leftarrow1$ to $k$}
    {
        \State{$u \leftarrow \arg\max_{v \in (V\backslash S)}(Cov_{\mathbb{C}^s}(B\cup\{v\})-Cov_{\mathbb{C}^s}(B))$}
        \State{$B\leftarrow B \cup \{u\}$}
    }
    \State{\textbf{return} $B$}
\end{algorithm}


\myblue{\myparagraph{LSBM algorithm}}
\myblue{Based on these concentration bounds, we devise a scalable implementation called LSBM for lower bounding function maximization}.
The pseudocode of LSBM is shown in Algorithm~\ref{alg:LSBM}.
Let \textit{ON} be the set of outgoing neighbors of $S$ (Line 1).
When the budget $k$ is no less than the number of nodes in \textit{ON}, LSBM directly returns \textit{ON} as the blocker set (Lines 2-3). 
Then we calculate $\hat{I}_G(s)$ with $(\beta,\frac{\delta}{6})$-approximation \myblue{by employing the generalized stopping rule algorithm introduced in~\cite{DBLP:journals/tkde/ZhuTTWL23}}. In addition, we derive the lower bound of $D_s^L(B^o_L)$ and define the constants $\theta_{\max}$ and $\theta_{0}$ (Lines 4-7). Afterwards, we generate two sets of CP sequences $\mathbb{C}_1^s$ and $\mathbb{C}_2^s$, each of size $\theta_0$ (Line 9). 
The subsequent part of the algorithm consists of at most $i_{max}$ iterations. 
In each iteration, we first invoke Procedure Max-Coverage to get a blocker set $B_L$, \ie finding a set of $k$ nodes such that $B_L$ intersects with as many CP sequences as possible in $\mathbb{C}_1^s$ (Line 12).
Then we derive $\sigma^L(B_L)$ and $\sigma^U(B_L^o)$ from $\mathbb{C}_2^s$ and $\mathbb{C}_1^s$, respectively (Lines 14-19). 
Specifically, $\sigma^L(B_L)$ is the lower bound of $D_s^L(B_L)/\mathbb{E}[I_G(s)]$ and $\sigma^U(B^o_L)$ is the upper bound of $D_s^L(B_L^o)/\mathbb{E}[I_G(s)]$.
If $\sigma^L(B_L)/ \sigma^U(B_L^o)\geq 1-1/e-\epsilon$ or $i=i_{\max}$, LSBM returns $B_L$ and terminates. 
Otherwise, the quantities of CP sequences in $\mathbb{C}_1^s$ and $\mathbb{C}_2^s$ will be doubled, and LSBM will proceed to the next iteration (Lines 20-22). Next, we will explain in detail how LSBM can return a solution with an approximation guarantee.
\myblue{Note that, when deriving $\theta_{\max}$, $\sigma^L(B_L)$ and $\sigma^U(B_L^o)$, it is imperative to carefully consider the error associated with estimating $\mathbb{E}[I_G(s)]$, to ensure the desired approximation guarantee.}

\noindent\underline{\textit{Deriving $\theta_{\max}$}}. Based on previous concentration bounds~\cite{DBLP:conf/sigmod/TangSX15}, Tang \etal derive an upper bound on the sample size required to ensure $(1-1/e-\epsilon)$-approximation holds with probability at least $1-\delta$ for IM. Similarly, we derive the corresponding upper bound on the sample size for the IMIN problem based on our proposed novel concentration bounds. 
The following lemma provides the setting of $\theta_{max}$, ensuring the correctness of LSBM when $i=i_{\max}$.
\begin{lemma}\label{lemma:imax}
    Let $\mathbb{C}^s$ be a set of random CP sequences, $B_L$ be a size-$k$ blocker set generated by applying Max-Coverage on $\mathbb{C}^s$, $B^o_L$ be the optimal solution with size-$k$, $OPT^L$ be the lower bound of $D_s^L(B^o_L)$ and $\hat{I}_G(s)$ be the estimated value of $\mathbb{E}[I_G(s)]$ with $(\beta,\frac{\delta}{6})$-approximation. For fixed $\beta,\epsilon$ and $\delta$, let 
    \begin{align*}
        \theta_{max}= \frac{2\hat{I}_G(s)\left( (1-1/e)\sqrt{\ln\frac{12}{\delta}}+\sqrt{(1-1/e)(\ln\binom{n-|S|}{k}+\ln\frac{12}{\delta})} \right)^2}{(1-\beta)\epsilon^2 OPT^L}\text{,}
    \end{align*}
    if $|\mathbb{C}^s|=\theta \geq \theta_{\max}$, then $B_L$ is $(1-1/e-\epsilon)$-approximate solution with at least $1-\delta/3$ probability.
\end{lemma}

\noindent\underline{\textit{Deriving $OPT^{L}$}}. Due to the different properties of the objective functions of IMIN and IM, we cannot directly replace $\text{OPT}^L$ with $k$ as OPIM-C~\cite{DBLP:conf/sigmod/TangTXY18}. To address this issue, we set $\text{OPT}^L=\sum_{v \in B^*}\Pr[s\rightarrow v]$, where $B^*$ denote the set of $k$ nodes of $\textit{ON}$ with the $k$ largest probability of being activated by $s$ and $\Pr[s\rightarrow v]$ is the probability that $s$ can activate $v$. 
Since when we select a node $v$ as a blocker, $D_s^L(\{v\}\cup B)- D_s^L(B)\geq \Pr[s\rightarrow v]$. 
The reason why we set $\textit{ON}$ as the candidate set is that the value of $\Pr[s\rightarrow v]$ of each node $v \in \textit{ON}$ can be computed efficiently, otherwise it will become the bottleneck of the whole algorithm.  

\noindent\underline{\textit{Deriving $\sigma^L(B_L)$ and $\sigma^U(B_L^o)$}}. 
Next, we derive the lower bound $\sigma^L(B_L)$ of $\frac{D_s^L(B_L)}{\mathbb{E}[I_G(s)]}$ and the upper bound $\sigma^U(B^o_L)$ of $\frac{D_s^L(B_L^o)}{\mathbb{E}[I_G(s)]}$ such that the approximation ratio $\frac{D_s^L(B_L)}{D_s^L(B_L^o)}\geq \frac{\sigma^L(B_L)}{\sigma^U(B_L^o)}$.
\begin{lemma}\label{lemma:bounds}
For any $0 \leq \beta,\epsilon,\delta \leq 1$, we have
\begin{align*}
    &\Pr[\sigma^L(B_L)\leq \frac{D_s^L(B_L)}{\mathbb{E}[I_G(s)]}]\geq 1-\frac{\delta}{3i_{\max}}\text{,}\\
    &\Pr[\sigma^U(B^o_L)\geq \frac{D_s^L(B_L^o)}{\mathbb{E}[I_G(s)]}]\geq 1-\frac{\delta}{3i_{\max}}.
\end{align*}
\end{lemma}

\myparagraph{Putting together} The reason that LSBM ensures $(1-1/e-\epsilon)$-approximation with at least 
$1-\delta$ probability can be explained as follows. First, the algorithm has at most $i_{\max}$ iterations. In each of the first $i_{\max}-1$ iterations, a blocker set $B_L$ is generated and we derive an approximation guarantee $\sigma^L(B_L)/ \sigma^U(B_L^o)$ that is incorrect with at most $2\delta / (3i_{\max})$ probability (Lemma \ref{lemma:bounds}). By the union bound, LSBM has at most $2\delta / 3$ to return an incorrect solution in the first $i_{\max}-1$ iterations. Meanwhile, in the last iteration, a blocker set $B_L$ obtained by applying Procedure Max-Coverage on $\mathbb{C}^s_1$, with $|\mathbb{C}^s_1|\geq \theta_{\max}$. This ensures that $B_L$ is an $(1-1/e-\epsilon)$-approximation with at least $1-\delta/3$ probability when $i=i_{\max}$ (Lemma \ref{lemma:imax}). Therefore, the probability that LSBM returns an incorrect solution in any iteration is at most $\delta$, leading to the following theorem.

\begin{theorem}\label{tg1}
    Given $0 \leq \beta,\epsilon,\delta \leq 1$, $B^o_L$ is the optimal solution of lower \myblue{bounding function} maximization, LSBM returns $B_L$ satisfies:
    \begin{align}
        \Pr[D_s^L(B_L) \geq (1-1/e-\epsilon)D_s^L(B^o_L)]\geq 1-\delta\label{ratio1}.
    \end{align}
\end{theorem}
Moreover, we have the following theorem to guarantee the expected time complexity of LSBM. 

\begin{theorem}\label{theorem:time complexity_LSBM}
    When $\frac{2\beta}{1+\beta}\leq \epsilon$ and $\delta \leq 1/2$, LSBM runs in $\mathcal{O}(\frac{(k\ln{(n-|S|)}+\ln{1/\delta})\mathbb{E}[I_G(s)](m\cdot \alpha(m,n)+ D_s^L(B_L^o))}{(\epsilon+\epsilon\beta-2\beta)^2 D_s^L(B_L^o)}+\frac{m\cdot\ln{1/\delta}}{\beta^2})$ expected time under the IC model.
\end{theorem}

\subsection{\myblue{Upper Bounding Function Maximization}}\label{section:GSBM}
\label{boundsmax:upper}

\myblue{In what follows, we propose Global Sampling based Bounding Function Maximization (GSBM) for upper bounding function maximization}.
GSBM is similar to the framework of LSBM and the differences between them are that we use Global Sampling to estimate $D_s^U(\cdot)$ with a certain number of random LRR sets and we set:
\begin{align*}
    &\theta_{\max}=\frac{2|V_s^{\prime}|\left( (1-1/e)\sqrt{\ln\frac{6}{\delta}}+\sqrt{(1-1/e)(\ln\binom{|V_s^{\prime}|-|S|}{k}+\ln\frac{6}{\delta})} \right)^2}{\epsilon^2 \text{OPT}^L}\text{,}\\
    &\theta_0=\theta_{\max}\cdot\epsilon^2 \text{OPT}^L/|V_s^{\prime}|\text{,}\\
    &\sigma^L(B_U)=\left((\sqrt{Cov_{\mathbb{L}_2}(B_U)+\frac{2a_1}{9}}-\sqrt{\frac{a_1}{2}})^2-\frac{a_1}{18}\right)\cdot\frac{|V_s^{\prime}|}{|\mathbb{L}_2|}\text{,}\\   &\sigma^U(B_U^o)\leftarrow\left(\sqrt{Cov_{\mathbb{L}_1}^u(B_U^o)+\frac{a_2}{2}}+\sqrt{\frac{a_2}{2}}\right)^2\cdot\frac{|V_s^{\prime}|}{|\mathbb{L}_1|}.
\end{align*}
Similar to the proof of Theorem \ref{tg1}, we also show GSBM can return a solution with $(1-1/e-\epsilon)$-approximation.
\begin{theorem}\label{tg2}
    Given $0 \leq \epsilon,\delta \leq 1$, $B^o_U$ is the optimal solution of upper \myblue{bounding function} maximization, GSBM returns $B_U$ satisfies:
    \begin{align}
        \Pr[D_s^U(B_U) \geq (1-1/e-\epsilon)D_s^U(B^o_U)]\geq 1-\delta\label{ratio2}.
    \end{align}
\end{theorem}
Besides, the time complexity of GSBM is shown in Theorem \ref{theorem:time complexity_GSBM}.
\begin{theorem}\label{theorem:time complexity_GSBM}
    When $\delta \leq 1/2$, the time complexity of GSBM under the IC model is $\mathcal{O}(\frac{(k\ln{(n-|S|)}+\ln({1/\delta}))(|V_s^{\prime}|+m)}{\epsilon^2})$.
    
\end{theorem}

\subsection{A Lightweight Heuristic for IMIN}\label{section:LHGA}
\label{boundsmax:original}
Here we consider the filling of the Sandwich, \ie the solution for the original problem IMIN.
To address this problem, we propose a Lightweight Heuristic algorithm that adopts the greedy framework (LHGA) combining the following proposed function to evaluate the quality score of each node:
\begin{align}
    s(v)=\Pr[s \rightarrow v]\cdot deg[v]\text{.}
\end{align}
Specifically, we iteratively select the node 
from $ON$ with the largest quality score $s(\cdot)$.
The motivations of our proposed function are that $i)$ the outgoing neighbors of the misinformation seeds are more likely to be blockers; 
$ii)$ the nodes with a relatively high probability of being activated by $s$ are more likely to be blockers; 
$iii)$ the nodes with large influence are more likely to be blockers.
If the budget is no less than the number of nodes in $ON$, we directly return $ON$ as the blocker set.
\myblue{Although LHGA does not make any theoretical contribution to SandIMIN as shown in Eq. (\ref{eq:sandwich}), it can enhance the effectiveness of our algorithm without sacrificing efficiency.
More details can be seen in Section \ref{sec:exp}.}

\myparagraph{Summary}
Based on the above results, \ie $B_L$ returned by LSBM, $B_U$ returned by GSBM and $B_R$ returned by LHGA, \sand returns the blocker set $B^* \in \{B_L,B_U,B_R\}$ with the smallest $(\gamma,\delta)$-estimated value of $\mathbb{E}[I_{G[V\backslash B^*]}(s)]$. According to the Theorem \ref{tg1} and \ref{tg2}, the constant $\alpha_1$ and $\alpha_2$ in Eq. (\ref{eq:sandwich}) are both set to $(1-1/e-\epsilon)$, by union bound, the blocker set $B$ produced by \sand has the following theoretical guarantees with at least $1-3\delta$ probability,
\begin{align}\label{data-ratio}
    \scalebox{1.0537}{$D_s(B) \geq \max \left\{\frac{D_s(B_U)}{D_s^U(B_U)},\frac{D_s^L(B_L^o)}{D_s(B^o)}\right\}(1-\frac{1}{e}-\epsilon)\frac{1-\gamma}{1+\gamma}D_s(B^o).$}
\end{align}


\section{Experiments}\label{sec:exp}
In this section, we conduct extensive experiments on 9 real-world datasets to evaluate the performance of our algorithms. 


\begin{table}[t]
  \begin{center}
    \caption{Statistics of datasets}\label{data}
    \label{expdatasets}
\setlength\tabcolsep{2pt}
    \begin{tabular}{|l|c|c|c|c|}
    \hline
      \textbf{Dataset} & \textbf{Type} & \textbf{$|V|$} & \textbf{$|E|$} & \textbf{Avg. deg} \\
        \hline
        \hline
	   EmailCore (EC) & directed & 1,005 & 25,571 & 49.6 \\
        \hline
	   Facebook (FB) & undirected & 4,039 & 88,234 & 43.7 \\
 	\hline
	   Wiki-Vote (WV) & directed & 7,115 & 103,689 & 29.1 \\
        \hline
	   EmailAll (EA) & directed & 265,214 & 420,045 & 3.2 \\
        \hline
          DBLP (DB) & undirected & 317,080 & 1,049,866 & 6.6 \\
	\hline
          Twitter (TT) & directed & 81,306 & 1,768,149 & 59.5 \\
        \hline
          Stanford (SF) & directed & 281,903 & 2,312,497 & 16.4 \\
        \hline
          Youtube (YT) & undirected & 1,134,890 & 2,987,624 & 5.3 \\
        \hline
          Pokec (PK) & directed & 1,632,803 & 30,622,564 & 37.5 \\
	\hline
    \end{tabular}
  \end{center}
\end{table}


\begin{figure}
	\centering
	\includegraphics[width=0.97\linewidth]{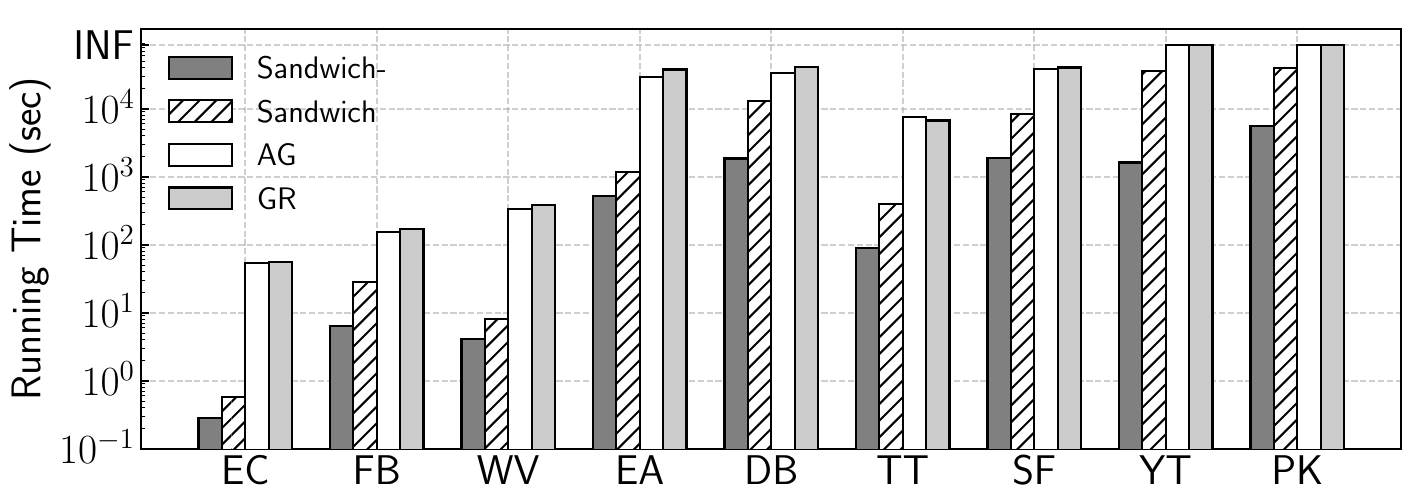}
	\caption{\myblue{Time cost on all the datasets}}
	\label{fig:time_line}
\end{figure}

\begin{figure*}[t]
\makebox[\textwidth][c]{
\includegraphics[scale=0.7]{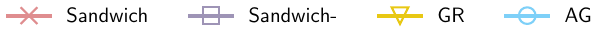}
}
\begin{minipage}{0.195\textwidth}
        \centering
        \includegraphics[width=\textwidth]{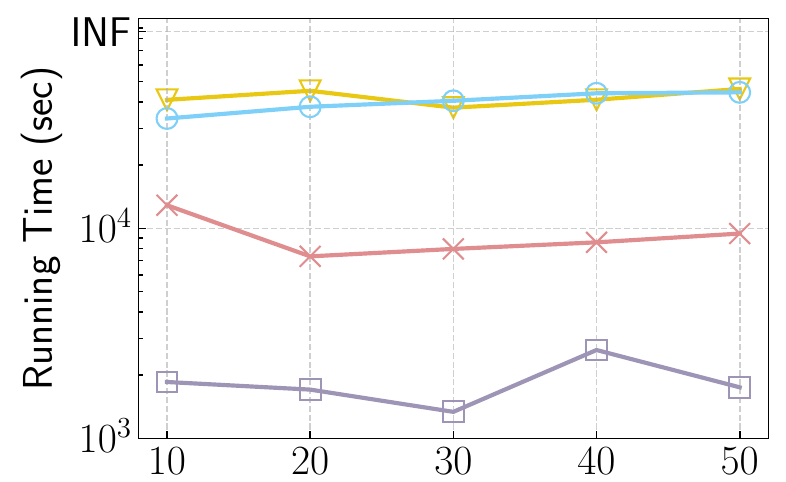}
        \subcaption{\myblue{DBLP (Vary $|S|$)}}
    \end{minipage} 
    \begin{minipage}{0.195\textwidth}
        \centering
        \includegraphics[width=\textwidth]{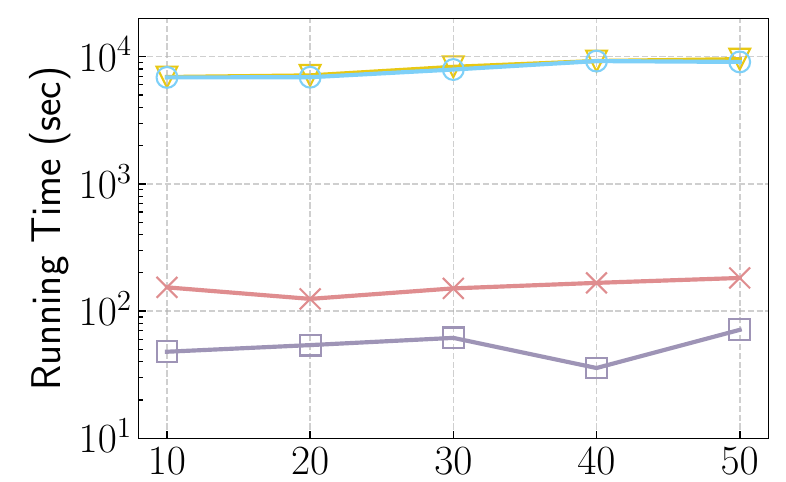}
        \subcaption{\myblue{Twitter (Vary $|S|$)}}
    \end{minipage} 
    \begin{minipage}{0.195\textwidth}
        \centering
        \begin{subfigure}{\textwidth}
            \includegraphics[width=\textwidth]{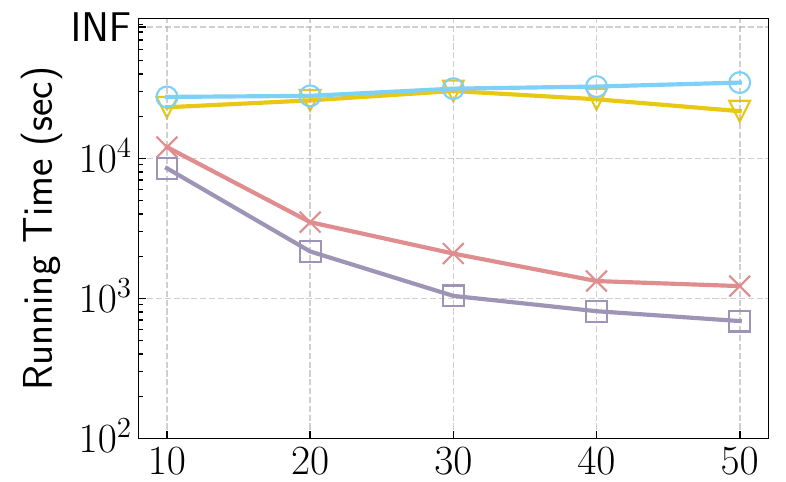}
            \subcaption{\myblue{Stanford (Vary $|S|$)}}
        \end{subfigure}
    \end{minipage}
    \begin{minipage}{0.195\textwidth}
        \centering
        \begin{subfigure}{\textwidth}
            \includegraphics[width=\textwidth]{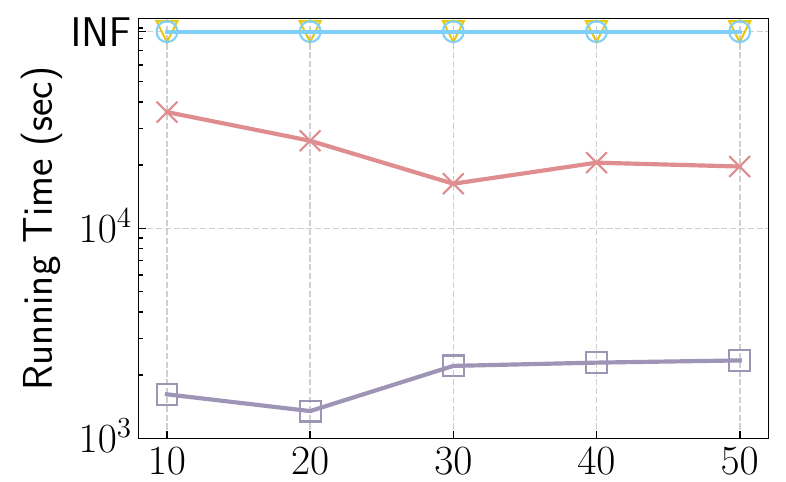}
            \subcaption{\myblue{Youtube (Vary $|S|$)}}
        \end{subfigure}
    \end{minipage}
    \begin{minipage}{0.195\textwidth}
        \centering
        \begin{subfigure}{\textwidth}
            \includegraphics[width=\textwidth]{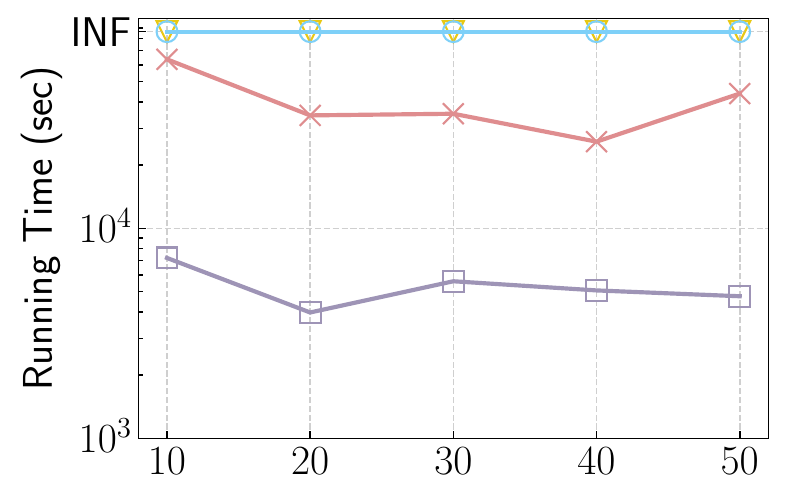}
            \subcaption{\myblue{Pokec (Vary $|S|$)}}
        \end{subfigure}
    \end{minipage}
    \begin{minipage}{0.195\textwidth}
        \centering
        \includegraphics[width=\textwidth]{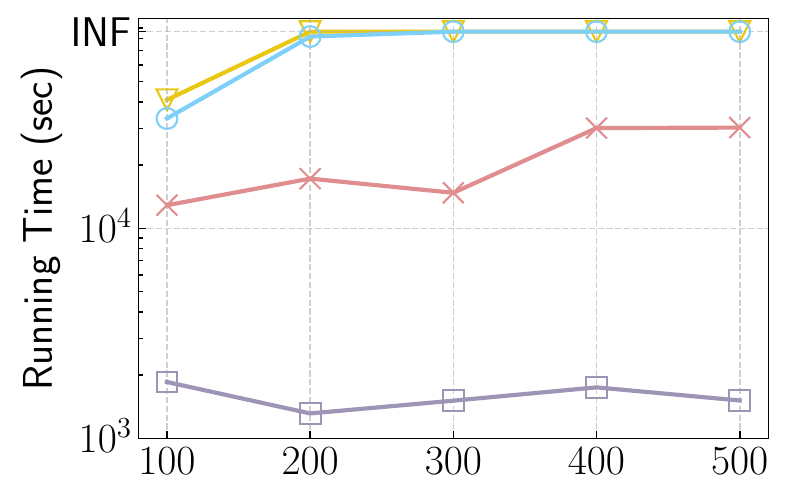}
        \subcaption{\myblue{DBLP (Vary $k$)}}
    \end{minipage} 
    \begin{minipage}{0.195\textwidth}
        \centering
        \includegraphics[width=\textwidth]{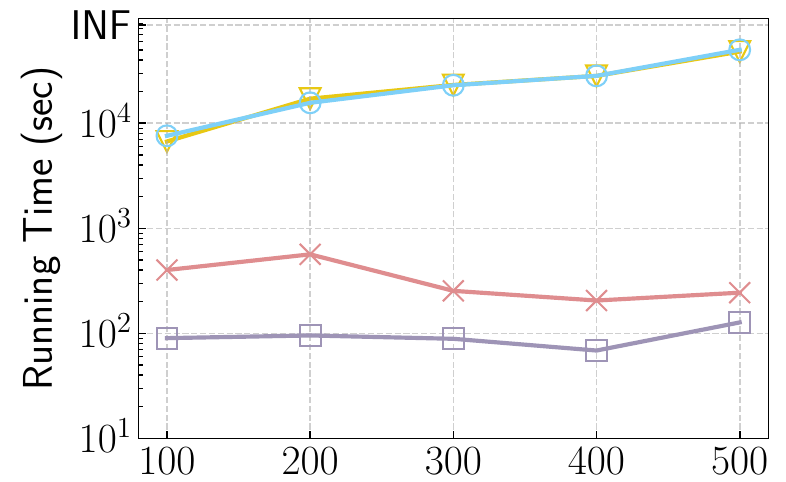}
        \subcaption{\myblue{Twitter (Vary $k$)}}
    \end{minipage} 
    \begin{minipage}{0.195\textwidth}
        \centering
        \begin{subfigure}{\textwidth}
            \includegraphics[width=\textwidth]{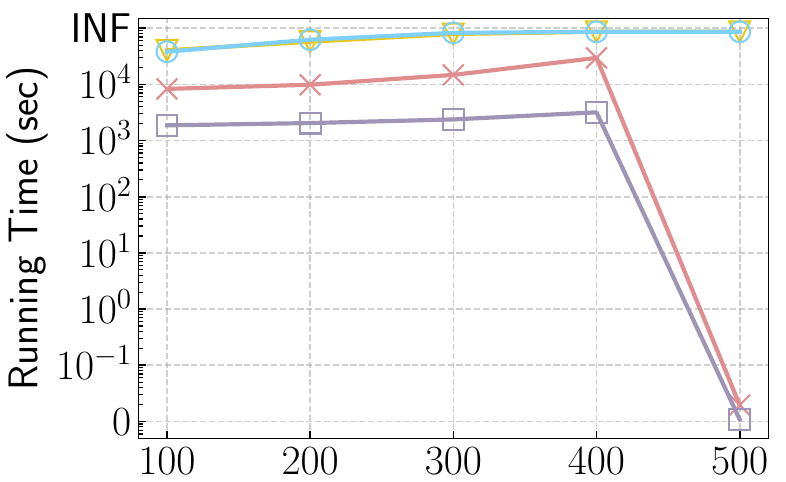}
            \subcaption{\myblue{Stanford (Vary $k$)}}
        \end{subfigure}
    \end{minipage}
    \begin{minipage}{0.195\textwidth}
        \centering
        \begin{subfigure}{\textwidth}
            \includegraphics[width=\textwidth]{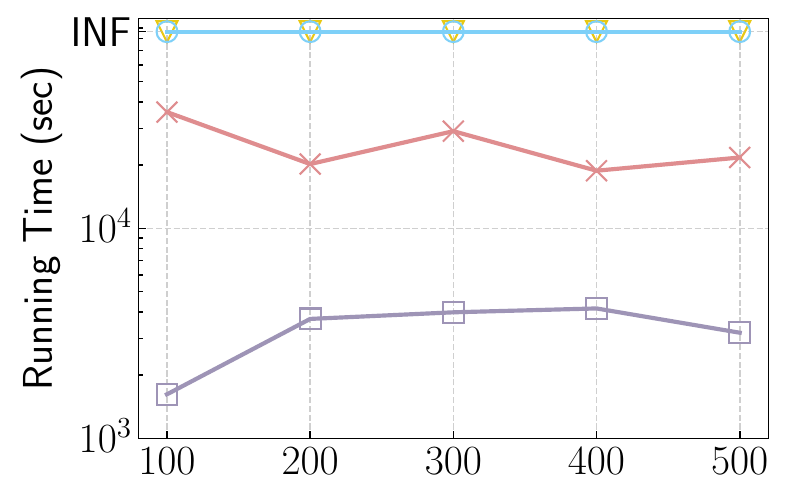}
            \subcaption{\myblue{Youtube (Vary $k$)}}
        \end{subfigure}
    \end{minipage}
    \begin{minipage}{0.195\textwidth}
        \centering
        \begin{subfigure}{\textwidth}
            \includegraphics[width=\textwidth]{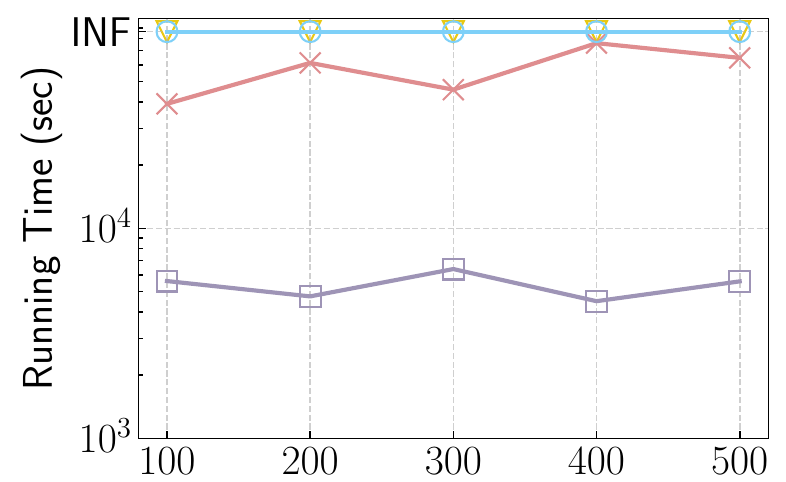}
            \subcaption{\myblue{Pokec (Vary $k$)}}
        \end{subfigure}
    \end{minipage}
	\caption{\myblue{Efficiency evaluation by varying $|S|$ and $k$}}
	\label{fig:time_vary_parameter}
\end{figure*}

\begin{figure}[t]
	\centering
	\includegraphics[width=0.97\linewidth]{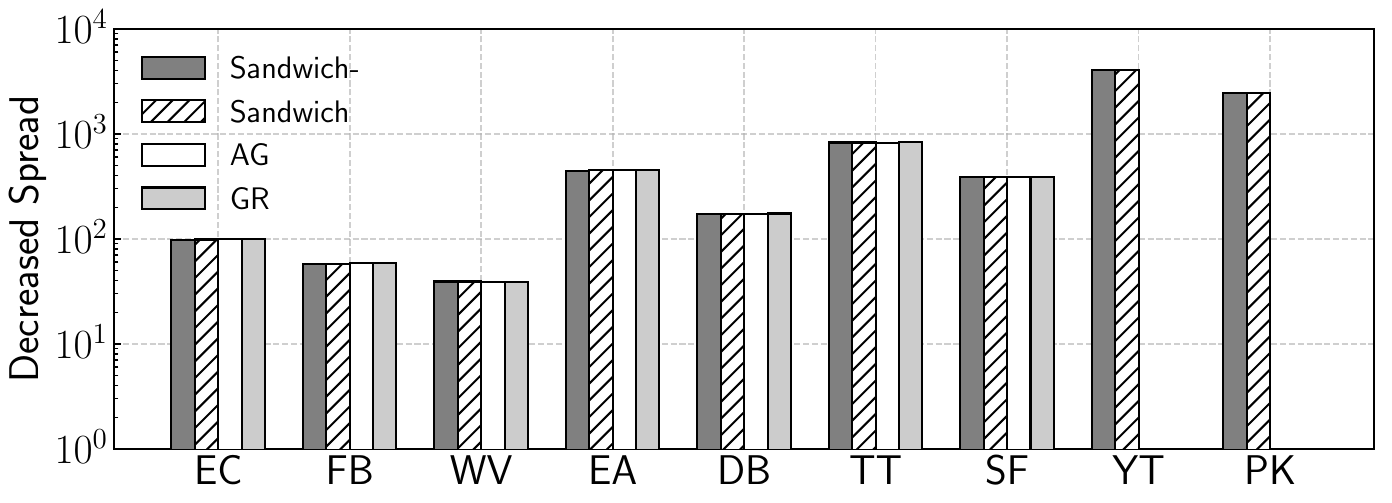}
	\caption{\myblue{Decreased spread on all the datasets}}
	\label{fig:inf_line}
\end{figure}

\myparagraph{Algorithms}
In the experiment, we implement and evaluate the following algorithms. 
$i)$ \myblue{\textbf{AG/GR}: the state-of-the-art algorithms proposed in \cite{DBLP:conf/icde/0002ZW0023}.}
$ii)$ \textbf{\sand}: the Sandwich framework based approach proposed in this paper with tight theoretical guarantee.
Note that, both AG and GR are heuristic solutions without theoretical guarantees about the final result. Therefore, in the experiment, we also implement $iii)$ \textbf{\sand-}, which relaxes the theoretical result of \sand by
setting $\alpha_1=1-1/e-\epsilon$, $\alpha_2=0$ in Algorithm \ref{alg:Sandwich}, i.e., without GSBM. It can provide better performance in efficiency and competitive quality of results.

\myparagraph{Datasets} We use 9 real datasets which are available on SNAP\footnote{http://snap.stanford.edu} in our experiments. The details of the datasets are reported in Table \ref{expdatasets} and we provide a short name for each dataset.

\myparagraph{Parameter settings} 
Following the convention~\cite{DBLP:conf/sigmod/TangXS14,DBLP:conf/sigmod/TangSX15,DBLP:conf/sigmod/NguyenTD16,DBLP:journals/pvldb/HuangWBXL17,DBLP:conf/sigmod/TangTXY18}, we set the propagation probability $p(u, v)$ of each edge $\langle u, v \rangle$ as the inverse of $v$'s in-degree in IC model. 
By default, we set $\epsilon=0.2$, $\beta=\gamma=0.1$ and $\delta=1/n$ for \sand and \sandl. 
For AG and GR, we set the number of generated realizations to $10^4$, which is recommended in \cite{DBLP:conf/icde/0002ZW0023}. 
$|S|$ and $k$ are 10 and 100 by default, respectively. 
In addition, the nodes in misinformation seed set are randomly selected from the top 200 most influential nodes. 
Finally, we estimate the expected spread of the seed set by taking the average of its spreads over $10^5$ Monte-Carlo simulations. For each parameter setting, we repeat each algorithm 10 times and report the average value.
For those experiments that cannot finish within 24 hours, we set them as \textbf{INF}.
All the programs are implemented in  C++ and performed on a PC with an Intel Xeon 2.10GHz CPU and 512GB memory.

\subsection{Efficiency Evaluation}

\myparagraph{Results on all the datasets} In Figure \ref{fig:time_line}, we first evaluate the time cost on all the datasets with the default settings. 
\myblue{As can be seen, \sand and \sandl cost less time than AG and GR on all the datasets and they can achieve up to two orders of magnitude speedup. 
In most cases, AG is more efficient than GR. 
Besides, AG and GR cannot complete on large datasets (\ie Youtube and Pokec) in a reasonable time.
The primary reasons are that $i)$ our methods transform the original non-submodular maximization problem into the submodular maximization scenario. This allows us to avoid generating new samples after selecting each blocker. Furthermore, our heuristic algorithm for IMIN incurs almost no computation overhead. $ii)$ Building upon the novel martingale-based concentration bounds, the sample size of our methods can be significantly reduced. 
As shown in Table \ref{tab:nocps}, on DBLP with $\epsilon=0.3$, the sample size generated by LSBM is only 5636, even smaller than the sample size produced by AG and GR in each round.}
In addition, \sandl is more efficient than \sand on all the datasets, which is not surprising given that \sandl does not execute GSBM.
\myblue{Besides, as the dataset becomes large, the gap between the time cost of our solutions and that of AG and GR becomes smaller.}
This is because, to provide theoretical guarantees for our solutions, the number of samples generated will increase as the dataset becomes large, \myblue{while AG and GR use a fixed number of realizations each round for graphs of any size and offer no approximation guarantee for the quality of the returned results.}

\begin{figure*}[t]
\small
\makebox[\textwidth][c]{
\includegraphics[scale=0.7]{figures/experiment/algo.pdf}
}
\begin{minipage}{0.195\textwidth}
        \centering
        \includegraphics[width=\textwidth]{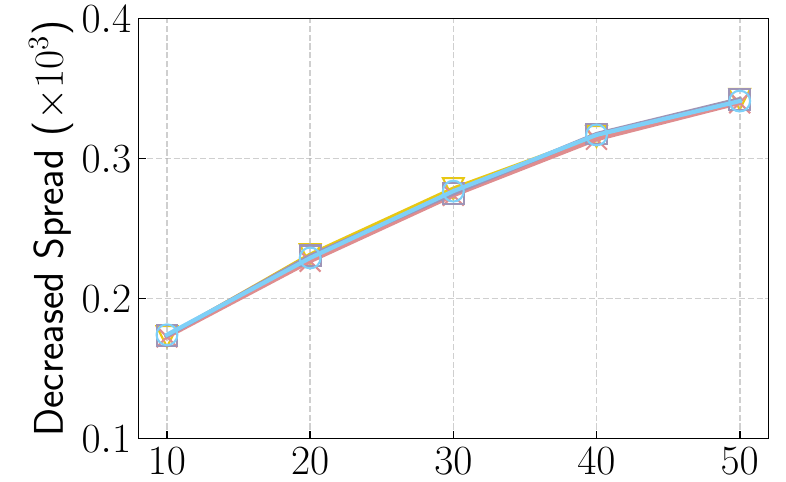}
        \subcaption{\myblue{DBLP (Vary $|S|$)}}
    \end{minipage} 
    \begin{minipage}{0.195\textwidth}
        \centering
        \includegraphics[width=\textwidth]{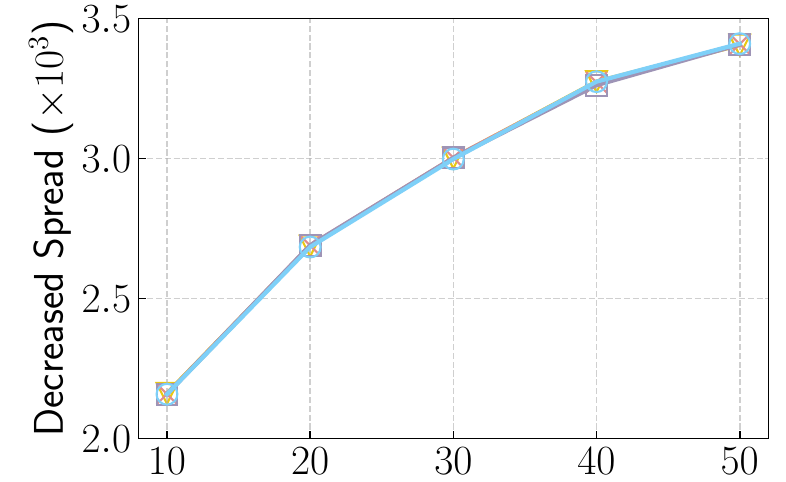}
        \subcaption{\myblue{Twitter (Vary $|S|$)}}
    \end{minipage} 
    \begin{minipage}{0.195\textwidth}
        \centering
        \begin{subfigure}{\textwidth}
            \includegraphics[width=\textwidth]{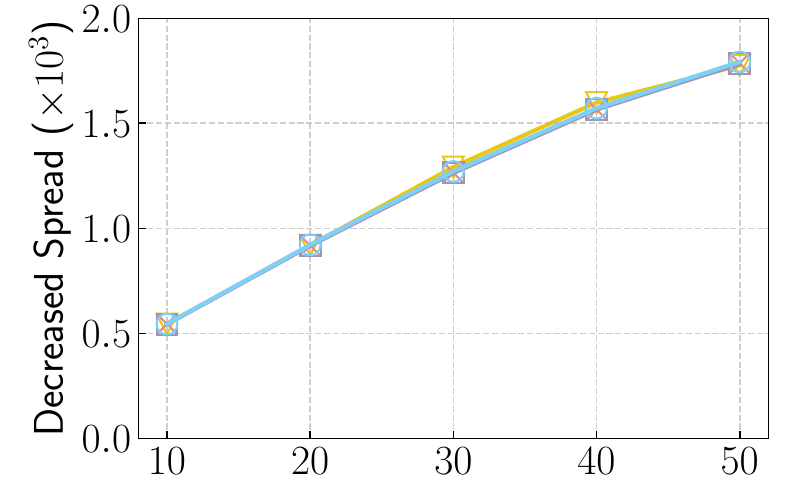}
            \subcaption{\myblue{Stanford (Vary $|S|$)}}
        \end{subfigure}
    \end{minipage}
    \begin{minipage}{0.195\textwidth}
        \centering
        \begin{subfigure}{\textwidth}
            \includegraphics[width=\textwidth]{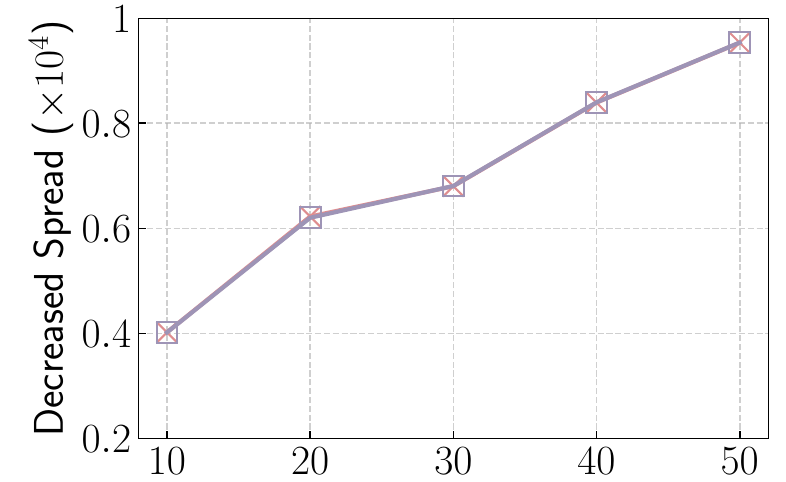}
            \subcaption{\myblue{Youtube (Vary $|S|$)}}
        \end{subfigure}
    \end{minipage}
    \begin{minipage}{0.195\textwidth}
        \centering
        \begin{subfigure}{\textwidth}
            \includegraphics[width=\textwidth]{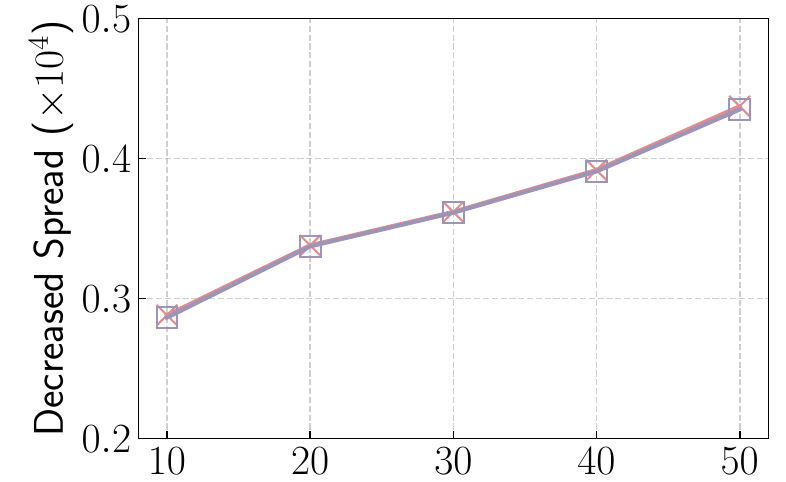}
            \subcaption{\myblue{Pokec (Vary $|S|$)}}
        \end{subfigure}
    \end{minipage}
    \begin{minipage}{0.195\textwidth}
        \centering
        \includegraphics[width=\textwidth]{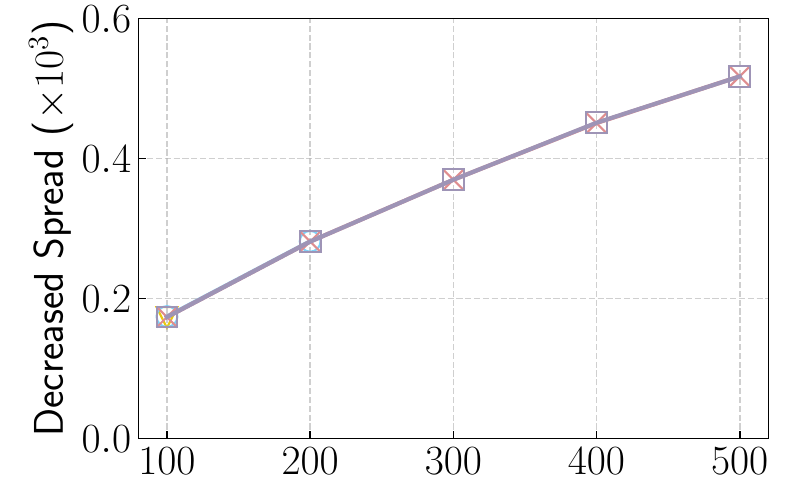}
        \subcaption{\myblue{DBLP (Vary $k$)}}
    \end{minipage} 
    \begin{minipage}{0.195\textwidth}
        \centering
        \includegraphics[width=\textwidth]{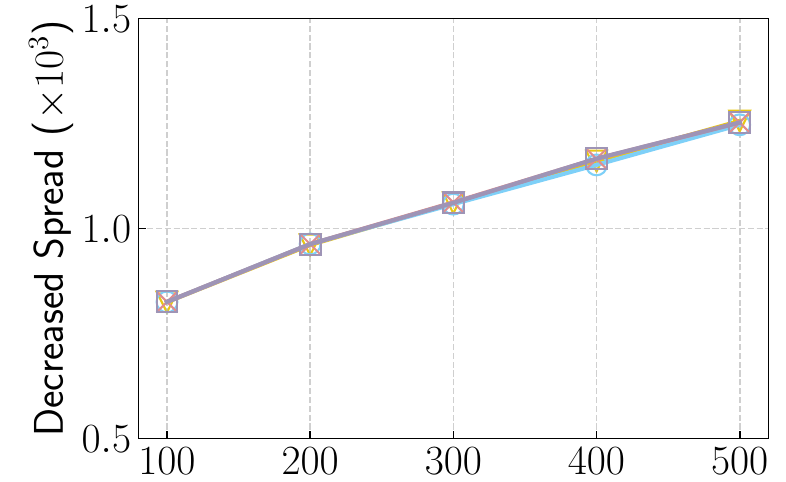}
        \subcaption{\myblue{Twitter (Vary $k$)}}
    \end{minipage} 
    \begin{minipage}{0.195\textwidth}
        \centering
        \begin{subfigure}{\textwidth}
            \includegraphics[width=\textwidth]{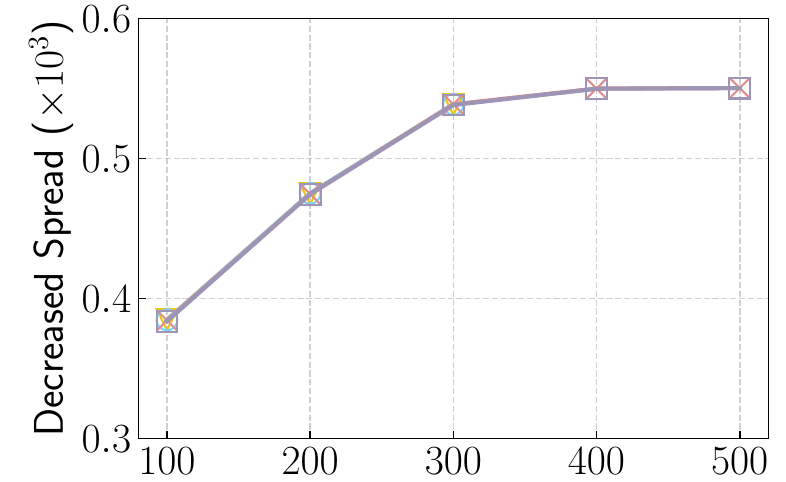}
            \subcaption{\myblue{Stanford (Vary $k$)}}
        \end{subfigure}
    \end{minipage}
    \begin{minipage}{0.195\textwidth}
        \centering
        \begin{subfigure}{\textwidth}
            \includegraphics[width=\textwidth]{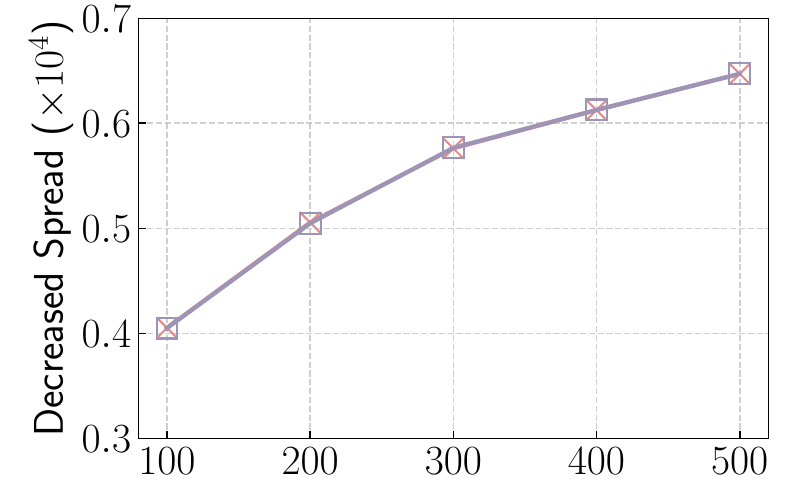}
            \subcaption{\myblue{Youtube (Vary $k$)}}
        \end{subfigure}
    \end{minipage}
    \begin{minipage}{0.195\textwidth}
        \centering
        \begin{subfigure}{\textwidth}
            \includegraphics[width=\textwidth]{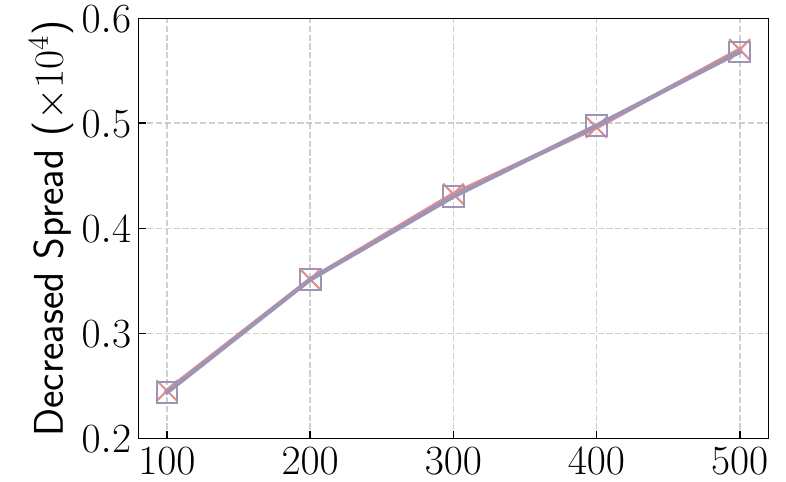}
            \subcaption{\myblue{Pokec (Vary $k$)}}
        \end{subfigure}
    \end{minipage}
	\caption{\myblue{Effectiveness evaluation by varying $|S|$ and $k$}}
	\label{fig:inf_vary_parameter}
\end{figure*}

\fixme{
\begin{table*}[h] 
\caption{\myblue{Decreased spread of GSBM, LSBM and LHGA by varying $k$}}
\label{tab:each algo}
\resizebox{1\linewidth}{!}{ 
\begin{tabular}{c|lll|lll|lll|lll|lll|lll}
\hline
 &
  \multicolumn{3}{c|}{\text{EmailCore}} &
  \multicolumn{3}{c|}{\text{EmailAll}} &
  \multicolumn{3}{c|}{\text{DBLP}} &
    \multicolumn{3}{c|}{\text{Stanford}} &
  \multicolumn{3}{c|}{\text{Youtube}} &
  \multicolumn{3}{c}{\text{Pokec}} \\ 
\text{$k$} &
  \multicolumn{1}{c}{\texttt{GSBM}} &
  \multicolumn{1}{c}{\texttt{LSBM}} &
  \multicolumn{1}{c|}{\texttt{LHGA}} &
    \multicolumn{1}{c}{\texttt{GSBM}} &
  \multicolumn{1}{c}{\texttt{LSBM}} &
  \multicolumn{1}{c|}{\texttt{LHGA}} &
  \multicolumn{1}{c}{\texttt{GSBM}} &
  \multicolumn{1}{c}{\texttt{LSBM}} &
  \multicolumn{1}{c|}{\texttt{LHGA}} &
  \multicolumn{1}{c}{\texttt{GSBM}} &
  \multicolumn{1}{c}{\texttt{LSBM}} &
  \multicolumn{1}{c|}{\texttt{LHGA}} &
  \multicolumn{1}{c}{\texttt{GSBM}} &
  \multicolumn{1}{c}{\texttt{LSBM}} &
  \multicolumn{1}{c|}{\texttt{LHGA}} &
  \multicolumn{1}{c}{\texttt{GSBM}} &
  \multicolumn{1}{c}{\texttt{LSBM}} &
  \multicolumn{1}{c}{\texttt{LHGA}} \\ \hline
\text{$k=10$} &
  \multicolumn{1}{c}{26.726} &
  \multicolumn{1}{c}{31.292} &
  \multicolumn{1}{c|}{\textbf{33.140}} &
  \multicolumn{1}{c}{\textbf{313.66}} &
  \multicolumn{1}{c}{307.643} &
  \multicolumn{1}{c|}{153.96} &
    \multicolumn{1}{c}{49.160} &
    \multicolumn{1}{c}{57.260} &
    \multicolumn{1}{c|}{\textbf{59.350
}} &
  \multicolumn{1}{c}{6921.6} &
  \multicolumn{1}{c}{\textbf{7123.6
}} &
  \multicolumn{1}{c|}{217.35} &
  \multicolumn{1}{c}{\textbf{1150.5}} &
  \multicolumn{1}{c}{705.20} &
  \multicolumn{1}{c|}{980.40} &
  \multicolumn{1}{c}{785.90} &
  \multicolumn{1}{c}{656.50} &
  \multicolumn{1}{c}{\textbf{1000.2}} \\ \hline
\text{$k=20$} &
  \multicolumn{1}{c}{\textbf{55.870}} &
  \multicolumn{1}{c}{52.544} &
  \multicolumn{1}{c|}{49.614} &
    \multicolumn{1}{c}{405.77} &
    \multicolumn{1}{c}{\textbf{455.15}} &
    \multicolumn{1}{c|}{240.40} &
  \multicolumn{1}{c}{65.790} &
  \multicolumn{1}{c}{67.150} &
  \multicolumn{1}{c|}{\textbf{74.980}} &
  \multicolumn{1}{c}{7631.3} &
  \multicolumn{1}{c}{\textbf{8108.0}} &
  \multicolumn{1}{c|}{689.32} &
  \multicolumn{1}{c}{\textbf{1962.6}} &
  \multicolumn{1}{c}{1795.9} &
  \multicolumn{1}{c|}{1295.1} &
  \multicolumn{1}{c}{\textbf{1506.9}} &
  \multicolumn{1}{c}{1302.4} &
  \multicolumn{1}{c}{1111.6} \\ \hline
\text{$k=30$} &
  \multicolumn{1}{c}{66.408} &
  \multicolumn{1}{c}{70.629} &
  \multicolumn{1}{c|}{\textbf{72.460}} &
  \multicolumn{1}{c}{538.61} &
  \multicolumn{1}{c}{\textbf{582.23}} &
  \multicolumn{1}{c|}{332.03} &
  \multicolumn{1}{c}{93.000} &
  \multicolumn{1}{c}{\textbf{98.620}} &
  \multicolumn{1}{c|}{80.850} &
  \multicolumn{1}{c}{8319.7} &
  \multicolumn{1}{c}{\textbf{8991.2}} &
  \multicolumn{1}{c|}{940.73} &
  \multicolumn{1}{c}{2246.9} &
  \multicolumn{1}{c}{2284.5} &
  \multicolumn{1}{c|}{\textbf{2310.5}} &
  \multicolumn{1}{c}{\textbf{1922.9}} &
  \multicolumn{1}{c}{1591.2} &
  \multicolumn{1}{c}{1545.3} \\ \hline
  \text{$k=40$} &
  \multicolumn{1}{c}{81.816} &
  \multicolumn{1}{c}{\textbf{83.650}} &
  \multicolumn{1}{c|}{80.379} &
  \multicolumn{1}{c}{588.76} &
  \multicolumn{1}{c}{\textbf{643.87}} &
  \multicolumn{1}{c|}{455.00} &
  \multicolumn{1}{c}{114.45} &
  \multicolumn{1}{c}{\textbf{116.67}} &
  \multicolumn{1}{c|}{92.290} &
  \multicolumn{1}{c}{8699.7} &
  \multicolumn{1}{c}{\textbf{9688.3}} &
  \multicolumn{1}{c|}{1302.6} &
  \multicolumn{1}{c}{\textbf{2743.6}} &
  \multicolumn{1}{c}{2694.1} &
  \multicolumn{1}{c|}{2406.6} &
  \multicolumn{1}{c}{2071.9} &
  \multicolumn{1}{c}{\textbf{2193.1}} &
  \multicolumn{1}{c}{1580.2 } \\ \hline
  \text{$k=50$} &
  \multicolumn{1}{c}{90.907} &
  \multicolumn{1}{c}{\textbf{102.50}} &
  \multicolumn{1}{c|}{97.018} &
  \multicolumn{1}{c}{664.88} &
  \multicolumn{1}{c}{\textbf{714.80}} &
  \multicolumn{1}{c|}{505.38} &
  \multicolumn{1}{c}{128.21} &
  \multicolumn{1}{c}{136.95} &
  \multicolumn{1}{c|}{\textbf{143.41}} &
  \multicolumn{1}{c}{9169.3} &
  \multicolumn{1}{c}{\textbf{10168}} &
  \multicolumn{1}{c|}{1991.2} &
  \multicolumn{1}{c}{3046.4} &
  \multicolumn{1}{c}{\textbf{3271.5 }} &
  \multicolumn{1}{c|}{2551.7} &
  \multicolumn{1}{c}{\textbf{2734.1}} &
  \multicolumn{1}{c}{2590.6} &
  \multicolumn{1}{c}{2179} \\ \hline
\end{tabular}
}
\end{table*}
}

\myblue{\myparagraph{Varying $|S|$ and $k$} 
Figures \ref{fig:time_vary_parameter}(a)-\ref{fig:time_vary_parameter}(e) report the response time by varying $|S|$ on the largest five datasets.
As shown, \sand and \sandl always run faster than AG and GR on all five datasets.} \sandl is faster than \sand due to the relaxed requirements.
\myblue{Generally, \sandl can achieve at least an order of magnitude speedup compared with AG and GR on all five datasets under all settings except Stanford.
In particular, AG and GR cannot complete within a day on Youtube and Pokec}, while \sandl takes only a few thousand seconds to complete. 
\myblue{In addition, with $|S|=40$ on Twitter, the response time of \sandl (resp. \sand) is 35.7 (resp. 166.7) seconds while AG (resp. GR) needs 9228.6 (resp. 9240.6) seconds to complete.
This is because AG and GR need to regenerate a large number of realizations in each iteration.}
On Twitter and Stanford, 
the gap of \sand and \sandl is smaller than that on the other three datasets. The reason is that the large number of nodes in these three datasets may result in a relatively small value for $\mathbb{E}[I_G(s)]/|V'_s|$. Under such circumstances, the Global Sampling is more prone to generating empty LRR sets, which do not contribute to the blocker set construction, and increase computation overhead.
Figures \ref{fig:time_vary_parameter}(f)-\ref{fig:time_vary_parameter}(j) present the response time by varying $k$, where similar trends can be observed.  
\myblue{In addition, AG and GR cannot finish within one day when $k$ becomes larger on all five datasets except Twitter and it is seen that \sand and \sandl are orders of magnitude faster than AG and GR on Twitter dataset.} 
Note that, the response time for our algorithms may either increase or decrease by increasing $|S|$ and $k$.
This is because the time cost of LSBM and GSBM mainly depends on when the stopping condition is reached, as we stated in Section \ref{sec:boundmaxi}. 
The larger $|S|$ ($k$) may make it easier or harder to reach the stopping condition.
Moreover, observe that on Stanford with $k=500$, \sand and \sandl only take very short time to finish. The reason lies in that our proposed algorithms can return the outgoing neighbors of $S$ as the blocker set directly when $k$ is large enough.


\subsection{Effectiveness Evaluation}

\myparagraph{Results on all the datasets} In Figure \ref{fig:inf_line}, we demonstrate the effectiveness of the proposed techniques on all the datasets with the default settings.  
\myblue{The results indicate that our solutions exhibit similar performance to AG and GR in terms of decreased spread on all the datasets.
Note that, on large datasets, \ie Youtube and Pokec, AG and GR cannot finish in a reasonable time. Therefore, the corresponding value is not shown in the figure.}

\myblue{\myparagraph{Varying $|S|$ and $k$}} 
Figure \ref{fig:inf_vary_parameter} shows the decreased spread by varying the parameters $|S|$ and $k$ on the largest five datasets. 
\myblue{It can be observed that \sand achieves the similar decreased spread to AG and GR under different parameter settings}, which reflects the effectiveness of our proposed method. In addition, the performance of \sand and \sandl is also very close. This validates that our proposed lower bound is very tight and close to the objective function, and 
the relaxation in theoretical parameters does not affect much on the real performance.
For all the algorithms, the decreased spread grows with the increase of $|S|$, since more nodes could be protected. Similarly, the decreased spread increases when $k$ becomes larger, since more blockers are selected.

\myblue{\myparagraph{Effectiveness evaluation of LSBM, GSBM and LHGA}
In Table \ref{tab:each algo}, we report the decreased spread of LSBM, GSBM, LHGA by varying $k$ on six datasets with different scales.
Recall that, \sand returns the best solution regarding the effectiveness among the results obtained by its three components.
As can be seen, within \sand, each sub-algorithm possesses the potential to outperform others in achieving the largest decrease in spread.
This demonstrates the effectiveness of the three algorithms.
In most cases, it can be seen that LSBM performs the best in terms of effectiveness, indicating that the scenario where multiple blockers jointly protect a node is not very common.
Therefore, the proposed lower bound is relatively tight w.r.t. the objective function.
Additionally, our trivial heuristic method, LHGA, exhibits the best decrease in spread in a few cases.
That is, if we remove LHGA from \sand, its empirical accuracy would decrease in such cases.
It is worth noting that LHGA almost incurs no time overhead. 
Therefore, LHGA can enhance the effectiveness of our algorithms without sacrificing efficiency.

\subsection{Sensitivity Evaluation by Varying $\epsilon$}

Figures \ref{fig:vary_ep}(a)-\ref{fig:vary_ep}(c) show the response time by varying $\epsilon$ on DBLP, Youtube and Pokec.
As observed, the time cost of \sand and \sandl is reduced with the increase of $\epsilon$, since a larger $\epsilon$ leads to smaller sample size. Figures \ref{fig:vary_ep}(d)-\ref{fig:vary_ep}(e) present the corresponding decreased spread by varying $\epsilon$.
It can be seen that with the increase of $\epsilon$, the decreased spread of \sand and \sandl slightly drops due to the smaller sample size.
For example, on Pokec, the decreased spread of \sand (resp. \sandl) is 2553.3 (resp. 2553.3) with $\epsilon=0.1$, and the reduction in spread of \sand (resp. \sandl) is 2513.6 (resp. 2511.2) with $\epsilon=0.5$.

\subsection{Approximation Quality Evaluation}
\label{sec:appro}
The approximation guarantee of SandIMIN can be seen in Eq. (\ref{data-ratio}). Obviously, the exact approximation ratio is intractable to compute, as $B^o$ and $B_L^o$ are unknown. 
According to \cite{DBLP:journals/tkde/WangYPCC17}, $\frac{(1-\gamma)^2}{(1+\gamma)^2}\cdot(1-1/e-\epsilon)\cdot\frac{\hat{D}_s(B_U)}{\hat{D}_s^U(B_U)}$ is a computable lower bound of the approximation ratio for SandIMIN.
Note that, no computable approximation ratio is provided for SandIMIN-, since it does not return $B_U$.
The average lower bound of the approximation ratio (\ie empirical approximation ratio) of SandIMIN on all the datasets (averaged over $k=10,50$ and $100$) is shown in Figure \ref{fig:ratio}. We report the results under two different settings of $\epsilon$ and $\gamma$.
In particular, on all the datasets, the empirical approximation ratio exceeds $20\%$ with $\epsilon=0.2,\gamma=0.1$, and exceeds $30\%$ with $\epsilon=0.1,\gamma=0.05$.
As a data-dependent approximation guaranteed algorithm, SandIMIN performs well in terms of approximation, as the ratios closely approximate the value of $1-1/e-\epsilon$ ($\approx$ 53.2$\%$ when $\epsilon=0.1$).}

\begin{figure}[t]
\centering
    \begin{minipage}{0.156\textwidth}
        \centering
        \includegraphics[width=\textwidth]{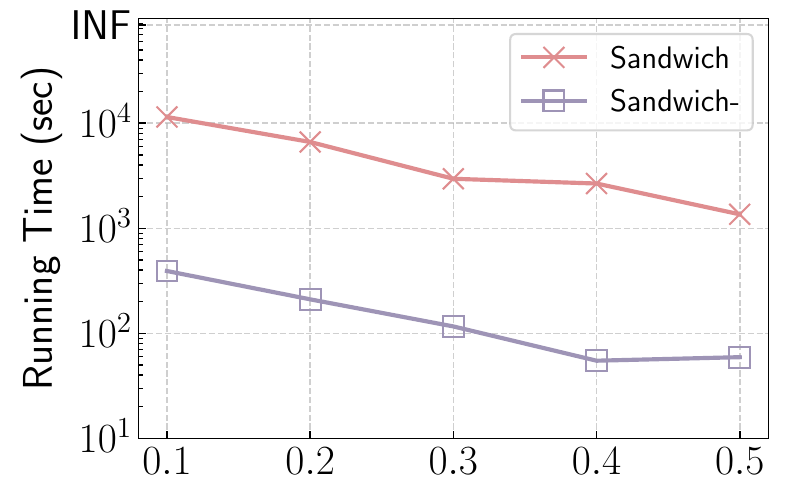}
        \subcaption{\small \myblue{DBLP (Vary $\epsilon$)}}
    \end{minipage} 
    \begin{minipage}{0.156\textwidth}
        \centering
        \includegraphics[width=\textwidth]{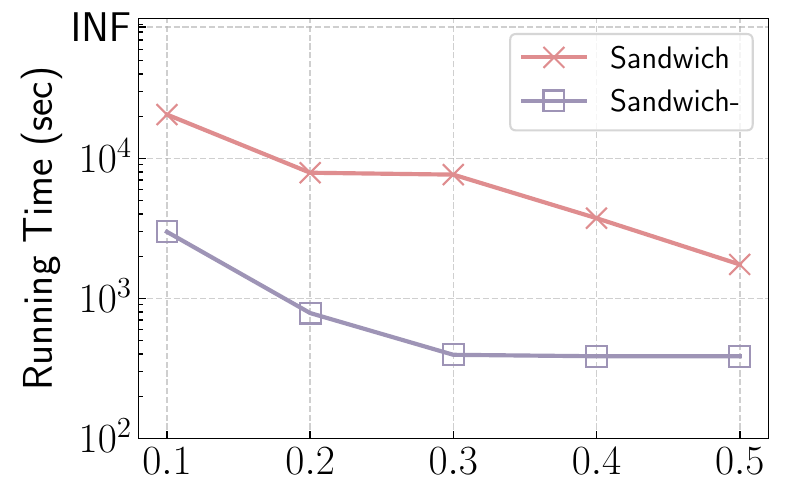}
        \subcaption{\myblue{Youtube (Vary $\epsilon$)}}
    \end{minipage} 
    \begin{minipage}{0.156\textwidth}
        \centering
        \begin{subfigure}{\textwidth}
            \includegraphics[width=\textwidth]{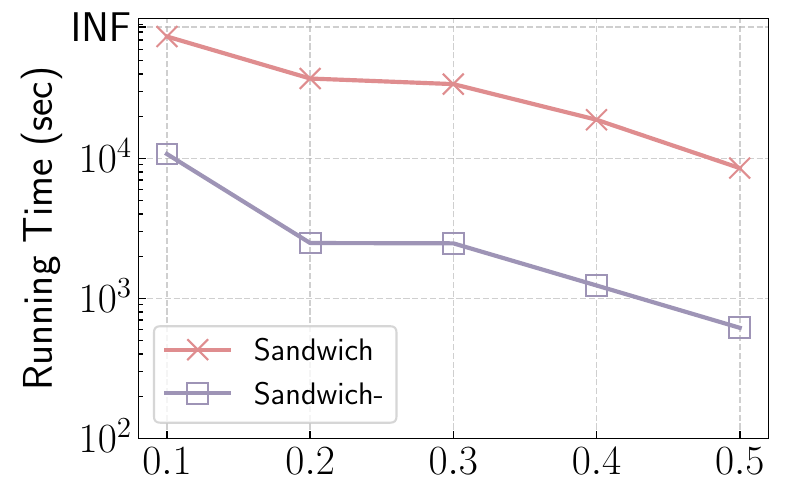}
            \subcaption{\myblue{Pokec (Vary $\epsilon$)}}
        \end{subfigure}
    \end{minipage}
    \begin{minipage}{0.156\textwidth}
        \centering
        \includegraphics[width=\textwidth]{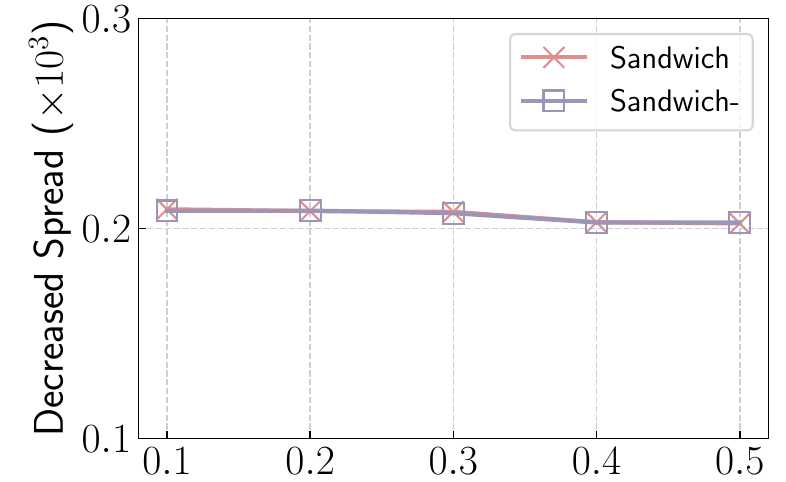}
        \subcaption{\small \myblue{DBLP (Vary $\epsilon$)}}
    \end{minipage} 
    \begin{minipage}{0.156\textwidth}
        \centering
        \includegraphics[width=\textwidth]{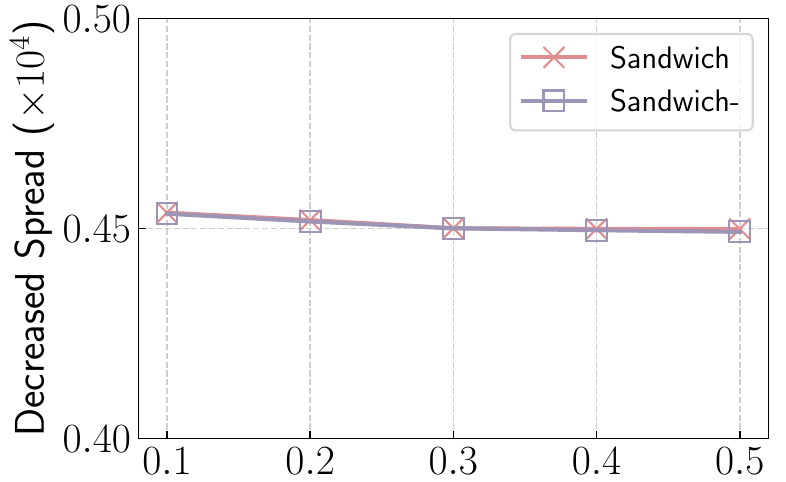}
        \subcaption{\myblue{Youtube (Vary $\epsilon$)}}
    \end{minipage} 
    \begin{minipage}{0.156\textwidth}
        \centering
        \begin{subfigure}{\textwidth}
            \includegraphics[width=\textwidth]{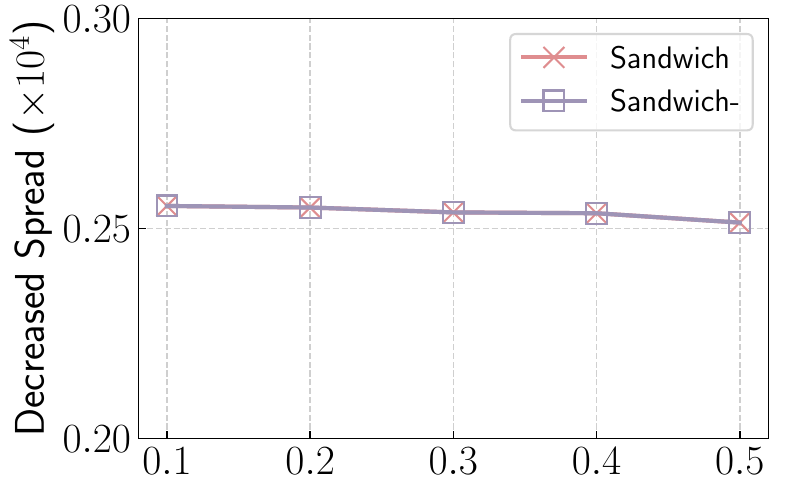}
            \subcaption{\myblue{Pokec (Vary $\epsilon$)}}
        \end{subfigure}
    \end{minipage}
	\caption{\myblue{Efficiency and effectiveness evaluation by varying $\epsilon$}}
	\label{fig:vary_ep}
    %
\end{figure}

\section{Related Work}\label{sec:related}


\myparagraph{Influence Maximization} 
The Influence Maximization (IM) problem, which aims to find a set of users with the largest expected spread, is a fundamental problem in graph analysis.
Kempe \etal~\cite{DBLP:conf/kdd/KempeKT03} first formulate IM problem and propose independent cascade (IC) as well as linear threshold (LT) models. 
In addition, they utilize a greedy algorithm that returns $(1-1/e-\epsilon)$-approximate solution. 
Afterwards, a large number of work~\cite{DBLP:conf/kdd/ChenWW10,DBLP:conf/kdd/ChenWY09,DBLP:conf/icdm/ChenYZ10,DBLP:conf/icdm/GoyalLL11,DBLP:conf/icdm/JungHC12,DBLP:conf/aaai/OhsakaAYK14,DBLP:conf/sigir/ChengSHCC14,DBLP:conf/cikm/ChengSHZC13,DBLP:conf/asunam/TangTY17,tang2018efficient,DBLP:journals/tkde/WangZZLC17} focuses on the development of heuristic algorithms to reduce the computation overhead. 
However, such solutions return the results without theoretical guarantees. 
To address this issue, Brogs \etal~\cite{DBLP:conf/soda/BorgsBCL14} propose the \textit{Reverse Influence Sampling} (RIS) technique, which reduces the time complexity to almost linear to the graph size. 
Subsequently, many RIS-based algorithms~\cite{DBLP:conf/sigmod/TangXS14,DBLP:conf/sigmod/TangSX15,DBLP:conf/sigmod/NguyenTD16,DBLP:journals/pvldb/HuangWBXL17,DBLP:conf/sigmod/TangTXY18} that ensure $(1-1/e-\epsilon)$-approximations with reduced computation overhead are proposed. 
In addition, the variants of IM have also been extensively studied, such as considering  time aspect~\cite{DBLP:conf/cikm/Khan16,DBLP:conf/icdm/LiuCXZ12} and location aspect~\cite{DBLP:conf/sigmod/LiCFTL14,DBLP:journals/tkde/WangZZL17,DBLP:conf/icde/WangZZL16}.
In~\cite{DBLP:journals/pvldb/LuCL15}, Lu \etal propose a Sandwich approximation strategy to solve non-submodular competitive and complementary IM. 
Then Wang \etal~\cite{DBLP:journals/tkde/WangYPCC17} extend Sandwich to the case where the
objective function is intractable. 
Huang \etal~\cite{DBLP:journals/pvldb/HuangLBS22} study influence maximization in closed social networks, which is non-submodular.
They resort to the influence lower bounds, which are computed with the Restricted Maximum Probability Path (RMPP) model~\cite{DBLP:conf/www/ChaojiRRB12}, to preserve submodularity.
Recently, Hu \etal~\cite{hu2023triangular} study the triangular stability maximization problem, which is also non-submodular.
They propose the Joint Baking Algorithmic Framework with theoretical guarantees to solve this problem. 

\begin{figure}
	\centering
	\includegraphics[width=0.97\linewidth]{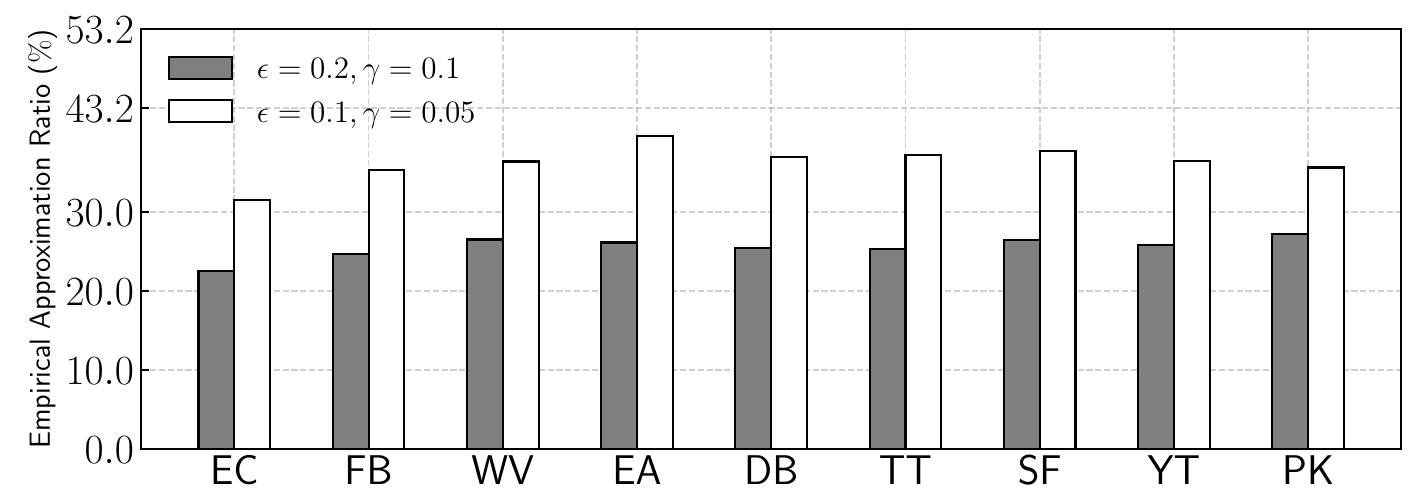}
	\caption{\myblue{Empirical approximation ratio on all the datasets}}
	\label{fig:ratio}
\end{figure}

\myparagraph{Influence Minimization} 
As an important variant of IM, Influence minimization (IMIN) has attracted great attention due to its wide applications~\cite{morozov2009swine,allcott2017social}.
Generally, the existing solutions for IMIN can be divided into three categories: positive information spreading, edge blocking and node blocking. 
Budak \etal~\cite{DBLP:conf/www/BudakAA11} first propose to spread positive information to achieve IMIN under the IC model. 
Under this strategy, the objective is shown to be monotonic and submodular. 
Based on these properties, Tong \etal~\cite{DBLP:conf/infocom/TongWGLLLD17,DBLP:conf/infocom/TongD19} later design the sampling based methods and present the algorithms that provide $(1-1/e-\epsilon)$-approximations. 
Simpson \etal~\cite{DBLP:journals/pvldb/SimpsonLH22} study a time-sensitive variant of IMIN via spreading positive information.
In \cite{DBLP:conf/pricai/KimuraSM08} and \cite{DBLP:conf/kdd/KhalilDS14}, IMIN via edge blocking is investigated under the LT and IC model, respectively.
Wang \etal~\cite{DBLP:conf/aaai/WangZCLZX13} first study IMIN through node blocking under the IC model. 
Recently, Xie \etal~\cite{DBLP:conf/icde/0002ZW0023} propose a novel approach based on dominator trees that can effectively estimate the decrease in influence of misinformation after blocking a specific node. 
However, the above solutions are all heuristic. 
Prior to our work, there is still no approximation algorithm with theoretical guarantees for IMIN via node blocking.

\section{Conclusion}\label{sec:conclusion}

In this paper, we study the influence minimization problem via node blocking. 
To our best knowledge, we are the first to propose algorithms with approximation guarantees for the problem.
Based on the Sandwich framework, we develop submodular and monotonic lower and upper bounds of the objective function and propose two sampling based methods to estimate the value of bounds.
Besides, we design novel martingale-based concentration bounds and devise two non-trivial algorithms that provide $(1-1/e-\epsilon)$-approximate solutions to maximize our proposed \myblue{bounding functions}.
We also present a lightweight heuristic for IMIN.
Finally, our algorithm, \sand, returns the best blocker set among these three solutions and yields a data-dependent approximation guarantee to the IMIN objective.
Extensive experiments over 9 real-world datasets demonstrate the effectiveness and efficiency of our proposed approaches.


\bibliographystyle{ACM-Reference-Format}
\bibliography{sample}


\begin{thebibliography}{47}


\ifx \showCODEN    \undefined \def \showCODEN     #1{\unskip}     \fi
\ifx \showDOI      \undefined \def \showDOI       #1{#1}\fi
\ifx \showISBNx    \undefined \def \showISBNx     #1{\unskip}     \fi
\ifx \showISBNxiii \undefined \def \showISBNxiii  #1{\unskip}     \fi
\ifx \showISSN     \undefined \def \showISSN      #1{\unskip}     \fi
\ifx \showLCCN     \undefined \def \showLCCN      #1{\unskip}     \fi
\ifx \shownote     \undefined \def \shownote      #1{#1}          \fi
\ifx \showarticletitle \undefined \def \showarticletitle #1{#1}   \fi
\ifx \showURL      \undefined \def \showURL       {\relax}        \fi
\providecommand\bibfield[2]{#2}
\providecommand\bibinfo[2]{#2}
\providecommand\natexlab[1]{#1}
\providecommand\showeprint[2][]{arXiv:#2}

\bibitem[\protect\citeauthoryear{Ackermann}{Ackermann}{1928}]%
        {ackermann1928hilbertschen}
\bibfield{author}{\bibinfo{person}{Wilhelm Ackermann}.} \bibinfo{year}{1928}\natexlab{}.
\newblock \showarticletitle{Zum hilbertschen aufbau der reellen zahlen}.
\newblock \bibinfo{journal}{\emph{Math. Ann.}} \bibinfo{volume}{99}, \bibinfo{number}{1} (\bibinfo{year}{1928}), \bibinfo{pages}{118--133}.
\newblock


\bibitem[\protect\citeauthoryear{Aho and Ullman}{Aho and Ullman}{1972}]%
        {DBLP:books/lib/AhoU72}
\bibfield{author}{\bibinfo{person}{Alfred~V. Aho} {and} \bibinfo{person}{Jeffrey~D. Ullman}.} \bibinfo{year}{1972}\natexlab{}.
\newblock \bibinfo{booktitle}{\emph{The theory of parsing, translation, and compiling. 1: Parsing}}.
\newblock \bibinfo{publisher}{Prentice-Hall}.
\newblock


\bibitem[\protect\citeauthoryear{Allcott and Gentzkow}{Allcott and Gentzkow}{2017}]%
        {allcott2017social}
\bibfield{author}{\bibinfo{person}{Hunt Allcott} {and} \bibinfo{person}{Matthew Gentzkow}.} \bibinfo{year}{2017}\natexlab{}.
\newblock \showarticletitle{Social media and fake news in the 2016 election}.
\newblock \bibinfo{journal}{\emph{Journal of economic perspectives}} (\bibinfo{year}{2017}).
\newblock


\bibitem[\protect\citeauthoryear{Bian, Buhmann, Krause, and Tschiatschek}{Bian et~al\mbox{.}}{2017}]%
        {DBLP:conf/icml/BianB0T17}
\bibfield{author}{\bibinfo{person}{Andrew~An Bian}, \bibinfo{person}{Joachim~M. Buhmann}, \bibinfo{person}{Andreas Krause}, {and} \bibinfo{person}{Sebastian Tschiatschek}.} \bibinfo{year}{2017}\natexlab{}.
\newblock \showarticletitle{Guarantees for Greedy Maximization of Non-submodular Functions with Applications}. In \bibinfo{booktitle}{\emph{{ICML}}} \emph{(\bibinfo{series}{Proceedings of Machine Learning Research})}, Vol.~\bibinfo{volume}{70}. \bibinfo{publisher}{{PMLR}}, \bibinfo{pages}{498--507}.
\newblock


\bibitem[\protect\citeauthoryear{Borgs, Brautbar, Chayes, and Lucier}{Borgs et~al\mbox{.}}{2014}]%
        {DBLP:conf/soda/BorgsBCL14}
\bibfield{author}{\bibinfo{person}{Christian Borgs}, \bibinfo{person}{Michael Brautbar}, \bibinfo{person}{Jennifer~T. Chayes}, {and} \bibinfo{person}{Brendan Lucier}.} \bibinfo{year}{2014}\natexlab{}.
\newblock \showarticletitle{Maximizing Social Influence in Nearly Optimal Time}. In \bibinfo{booktitle}{\emph{{SODA}}}. \bibinfo{publisher}{{SIAM}}, \bibinfo{pages}{946--957}.
\newblock


\bibitem[\protect\citeauthoryear{Budak, Agrawal, and Abbadi}{Budak et~al\mbox{.}}{2011}]%
        {DBLP:conf/www/BudakAA11}
\bibfield{author}{\bibinfo{person}{Ceren Budak}, \bibinfo{person}{Divyakant Agrawal}, {and} \bibinfo{person}{Amr~El Abbadi}.} \bibinfo{year}{2011}\natexlab{}.
\newblock \showarticletitle{Limiting the spread of misinformation in social networks}. In \bibinfo{booktitle}{\emph{{WWW}}}. \bibinfo{publisher}{{ACM}}, \bibinfo{pages}{665--674}.
\newblock


\bibitem[\protect\citeauthoryear{Chaoji, Ranu, Rastogi, and Bhatt}{Chaoji et~al\mbox{.}}{2012}]%
        {DBLP:conf/www/ChaojiRRB12}
\bibfield{author}{\bibinfo{person}{Vineet Chaoji}, \bibinfo{person}{Sayan Ranu}, \bibinfo{person}{Rajeev Rastogi}, {and} \bibinfo{person}{Rushi Bhatt}.} \bibinfo{year}{2012}\natexlab{}.
\newblock \showarticletitle{Recommendations to boost content spread in social networks}. In \bibinfo{booktitle}{\emph{{WWW}}}. \bibinfo{publisher}{{ACM}}, \bibinfo{pages}{529--538}.
\newblock


\bibitem[\protect\citeauthoryear{Chen, Wang, and Wang}{Chen et~al\mbox{.}}{2010a}]%
        {DBLP:conf/kdd/ChenWW10}
\bibfield{author}{\bibinfo{person}{Wei Chen}, \bibinfo{person}{Chi Wang}, {and} \bibinfo{person}{Yajun Wang}.} \bibinfo{year}{2010}\natexlab{a}.
\newblock \showarticletitle{Scalable influence maximization for prevalent viral marketing in large-scale social networks}. In \bibinfo{booktitle}{\emph{{KDD}}}. \bibinfo{publisher}{{ACM}}, \bibinfo{pages}{1029--1038}.
\newblock


\bibitem[\protect\citeauthoryear{Chen, Wang, and Yang}{Chen et~al\mbox{.}}{2009}]%
        {DBLP:conf/kdd/ChenWY09}
\bibfield{author}{\bibinfo{person}{Wei Chen}, \bibinfo{person}{Yajun Wang}, {and} \bibinfo{person}{Siyu Yang}.} \bibinfo{year}{2009}\natexlab{}.
\newblock \showarticletitle{Efficient influence maximization in social networks}. In \bibinfo{booktitle}{\emph{{KDD}}}. \bibinfo{publisher}{{ACM}}, \bibinfo{pages}{199--208}.
\newblock


\bibitem[\protect\citeauthoryear{Chen, Yuan, and Zhang}{Chen et~al\mbox{.}}{2010b}]%
        {DBLP:conf/icdm/ChenYZ10}
\bibfield{author}{\bibinfo{person}{Wei Chen}, \bibinfo{person}{Yifei Yuan}, {and} \bibinfo{person}{Li Zhang}.} \bibinfo{year}{2010}\natexlab{b}.
\newblock \showarticletitle{Scalable Influence Maximization in Social Networks under the Linear Threshold Model}. In \bibinfo{booktitle}{\emph{{ICDM}}}. \bibinfo{publisher}{{IEEE} Computer Society}, \bibinfo{pages}{88--97}.
\newblock


\bibitem[\protect\citeauthoryear{Cheng, Shen, Huang, Chen, and Cheng}{Cheng et~al\mbox{.}}{2014}]%
        {DBLP:conf/sigir/ChengSHCC14}
\bibfield{author}{\bibinfo{person}{Suqi Cheng}, \bibinfo{person}{Huawei Shen}, \bibinfo{person}{Junming Huang}, \bibinfo{person}{Wei Chen}, {and} \bibinfo{person}{Xueqi Cheng}.} \bibinfo{year}{2014}\natexlab{}.
\newblock \showarticletitle{IMRank: influence maximization via finding self-consistent ranking}. In \bibinfo{booktitle}{\emph{{SIGIR}}}. \bibinfo{publisher}{{ACM}}, \bibinfo{pages}{475--484}.
\newblock


\bibitem[\protect\citeauthoryear{Cheng, Shen, Huang, Zhang, and Cheng}{Cheng et~al\mbox{.}}{2013}]%
        {DBLP:conf/cikm/ChengSHZC13}
\bibfield{author}{\bibinfo{person}{Suqi Cheng}, \bibinfo{person}{Huawei Shen}, \bibinfo{person}{Junming Huang}, \bibinfo{person}{Guoqing Zhang}, {and} \bibinfo{person}{Xueqi Cheng}.} \bibinfo{year}{2013}\natexlab{}.
\newblock \showarticletitle{StaticGreedy: solving the scalability-accuracy dilemma in influence maximization}. In \bibinfo{booktitle}{\emph{{CIKM}}}. \bibinfo{publisher}{{ACM}}, \bibinfo{pages}{509--518}.
\newblock


\bibitem[\protect\citeauthoryear{Chung and Lu}{Chung and Lu}{2006}]%
        {DBLP:journals/im/ChungL06}
\bibfield{author}{\bibinfo{person}{Fan R.~K. Chung} {and} \bibinfo{person}{Lincoln Lu}.} \bibinfo{year}{2006}\natexlab{}.
\newblock \showarticletitle{Survey: Concentration Inequalities and Martingale Inequalities: {A} Survey}.
\newblock \bibinfo{journal}{\emph{Internet Math.}} \bibinfo{volume}{3}, \bibinfo{number}{1} (\bibinfo{year}{2006}), \bibinfo{pages}{79--127}.
\newblock


\bibitem[\protect\citeauthoryear{Goyal, Lu, and Lakshmanan}{Goyal et~al\mbox{.}}{2011}]%
        {DBLP:conf/icdm/GoyalLL11}
\bibfield{author}{\bibinfo{person}{Amit Goyal}, \bibinfo{person}{Wei Lu}, {and} \bibinfo{person}{Laks V.~S. Lakshmanan}.} \bibinfo{year}{2011}\natexlab{}.
\newblock \showarticletitle{{SIMPATH:} An Efficient Algorithm for Influence Maximization under the Linear Threshold Model}. In \bibinfo{booktitle}{\emph{{ICDM}}}. \bibinfo{publisher}{{IEEE} Computer Society}, \bibinfo{pages}{211--220}.
\newblock


\bibitem[\protect\citeauthoryear{Hu, Zheng, and Lian}{Hu et~al\mbox{.}}{2023}]%
        {hu2023triangular}
\bibfield{author}{\bibinfo{person}{Zheng Hu}, \bibinfo{person}{Weiguo Zheng}, {and} \bibinfo{person}{Xiang Lian}.} \bibinfo{year}{2023}\natexlab{}.
\newblock \showarticletitle{Triangular Stability Maximization by Influence Spread over Social Networks}.
\newblock \bibinfo{journal}{\emph{Proceedings of the VLDB Endowment}} \bibinfo{volume}{16}, \bibinfo{number}{11} (\bibinfo{year}{2023}), \bibinfo{pages}{2818--2831}.
\newblock


\bibitem[\protect\citeauthoryear{Huang, Wang, Bevilacqua, Xiao, and Lakshmanan}{Huang et~al\mbox{.}}{2017}]%
        {DBLP:journals/pvldb/HuangWBXL17}
\bibfield{author}{\bibinfo{person}{Keke Huang}, \bibinfo{person}{Sibo Wang}, \bibinfo{person}{Glenn~S. Bevilacqua}, \bibinfo{person}{Xiaokui Xiao}, {and} \bibinfo{person}{Laks V.~S. Lakshmanan}.} \bibinfo{year}{2017}\natexlab{}.
\newblock \showarticletitle{Revisiting the Stop-and-Stare Algorithms for Influence Maximization}.
\newblock \bibinfo{journal}{\emph{Proc. {VLDB} Endow.}} \bibinfo{volume}{10}, \bibinfo{number}{9} (\bibinfo{year}{2017}), \bibinfo{pages}{913--924}.
\newblock


\bibitem[\protect\citeauthoryear{Huang, Lin, Bao, and Sun}{Huang et~al\mbox{.}}{2022}]%
        {DBLP:journals/pvldb/HuangLBS22}
\bibfield{author}{\bibinfo{person}{Shixun Huang}, \bibinfo{person}{Wenqing Lin}, \bibinfo{person}{Zhifeng Bao}, {and} \bibinfo{person}{Jiachen Sun}.} \bibinfo{year}{2022}\natexlab{}.
\newblock \showarticletitle{Influence Maximization in Real-World Closed Social Networks}.
\newblock \bibinfo{journal}{\emph{Proc. {VLDB} Endow.}} \bibinfo{volume}{16}, \bibinfo{number}{2} (\bibinfo{year}{2022}), \bibinfo{pages}{180--192}.
\newblock


\bibitem[\protect\citeauthoryear{Jung, Heo, and Chen}{Jung et~al\mbox{.}}{2012}]%
        {DBLP:conf/icdm/JungHC12}
\bibfield{author}{\bibinfo{person}{Kyomin Jung}, \bibinfo{person}{Wooram Heo}, {and} \bibinfo{person}{Wei Chen}.} \bibinfo{year}{2012}\natexlab{}.
\newblock \showarticletitle{{IRIE:} Scalable and Robust Influence Maximization in Social Networks}. In \bibinfo{booktitle}{\emph{{ICDM}}}. \bibinfo{publisher}{{IEEE} Computer Society}, \bibinfo{pages}{918--923}.
\newblock


\bibitem[\protect\citeauthoryear{Kempe, Kleinberg, and Tardos}{Kempe et~al\mbox{.}}{2003}]%
        {DBLP:conf/kdd/KempeKT03}
\bibfield{author}{\bibinfo{person}{David Kempe}, \bibinfo{person}{Jon~M. Kleinberg}, {and} \bibinfo{person}{{\'{E}}va Tardos}.} \bibinfo{year}{2003}\natexlab{}.
\newblock \showarticletitle{Maximizing the spread of influence through a social network}. In \bibinfo{booktitle}{\emph{{KDD}}}. \bibinfo{publisher}{{ACM}}, \bibinfo{pages}{137--146}.
\newblock


\bibitem[\protect\citeauthoryear{Khalil, Dilkina, and Song}{Khalil et~al\mbox{.}}{2014}]%
        {DBLP:conf/kdd/KhalilDS14}
\bibfield{author}{\bibinfo{person}{Elias~Boutros Khalil}, \bibinfo{person}{Bistra Dilkina}, {and} \bibinfo{person}{Le Song}.} \bibinfo{year}{2014}\natexlab{}.
\newblock \showarticletitle{Scalable diffusion-aware optimization of network topology}. In \bibinfo{booktitle}{\emph{{KDD}}}. \bibinfo{publisher}{{ACM}}, \bibinfo{pages}{1226--1235}.
\newblock


\bibitem[\protect\citeauthoryear{Khan}{Khan}{2016}]%
        {DBLP:conf/cikm/Khan16}
\bibfield{author}{\bibinfo{person}{Arijit Khan}.} \bibinfo{year}{2016}\natexlab{}.
\newblock \showarticletitle{Towards Time-Discounted Influence Maximization}. In \bibinfo{booktitle}{\emph{{CIKM}}}. \bibinfo{publisher}{{ACM}}, \bibinfo{pages}{1873--1876}.
\newblock


\bibitem[\protect\citeauthoryear{Kimura, Saito, and Motoda}{Kimura et~al\mbox{.}}{2008}]%
        {DBLP:conf/pricai/KimuraSM08}
\bibfield{author}{\bibinfo{person}{Masahiro Kimura}, \bibinfo{person}{Kazumi Saito}, {and} \bibinfo{person}{Hiroshi Motoda}.} \bibinfo{year}{2008}\natexlab{}.
\newblock \showarticletitle{Solving the Contamination Minimization Problem on Networks for the Linear Threshold Model}. In \bibinfo{booktitle}{\emph{{PRICAI}}} \emph{(\bibinfo{series}{Lecture Notes in Computer Science})}, Vol.~\bibinfo{volume}{5351}. \bibinfo{publisher}{Springer}, \bibinfo{pages}{977--984}.
\newblock


\bibitem[\protect\citeauthoryear{Lengauer and Tarjan}{Lengauer and Tarjan}{1979}]%
        {DBLP:journals/toplas/LengauerT79}
\bibfield{author}{\bibinfo{person}{Thomas Lengauer} {and} \bibinfo{person}{Robert~Endre Tarjan}.} \bibinfo{year}{1979}\natexlab{}.
\newblock \showarticletitle{A Fast Algorithm for Finding Dominators in a Flowgraph}.
\newblock \bibinfo{journal}{\emph{{ACM} Trans. Program. Lang. Syst.}} \bibinfo{volume}{1}, \bibinfo{number}{1} (\bibinfo{year}{1979}), \bibinfo{pages}{121--141}.
\newblock


\bibitem[\protect\citeauthoryear{Li, Chen, Feng, Tan, and Li}{Li et~al\mbox{.}}{2014}]%
        {DBLP:conf/sigmod/LiCFTL14}
\bibfield{author}{\bibinfo{person}{Guoliang Li}, \bibinfo{person}{Shuo Chen}, \bibinfo{person}{Jianhua Feng}, \bibinfo{person}{Kian{-}Lee Tan}, {and} \bibinfo{person}{Wen{-}Syan Li}.} \bibinfo{year}{2014}\natexlab{}.
\newblock \showarticletitle{Efficient location-aware influence maximization}. In \bibinfo{booktitle}{\emph{{SIGMOD} Conference}}. \bibinfo{publisher}{{ACM}}, \bibinfo{pages}{87--98}.
\newblock


\bibitem[\protect\citeauthoryear{Liu, Cong, Xu, and Zeng}{Liu et~al\mbox{.}}{2012}]%
        {DBLP:conf/icdm/LiuCXZ12}
\bibfield{author}{\bibinfo{person}{Bo Liu}, \bibinfo{person}{Gao Cong}, \bibinfo{person}{Dong Xu}, {and} \bibinfo{person}{Yifeng Zeng}.} \bibinfo{year}{2012}\natexlab{}.
\newblock \showarticletitle{Time Constrained Influence Maximization in Social Networks}. In \bibinfo{booktitle}{\emph{{ICDM}}}. \bibinfo{publisher}{{IEEE} Computer Society}, \bibinfo{pages}{439--448}.
\newblock


\bibitem[\protect\citeauthoryear{Lowry and Medlock}{Lowry and Medlock}{1969}]%
        {DBLP:journals/cacm/LowryM69}
\bibfield{author}{\bibinfo{person}{Edward~S. Lowry} {and} \bibinfo{person}{C.~W. Medlock}.} \bibinfo{year}{1969}\natexlab{}.
\newblock \showarticletitle{Object code optimization}.
\newblock \bibinfo{journal}{\emph{Commun. {ACM}}} \bibinfo{volume}{12}, \bibinfo{number}{1} (\bibinfo{year}{1969}), \bibinfo{pages}{13--22}.
\newblock


\bibitem[\protect\citeauthoryear{Lu, Chen, and Lakshmanan}{Lu et~al\mbox{.}}{2015}]%
        {DBLP:journals/pvldb/LuCL15}
\bibfield{author}{\bibinfo{person}{Wei Lu}, \bibinfo{person}{Wei Chen}, {and} \bibinfo{person}{Laks V.~S. Lakshmanan}.} \bibinfo{year}{2015}\natexlab{}.
\newblock \showarticletitle{From Competition to Complementarity: Comparative Influence Diffusion and Maximization}.
\newblock \bibinfo{journal}{\emph{Proc. {VLDB} Endow.}} \bibinfo{volume}{9}, \bibinfo{number}{2} (\bibinfo{year}{2015}), \bibinfo{pages}{60--71}.
\newblock


\bibitem[\protect\citeauthoryear{Morozov}{Morozov}{2009}]%
        {morozov2009swine}
\bibfield{author}{\bibinfo{person}{Evgeny Morozov}.} \bibinfo{year}{2009}\natexlab{}.
\newblock \showarticletitle{Swine flu: Twitter’s power to misinform}.
\newblock \bibinfo{journal}{\emph{Foreign policy}} (\bibinfo{year}{2009}).
\newblock


\bibitem[\protect\citeauthoryear{Nemhauser, Wolsey, and Fisher}{Nemhauser et~al\mbox{.}}{1978}]%
        {DBLP:journals/mp/NemhauserWF78}
\bibfield{author}{\bibinfo{person}{George~L. Nemhauser}, \bibinfo{person}{Laurence~A. Wolsey}, {and} \bibinfo{person}{Marshall~L. Fisher}.} \bibinfo{year}{1978}\natexlab{}.
\newblock \showarticletitle{An analysis of approximations for maximizing submodular set functions - {I}}.
\newblock \bibinfo{journal}{\emph{Math. Program.}} \bibinfo{volume}{14}, \bibinfo{number}{1} (\bibinfo{year}{1978}), \bibinfo{pages}{265--294}.
\newblock


\bibitem[\protect\citeauthoryear{Nguyen, Thai, and Dinh}{Nguyen et~al\mbox{.}}{2016}]%
        {DBLP:conf/sigmod/NguyenTD16}
\bibfield{author}{\bibinfo{person}{Hung~T. Nguyen}, \bibinfo{person}{My~T. Thai}, {and} \bibinfo{person}{Thang~N. Dinh}.} \bibinfo{year}{2016}\natexlab{}.
\newblock \showarticletitle{Stop-and-Stare: Optimal Sampling Algorithms for Viral Marketing in Billion-scale Networks}. In \bibinfo{booktitle}{\emph{{SIGMOD} Conference}}. \bibinfo{publisher}{{ACM}}, \bibinfo{pages}{695--710}.
\newblock


\bibitem[\protect\citeauthoryear{Ohsaka, Akiba, Yoshida, and Kawarabayashi}{Ohsaka et~al\mbox{.}}{2014}]%
        {DBLP:conf/aaai/OhsakaAYK14}
\bibfield{author}{\bibinfo{person}{Naoto Ohsaka}, \bibinfo{person}{Takuya Akiba}, \bibinfo{person}{Yuichi Yoshida}, {and} \bibinfo{person}{Ken{-}ichi Kawarabayashi}.} \bibinfo{year}{2014}\natexlab{}.
\newblock \showarticletitle{Fast and Accurate Influence Maximization on Large Networks with Pruned Monte-Carlo Simulations}. In \bibinfo{booktitle}{\emph{{AAAI}}}. \bibinfo{publisher}{{AAAI} Press}, \bibinfo{pages}{138--144}.
\newblock


\bibitem[\protect\citeauthoryear{Simpson, Lakshmanan, and Hashemi}{Simpson et~al\mbox{.}}{2022}]%
        {DBLP:journals/pvldb/SimpsonLH22}
\bibfield{author}{\bibinfo{person}{Michael Simpson}, \bibinfo{person}{Laks V.~S. Lakshmanan}, {and} \bibinfo{person}{Farnoosh Hashemi}.} \bibinfo{year}{2022}\natexlab{}.
\newblock \showarticletitle{Misinformation Mitigation under Differential Propagation Rates and Temporal Penalties}.
\newblock \bibinfo{journal}{\emph{Proc. {VLDB} Endow.}} \bibinfo{volume}{15}, \bibinfo{number}{10} (\bibinfo{year}{2022}), \bibinfo{pages}{2216--2229}.
\newblock


\bibitem[\protect\citeauthoryear{Tang, Tang, Xiao, and Yuan}{Tang et~al\mbox{.}}{2018b}]%
        {DBLP:conf/sigmod/TangTXY18}
\bibfield{author}{\bibinfo{person}{Jing Tang}, \bibinfo{person}{Xueyan Tang}, \bibinfo{person}{Xiaokui Xiao}, {and} \bibinfo{person}{Junsong Yuan}.} \bibinfo{year}{2018}\natexlab{b}.
\newblock \showarticletitle{Online Processing Algorithms for Influence Maximization}. In \bibinfo{booktitle}{\emph{{SIGMOD} Conference}}. \bibinfo{publisher}{{ACM}}, \bibinfo{pages}{991--1005}.
\newblock


\bibitem[\protect\citeauthoryear{Tang, Tang, and Yuan}{Tang et~al\mbox{.}}{2017}]%
        {DBLP:conf/asunam/TangTY17}
\bibfield{author}{\bibinfo{person}{Jing Tang}, \bibinfo{person}{Xueyan Tang}, {and} \bibinfo{person}{Junsong Yuan}.} \bibinfo{year}{2017}\natexlab{}.
\newblock \showarticletitle{Influence Maximization Meets Efficiency and Effectiveness: {A} Hop-Based Approach}. In \bibinfo{booktitle}{\emph{{ASONAM}}}. \bibinfo{publisher}{{ACM}}, \bibinfo{pages}{64--71}.
\newblock


\bibitem[\protect\citeauthoryear{Tang, Tang, and Yuan}{Tang et~al\mbox{.}}{2018a}]%
        {tang2018efficient}
\bibfield{author}{\bibinfo{person}{Jing Tang}, \bibinfo{person}{Xueyan Tang}, {and} \bibinfo{person}{Junsong Yuan}.} \bibinfo{year}{2018}\natexlab{a}.
\newblock \showarticletitle{An efficient and effective hop-based approach for influence maximization in social networks}.
\newblock \bibinfo{journal}{\emph{Social Network Analysis and Mining}}  \bibinfo{volume}{8} (\bibinfo{year}{2018}), \bibinfo{pages}{1--19}.
\newblock


\bibitem[\protect\citeauthoryear{Tang, Shi, and Xiao}{Tang et~al\mbox{.}}{2015}]%
        {DBLP:conf/sigmod/TangSX15}
\bibfield{author}{\bibinfo{person}{Youze Tang}, \bibinfo{person}{Yanchen Shi}, {and} \bibinfo{person}{Xiaokui Xiao}.} \bibinfo{year}{2015}\natexlab{}.
\newblock \showarticletitle{Influence Maximization in Near-Linear Time: {A} Martingale Approach}. In \bibinfo{booktitle}{\emph{{SIGMOD} Conference}}. \bibinfo{publisher}{{ACM}}, \bibinfo{pages}{1539--1554}.
\newblock


\bibitem[\protect\citeauthoryear{Tang, Xiao, and Shi}{Tang et~al\mbox{.}}{2014}]%
        {DBLP:conf/sigmod/TangXS14}
\bibfield{author}{\bibinfo{person}{Youze Tang}, \bibinfo{person}{Xiaokui Xiao}, {and} \bibinfo{person}{Yanchen Shi}.} \bibinfo{year}{2014}\natexlab{}.
\newblock \showarticletitle{Influence maximization: near-optimal time complexity meets practical efficiency}. In \bibinfo{booktitle}{\emph{{SIGMOD} Conference}}. \bibinfo{publisher}{{ACM}}, \bibinfo{pages}{75--86}.
\newblock


\bibitem[\protect\citeauthoryear{Tong and Du}{Tong and Du}{2019}]%
        {DBLP:conf/infocom/TongD19}
\bibfield{author}{\bibinfo{person}{Guangmo~Amo Tong} {and} \bibinfo{person}{Ding{-}Zhu Du}.} \bibinfo{year}{2019}\natexlab{}.
\newblock \showarticletitle{Beyond Uniform Reverse Sampling: {A} Hybrid Sampling Technique for Misinformation Prevention}. In \bibinfo{booktitle}{\emph{{INFOCOM}}}. \bibinfo{publisher}{{IEEE}}, \bibinfo{pages}{1711--1719}.
\newblock


\bibitem[\protect\citeauthoryear{Tong, Wu, Guo, Li, Liu, Liu, and Du}{Tong et~al\mbox{.}}{2017}]%
        {DBLP:conf/infocom/TongWGLLLD17}
\bibfield{author}{\bibinfo{person}{Guangmo~Amo Tong}, \bibinfo{person}{Weili Wu}, \bibinfo{person}{Ling Guo}, \bibinfo{person}{Deying Li}, \bibinfo{person}{Cong Liu}, \bibinfo{person}{Bin Liu}, {and} \bibinfo{person}{Ding{-}Zhu Du}.} \bibinfo{year}{2017}\natexlab{}.
\newblock \showarticletitle{An efficient randomized algorithm for rumor blocking in online social networks}. In \bibinfo{booktitle}{\emph{{INFOCOM}}}. \bibinfo{publisher}{{IEEE}}, \bibinfo{pages}{1--9}.
\newblock


\bibitem[\protect\citeauthoryear{Wang, Zhao, Chen, Li, Zhang, and Xia}{Wang et~al\mbox{.}}{2013}]%
        {DBLP:conf/aaai/WangZCLZX13}
\bibfield{author}{\bibinfo{person}{Senzhang Wang}, \bibinfo{person}{Xiaojian Zhao}, \bibinfo{person}{Yan Chen}, \bibinfo{person}{Zhoujun Li}, \bibinfo{person}{Kai Zhang}, {and} \bibinfo{person}{Jiali Xia}.} \bibinfo{year}{2013}\natexlab{}.
\newblock \showarticletitle{Negative Influence Minimizing by Blocking Nodes in Social Networks}. In \bibinfo{booktitle}{\emph{{AAAI} (Late-Breaking Developments)}} \emph{(\bibinfo{series}{{AAAI} Technical Report})}, Vol.~\bibinfo{volume}{{WS-13-17}}. \bibinfo{publisher}{{AAAI}}.
\newblock


\bibitem[\protect\citeauthoryear{Wang, Zhang, Zhang, and Lin}{Wang et~al\mbox{.}}{2016}]%
        {DBLP:conf/icde/WangZZL16}
\bibfield{author}{\bibinfo{person}{Xiaoyang Wang}, \bibinfo{person}{Ying Zhang}, \bibinfo{person}{Wenjie Zhang}, {and} \bibinfo{person}{Xuemin Lin}.} \bibinfo{year}{2016}\natexlab{}.
\newblock \showarticletitle{Distance-aware influence maximization in geo-social network}. In \bibinfo{booktitle}{\emph{{ICDE}}}. \bibinfo{publisher}{{IEEE} Computer Society}, \bibinfo{pages}{1--12}.
\newblock


\bibitem[\protect\citeauthoryear{Wang, Zhang, Zhang, and Lin}{Wang et~al\mbox{.}}{2017b}]%
        {DBLP:journals/tkde/WangZZL17}
\bibfield{author}{\bibinfo{person}{Xiaoyang Wang}, \bibinfo{person}{Ying Zhang}, \bibinfo{person}{Wenjie Zhang}, {and} \bibinfo{person}{Xuemin Lin}.} \bibinfo{year}{2017}\natexlab{b}.
\newblock \showarticletitle{Efficient Distance-Aware Influence Maximization in Geo-Social Networks}.
\newblock \bibinfo{journal}{\emph{{IEEE} Trans. Knowl. Data Eng.}} \bibinfo{volume}{29}, \bibinfo{number}{3} (\bibinfo{year}{2017}), \bibinfo{pages}{599--612}.
\newblock


\bibitem[\protect\citeauthoryear{Wang, Zhang, Zhang, Lin, and Chen}{Wang et~al\mbox{.}}{2017c}]%
        {DBLP:journals/tkde/WangZZLC17}
\bibfield{author}{\bibinfo{person}{Xiaoyang Wang}, \bibinfo{person}{Ying Zhang}, \bibinfo{person}{Wenjie Zhang}, \bibinfo{person}{Xuemin Lin}, {and} \bibinfo{person}{Chen Chen}.} \bibinfo{year}{2017}\natexlab{c}.
\newblock \showarticletitle{Bring Order into the Samples: {A} Novel Scalable Method for Influence Maximization}.
\newblock \bibinfo{journal}{\emph{{IEEE} Trans. Knowl. Data Eng.}} \bibinfo{volume}{29}, \bibinfo{number}{2} (\bibinfo{year}{2017}), \bibinfo{pages}{243--256}.
\newblock


\bibitem[\protect\citeauthoryear{Wang, Yang, Pei, Chu, and Chen}{Wang et~al\mbox{.}}{2017a}]%
        {DBLP:journals/tkde/WangYPCC17}
\bibfield{author}{\bibinfo{person}{Zhefeng Wang}, \bibinfo{person}{Yu Yang}, \bibinfo{person}{Jian Pei}, \bibinfo{person}{Lingyang Chu}, {and} \bibinfo{person}{Enhong Chen}.} \bibinfo{year}{2017}\natexlab{a}.
\newblock \showarticletitle{Activity Maximization by Effective Information Diffusion in Social Networks}.
\newblock \bibinfo{journal}{\emph{{IEEE} Trans. Knowl. Data Eng.}} \bibinfo{volume}{29}, \bibinfo{number}{11} (\bibinfo{year}{2017}), \bibinfo{pages}{2374--2387}.
\newblock


\bibitem[\protect\citeauthoryear{Wu, Wang, Sun, Chen, Wang, and Zhang}{Wu et~al\mbox{.}}{2024}]%
        {wu2024targeted}
\bibfield{author}{\bibinfo{person}{Yanping Wu}, \bibinfo{person}{Jinghao Wang}, \bibinfo{person}{Renjie Sun}, \bibinfo{person}{Chen Chen}, \bibinfo{person}{Xiaoyang Wang}, {and} \bibinfo{person}{Ying Zhang}.} \bibinfo{year}{2024}\natexlab{}.
\newblock \showarticletitle{Targeted Filter Bubbles Mitigating via Edges Insertion}. In \bibinfo{booktitle}{\emph{Companion Proceedings of the ACM on Web Conference 2024}}. \bibinfo{pages}{734--737}.
\newblock


\bibitem[\protect\citeauthoryear{Xie, Zhang, Wang, Lin, and Zhang}{Xie et~al\mbox{.}}{2023}]%
        {DBLP:conf/icde/0002ZW0023}
\bibfield{author}{\bibinfo{person}{Jiadong Xie}, \bibinfo{person}{Fan Zhang}, \bibinfo{person}{Kai Wang}, \bibinfo{person}{Xuemin Lin}, {and} \bibinfo{person}{Wenjie Zhang}.} \bibinfo{year}{2023}\natexlab{}.
\newblock \showarticletitle{Minimizing the Influence of Misinformation via Vertex Blocking}. In \bibinfo{booktitle}{\emph{{ICDE}}}. \bibinfo{publisher}{{IEEE}}, \bibinfo{pages}{789--801}.
\newblock


\bibitem[\protect\citeauthoryear{Zhu, Tang, Tang, Wang, and Lim}{Zhu et~al\mbox{.}}{2023}]%
        {DBLP:journals/tkde/ZhuTTWL23}
\bibfield{author}{\bibinfo{person}{Yuqing Zhu}, \bibinfo{person}{Jing Tang}, \bibinfo{person}{Xueyan Tang}, \bibinfo{person}{Sibo Wang}, {and} \bibinfo{person}{Andrew Lim}.} \bibinfo{year}{2023}\natexlab{}.
\newblock \showarticletitle{2-hop+ Sampling: Efficient and Effective Influence Estimation}.
\newblock \bibinfo{journal}{\emph{{IEEE} Trans. Knowl. Data Eng.}} \bibinfo{volume}{35}, \bibinfo{number}{2} (\bibinfo{year}{2023}), \bibinfo{pages}{1088--1103}.
\newblock


\end{thebibliography}
\clearpage
\section*{Appendix}
\label{appendix}

\begin{proof}[\textbf{Proof of Lemma \ref{lemma:lower_submodular}}]
Let ${p}(s,v)$ denote the probability that $v$ is activated by $s$. 
Suppose the blocker set $B=\emptyset$. 
For any $v \in (V \backslash S)$, $D_s^L(B \cup \{v\})-D_s^L(B) \geq p(s,v) \geq 0$, thus $D_s^L(\cdot)$ is monotone nondecreasing.
Let $f(B)= |\cup_{v \in B}N_{\phi}(v)|$. Since a nonnegative linear combination of submodular functions is also submodular, we only need to prove $f(\cdot)$ is submodular for any $\phi$. 
Due to the lower bound ignoring the combination effect, we have $f(S \cup \{x\})-f(S)=|N_{\phi}(x)\backslash N_{\phi}(S)|$ and $f(T \cup \{x\})-f(T)=|N_{\phi}(x)\backslash N_{\phi}(T)|$.
Since $S \subseteq T$, we have $N_{\phi}(S)\subseteq N_{\phi}(T)$. Thus $N_{\phi}(x)\backslash N_{\phi}(S)\supseteq N_{\phi}(x)\backslash N_{\phi}(T)$. 
Accordingly, $f(S \cup \{x\})-f(S)\geq f(T \cup \{x\})-f(T)$, which shows that $D_s^L(\cdot)$ is submodular. The lemma holds.
\end{proof}

\begin{proof}[\textbf{Proof of Lemma \ref{lem:lowerest}}]
    \begin{align*}
        D^L_s(B)&=\sum_{v \in (V\backslash S)}\sum_{\phi \in \Phi(v,B)}p(\phi)\\
        &=\sum_{v \in (V\backslash S)}\sum_{\phi \in \Phi_s(v)}p(\phi)\cdot \min\{|B \cap C_{\phi}(s,v)|,1\}\\
        &=\sum_{\phi \in \Omega}p(\phi) \sum_{v \in (R_{\phi}(s)\backslash S)} \min\{|B \cap C_{\phi}(s,v)|,1\}\\
        &=\sum_{\phi \in \Omega}p(\phi)\cdot Cov_{C_{\phi}^s}(B)\\
        &=\mathbb{E}_{\Phi \sim \Omega}[Cov_{C_{\Phi}^s}(B)]=\mathbb{E}[\frac{Cov_{\mathbb{C}^s}(B)}{|\mathbb{C}^s|}]\text{,}
    \end{align*}
    where $\Phi(v,B)=\{\phi \in \Omega:\exists~u \in B~\text{s.t.}~u \in \bigcap_{i=1}^j P_i(s,v)\}$, 
    $j$ is the number of paths from $s$ to $v$ and $\Phi_s(v)=\{\phi \in \Omega:  v \in R_{\phi}(s) \}$.
\end{proof}

\begin{proof}[\textbf{Proof of Lemma \ref{lem:upper esti}}]
    \begin{align*}
       D_s^U(B) &= \sum_{v \in V_s^{\prime}}\Pr_{\Phi \sim \Omega}[B \cap L_{\Phi}(v) \neq \emptyset]\\
       &= |V_s^{\prime}|\cdot\sum_{v \in V_s^{\prime}}\Pr_{\Phi \sim \Omega}[B \cap L_{\Phi}(v) \neq \emptyset]\cdot \frac{1}{|V_s^{\prime}|}\\
       &=|V_s^{\prime}|\cdot \Pr_{\Phi \sim \Omega,v\sim V^{\prime}_s}[B \cap L_{\Phi}(v) \neq \emptyset]\\
       &=|V_s^{\prime}|\cdot \mathbb{E}[\frac{Cov_{\mathbb{L}}(B)}{|\mathbb{L}|}].
    \end{align*}
    Thus, the lemma holds.
\end{proof}

\begin{proof}[\textbf{Proof of Lemma \ref{new bound}}]
        Martingale is a sequence of random variables $Y_1,Y_2,\ldots,Y_i$ if and only if $E[Y_i\mid Y_1,Y_2,\ldots,Y_{i-1}]=Y_{i-1}$ and $E[|Y_i|]<+\infty$ for any $i$. The following lemma shows a concentration result for martingales.
    
        \begin{lemma} [\cite{DBLP:journals/im/ChungL06}]\label{martingale bound}
        Let $Y_1,Y_2,\ldots,Y_i$ be a martingale, such that $Y_1\leq a,|Y_j-Y_{j-1}|\leq a$ for any $j\in [2,i]$, and 
        \begin{align*}
            \text{Var}[Y_1]+\sum\limits_{j=2}^i\text{Var}[Y_j\mid Y_1,Y_2,\ldots,Y_{j-1}]\leq b \text{,}
        \end{align*}
        where Var$[\cdot]$ is the variance of a random variable. For any $\lambda >0$,
        \begin{align*}
            \Pr[Y_i-\mathbb{E}[Y_i]\geq \lambda] \leq \exp({-\frac{\lambda^2}{\frac{2}{3}a\lambda+2b}}).
        \end{align*}
    \end{lemma}
    
    Let $C_1^s,C_2^s,\ldots,C_{\theta}^s$ denote $\theta$ random CP sequences in $\mathbb{C}^s$, $\frac{Cov_{C_i^s}(B)}{\mathbb{E}[I_G(s)]}=x_i$ and $\frac{D^L_s(B)}{\mathbb{E}[I_G(s)]}=p$. Since each $C_{i}^s$ is generated from a random realization, for any $i \in [1,\theta]$, we have
    \begin{align*}
        \mathbb{E}[x_i\mid x_1,x_2,\ldots,x_{i-1}]=\mathbb{E}[x_i]=p.
    \end{align*}
    
    Let $M_i=\sum_{j=1}^i(x_j-p)$, thus $\mathbb{E}[M_i]=\sum_{j=1}^i\mathbb{E}[(x_j-p)]=0$, and
    \begin{align*}
        \mathbb{E}[M_i\mid M_1,M_2,\ldots,M_{i-1}]&=\mathbb{E}[M_{i-1}+(x_i-p)\mid M_1,M_2,\ldots,M_{i-1}]\\
        &=M_{i-1}+\mathbb{E}[x_i]-p=M_{i-1}.
    \end{align*}
    Therefore, $M_1,M_2,\ldots,M_{\theta}$ is a martingale.
    Since $x_i,p \in [0,1]$, we have $M_1\leq 1$ and $M_k-M_{k-1} \leq 1$ for any $k \in [2,\theta]$. In addition,  
    \begin{align*}
        &\text{Var}[M_1]+\sum_{k=2}^{\theta}\text{Var}[M_k\mid M_1,M_2,\ldots,M_{k-1}]=\sum_{k=1}^{\theta}\text{Var}[x_k].
    \end{align*}
    We know that $\text{Var}[x_i]=\mathbb{E}[x_i^2]-(\mathbb{E}[x_i])^2$ and $x_i \in [0,1]$, thus we have $\text{Var}[x_i]\leq \mathbb{E}[x_i]-(\mathbb{E}[x_i])^2=p(1-p)$. Therefore,
    \begin{align*}
        &\text{Var}[M_1]+\sum_{k=2}^{\theta}\text{Var}[M_k\mid M_1,M_2,\ldots,M_{k-1}]\leq \theta p(1-p).
    \end{align*}
    Then by Lemma \ref{martingale bound} and $M_{\theta}=\sum_{i=1}^{\theta}(x_i-p)$, we can get that Eq. (\ref{new bound1}) holds. In addition, by applying Lemma \ref{martingale bound} on the martingale $-M_1,-M_2,\ldots,-M_{\theta}$, we can deduce that Eq. (\ref{new bound2}) holds.
\end{proof}

\begin{proof}[\textbf{Proof of Lemma \ref{lemma:imax}}]
    Let $\theta_1=\frac{2\mathbb{E}[I_G(s)]\ln(12/\delta)}{\epsilon_1^2\cdot D_s^L(B^o_L)}$ and when $\theta \geq \theta_1$, we have:
    \begin{align}
        \notag &\Pr[Cov_{\mathbb{C}_1^s}(B_L^o)/\theta \leq (1-\epsilon_1)D_s^L(B_L^o)]\\
        \notag =&\Pr[\frac{Cov_{\mathbb{C}_1^s}(B_L^o)}{\mathbb{E}[I_G(s)]}-\frac{D_s^L(B_L^o)\cdot \theta}{\mathbb{E}[I_G(s)]}\leq -\frac{\epsilon_1\cdot\theta\cdot D_s^L(B_L^o)}{\mathbb{E}[I_G(s)]}]\\
        \leq &\exp{[-\frac{\epsilon_1^2}{2}\cdot \theta \cdot \frac{D_s^L(B_L^o)}{\mathbb{E}[I_G(s)]}]} \leq  \frac{\delta}{12}\text{,}
    \end{align}
    where $\epsilon_1 < \epsilon$ and the first inequality is due to Eq. (\ref{new bound2}). Since $D_s^L(\cdot)$ is submodular, we can also deduce that $Cov_{\mathbb{C}_1^s}(\cdot)$ is submodular, thus $Cov_{\mathbb{C}_1^s}(B_L) \geq (1-1/e)Cov_{\mathbb{C}_1^s}(B_L^o)$ and we have:
    \begin{align}
        \Pr[\frac{Cov_{\mathbb{C}_1^s}(B_L)}{\theta} \geq (1-1/e)(1-\epsilon_1)D_s^L(B_L^o)]\geq 1-\frac{\delta}{12}\label{eq15}.
    \end{align}
    Let $\theta_2=\frac{(2-2/e)\mathbb{E}[I_G(s)]\ln(\binom{n-|S|}{k}\cdot12/\delta)}{D_s^L(B_L^o)(\epsilon-(1-1/e)\epsilon_1)^2}$ and $\epsilon_2=\epsilon-(1-1/e)\cdot \epsilon_1$, when $\theta \geq \theta_2$, we assume $D_s^L(B_L) < (1-1/e-\epsilon)D_s^L(B^o_L)$, thus:
    \begin{align}
        \notag &\Pr[\frac{Cov_{\mathbb{C}_1^s}(B_L)}{\theta}-D_s^L(B_L)\geq \epsilon_2 D_s^L(B_L^o)]\\
        \notag =&\Pr[\frac{Cov_{\mathbb{C}_1^s}(B_L)}{\mathbb{E}[I_G(s)]}-\frac{\theta \cdot D_s^L(B_L)}{\mathbb{E}[I_G(s)]} \geq \frac{\theta \cdot \epsilon_2 \cdot D_s^L(B_L^o)}{\mathbb{E}[I_G(s)]}]\\
        \notag \leq &\exp(-\frac{\epsilon_2^2\cdot D_s^L(B_L^o)^2\cdot \theta}{(2D_s^L(B_L)+\frac{2}{3}D_s^L(B_L^o)\cdot \epsilon_2)\cdot \mathbb{E}[I_G(s)]})\\
         \leq &\exp(-\frac{\epsilon_2^2\cdot D_s^L(B_L^o) \cdot \theta}{(2(1-1/e-\epsilon)+\frac{2}{3}\epsilon_2)\cdot \mathbb{E}[I_G(s)]})
        \leq  \frac{\delta}{12\cdot \binom{n-|S|}{k}}\label{eq16}\text{,}
    \end{align}
    where the first inequality is due to Eq. (\ref{new bound1}). According to Eq. (\ref{eq15}), Eq. (\ref{eq16}) and there exists at most $\binom{n-|S|}{k}$ blocker sets, when $\theta \geq \max\{\theta_1,\theta_2\}$, we can get:
    \begin{align}
        \Pr[D_s^L(B_L) \geq (1-1/e-\epsilon)D_s^L(B^o_L)]\geq 1-\frac{\delta}{6}\text{,}
    \end{align}
    which contradicts the assumption. Based on Line 4 of Algorithm \ref{alg:LSBM}, we have $\Pr\left[(1-\beta)\mathbb{E}[I_G(s)]\leq\hat{I}_G(s)\leq(1+\beta)\mathbb{E}[I_G(s)]\right]\geq 1-\delta/6$. In addition, $\text{OPT}^L$ is the lower bound of $D_s^L(B_L^o)$, thus we have $\theta_{\max} \geq \max\{\theta_1,\theta_2\}$ holds at least $1-\delta/6$ probability. Accordingly, by the union bound, the proof of this lemma is done.
\end{proof}

\begin{proof}[\textbf{Proof of Lemma \ref{lemma:bounds}}]
    \begin{align*}
        &\Pr\left[\frac{D_s^L(B_L)}{\mathbb{E}[I_G(s)]}<\left((\sqrt{\frac{Cov_{\mathbb{C}^s_2}(B_L)}{\mathbb{E}[I_G(s)]}+\frac{2a_1}{9}}-\sqrt{\frac{a_1}{2}})^2-\frac{a_1}{18}\right)\cdot\frac{1}{|\mathbb{C}_2^s|}\right]\\
        =&\Pr\left[\frac{Cov_{\mathbb{C}^s_2}(B_L)}{\mathbb{E}[I_G(s)]}-\frac{D_s^L(B_L)\cdot|\mathbb{C}_2^s|}{\mathbb{E}[I_G(s)]}>\sqrt{2a_1\frac{D_s^L(B_L)\cdot|\mathbb{C}_2^s|}{\mathbb{E}[I_G(s)]}+\frac{a_1^2}{9}}+\frac{a_1}{3}\right]\\
        \leq & \exp{(-a_1)}=\frac{\delta}{3i_{\max}}.
    \end{align*}
     Let $x=Cov_{\mathbb{C}^s_2}(B_L)/\mathbb{E}[I_G(s)]$ and $f(x)=(\sqrt{x+\frac{2a_1}{9}}-\sqrt{\frac{a_1}{2}})^2-\frac{a_1}{18})\cdot\frac{1}{|\mathbb{C}_2^s|}$. When $x \geq 5a_1/18$, $f(x)$ monotonically increasing, otherwise $f(x)$ decreases monotonically. Thus, when $x\geq Cov_{\mathbb{C}^s_2}(B_L)\cdot(1-\beta)/\hat{I}_G(s)\geq 5a_1/18$, $f(x)$ monotonically increasing and we have:
     \begin{align}
        \notag \frac{D_s^L(B_L)}{\mathbb{E}[I_G(s)]}\geq f(x) \geq f(\frac{Cov_{\mathbb{C}^s_2}(B_L)\cdot(1-\beta)}{\hat{I}_G(s)}).
     \end{align}
     Similarly, when $x \leq Cov_{\mathbb{C}^s_2}(B_L)\cdot(1+\beta)/\hat{I}_G(s)\leq 5a_1/18$, $f(x)$ decreases monotonically and we have:
     \begin{align*}
          \frac{D_s^L(B_L)}{\mathbb{E}[I_G(s)]}\geq f(x) \geq f(\frac{Cov_{\mathbb{C}^s_2}(B_L)\cdot(1+\beta)}{\hat{I}_G(s)}).
     \end{align*}
     Therefore, $\sigma^L(B_L)$ is the lower bound of $\frac{D_s^L(B_L)}{\mathbb{E}[I_G(s)]}$ with at least $1-\frac{\delta}{3i_{\max}}$ probability. Similar to the proof of Lemma 5.2 in \cite{DBLP:conf/sigmod/TangTXY18}, we can get $Cov_{\mathbb{C}^s_1}^u(B_L^o)$ is the upper bound of $Cov_{\mathbb{C}^s_1}(B_L^o)$, where $B_i$ be a set containing the $i$ nodes that are selected in the first $i$ iterations of the Procedure Max-Coverage and $maxMC(B_i,k)$ denotes the set of $k$ nodes with the $k$ largest marginal coverage in $\mathbb{C}^s_1$ with respect to $B_i$. Thus, 
     \begin{align*}
         &\Pr\left[ \frac{D_s^L(B_L^o)}{\mathbb{E}[I_G(s)]}> \left(\sqrt{\frac{Cov_{\mathbb{C}^s_1}^u(B_L^o)}{\mathbb{E}[I_G(s)]}+\frac{a_2}{2}}+\sqrt{\frac{a_2}{2}}\right)^2\cdot\frac{1}{|\mathbb{C}_1^s|}\right]\\
         \leq&\Pr\left[ \frac{D_s^L(B_L^o)}{\mathbb{E}[I_G(s)]}> \left(\sqrt{\frac{Cov_{\mathbb{C}^s_1}(B_L^o)}{\mathbb{E}[I_G(s)]}+\frac{a_2}{2}}+\sqrt{\frac{a_2}{2}}\right)^2\cdot\frac{1}{|\mathbb{C}_1^s|}\right]\\
         \leq&\Pr\left[ \frac{Cov_{\mathbb{C}^s_1}(B_L^o)}{\mathbb{E}[I_G(s)]}-\frac{|\mathbb{C}_1^s|\cdot D_s^L(B_L^o)}{\mathbb{E}[I_G(s)]}<-\sqrt{2a_2\cdot\frac{|\mathbb{C}_1^s|\cdot D_s^L(B_L^o)}{\mathbb{E}[I_G(s)]}}\right]\\
         \leq&\exp{(-a_2)}=\frac{\delta}{3i_{\max}}.
     \end{align*}
    Let $y=\frac{Cov_{\mathbb{C}^s_1}^u(B_L^o)}{\mathbb{E}[I_G(s)]}$ and $g(y)=\left(\sqrt{y+\frac{a_2}{2}}+\sqrt{\frac{a_2}{2}}\right)^2\cdot\frac{1}{|\mathbb{C}_1^s|}$. $g(y)$ always increases monotonically and we have:
    \begin{align*}
        \frac{D_s^L(B_L^o)}{\mathbb{E}[I_G(s)]} \leq g(y) \leq g(\frac{Cov_{\mathbb{C}^s_1}^u(B_L^o)\cdot(1+\beta)}{\hat{I}_G(s)}).
    \end{align*}
    Therefore, $\sigma^U(B^o_L)$ is the upper bound of $\frac{D_s^L(B_L^o)}{\mathbb{E}[I_G(s)]}$ with at least $1-\frac{\delta}{3i_{\max}}$ probability.
\end{proof}

\begin{proof}[\textbf{Proof of Theorem \ref{theorem:time complexity_LSBM}}]
    The primary time overhead of LSBM lies in (i) estimating the influence of seeds of misinformation; (ii) generating the CP sequences; (iii) executing the Max-Coverage algorithm and computing $\sigma^L(B_L)$ and $\sigma^U(B_L^o)$ in all iterations.

    As shown in \cite{DBLP:journals/tkde/ZhuTTWL23}, the time complexity of influence estimation is $\mathcal{O}(\frac{m\cdot\ln{1/\delta}}{\beta^2})$. Then. we analyze the number of CP sequences generated by LSBM.

    Let $\epsilon_1=\epsilon$, $\tilde{\epsilon}_1=\epsilon/e$, 
    $\hat{\epsilon}_1=\sqrt{\frac{2a_3\mathbb{E}[I_G(s)]}{D_s^L(B_L^o)\theta_1}}$, $\epsilon_2=\sqrt{\frac{2a_3\mathbb{E}[I_G(s)]}{D_s^L(B_L)\theta_2}}$, 
    $\tilde{\epsilon}_2=(\sqrt{\frac{2a_3D_s^L(B_L)\theta_2}{\mathbb{E}[I_G(s)]}+\frac{a_3^2}{9}}+\frac{a_3}{3})\cdot\frac{\mathbb{E}[I_G(s)]}{D_s^L(B_L)\theta_2}$,
    $a_3=c\ln{(\frac{3i_{\max}}{\delta})}$ for any $c \geq 1$.
    In addition, let 
    $\theta_a=\frac{2\mathbb{E}[I_G(s)]\ln{\frac{6}{\delta}}}{(1-1/e-\epsilon)\epsilon_1^2D_s^L(B_L^o)}$, $\theta_b=\frac{(2+2\tilde{\epsilon}_1/3)\mathbb{E}[I_G(s)]\ln{\frac{6\binom{n-|S|}{k}}{\delta}}}{\tilde{\epsilon}_1^2 D_s^L(B_L^o)}$, 
    $\theta_c=\frac{27\mathbb{E}[I_G(s)] \ln{\frac{3i_{\max}}{\delta}}(1+\beta)^2}{(1-1/e-\epsilon)(\epsilon_1+\epsilon_1\beta-2\beta)^2 D_s^L(B_L^o)}$, $\theta_d=\frac{5\ln{\frac{3i_{\max}}{\delta}} \mathbb{E}[I_G(s)]}{18(1-\epsilon_1)(1-1/e-\epsilon)D_s^L(B_L^o)}$ and $\theta^{\prime}=\max\{\theta_a,\theta_b,\theta_c,\theta_d\}$. 
    It is easy to verify that
    \begin{align}
        \theta^{\prime}=\mathcal{O}\left(\frac{(k\ln{(n-|S|)}+\ln{1/\delta})\mathbb{E}[I_G(s)]}{(\epsilon+\epsilon\beta-2\beta)^2 D_s^L(B_L^o)}\right).
    \end{align}
    When $\theta_1=\theta_2=c\theta^{\prime}$, based on Eq. (\ref{new bound1}) and Eq. (\ref{new bound2}), we have:
    \begin{align}
        &\Pr[\frac{Cov_{\mathbb{C}_1^s}(B_L^o)}{\theta_1}<(1-\epsilon_1)\cdot D_s^L(B_L^o)]\leq \left(\frac{\delta}{6}\right)^c\text{,}\label{time_eq_1}\\
        &\Pr[\frac{Cov_{\mathbb{C}_2^s}(B_L)}{\theta_2}<(1-\epsilon_1)\cdot D_s^L(B_L)]\leq \left(\frac{\delta}{6}\right)^c\text{,}\label{time_eq_6}\\
        &\Pr[\frac{Cov_{\mathbb{C}_1^s}(B_L)}{\theta_1}>D_s^L(B_L)+\tilde{\epsilon}_1\cdot D_s^L(B_L^o)]\leq \left(\frac{\delta}{6\binom{n-|S|}{k}}\right)^c\text{,}\label{time_eq_2}\\
        &\Pr[\frac{Cov_{\mathbb{C}_1^s}(B_L^o)}{\theta_1}<(1-\hat{\epsilon}_1)\cdot D_s^L(B_L^o)] \leq \left(\frac{\delta}{3i_{\max}}\right)^c\text{,}\label{time_eq_3}\\
        &\Pr[\frac{Cov_{\mathbb{C}_2^s}(B_L)}{\theta_2}<(1-\epsilon_2)\cdot D_s^L(B_L)]\leq \left(\frac{\delta}{3i_{\max}}\right)^c\text{,}\label{time_eq_4}\\
        &\Pr[\frac{Cov_{\mathbb{C}_2^s}(B_L)}{\theta_2}>(1+
        \tilde{\epsilon}_2)\cdot D_s^L(B_L)]\leq \left(\frac{\delta}{3i_{\max}}\right)^c\label{time_eq_5}.
    \end{align}
    Specially, when $\theta_1\geq c\theta_a$, Eq. (\ref{time_eq_1}) holds; when $\theta_1 \geq c\theta_b$, Eq. (\ref{time_eq_2}) holds. Eq. (\ref{time_eq_3})-Eq. (\ref{time_eq_5}) are obtained based on the definition of $\hat{\epsilon}_1,\epsilon_2$ and $\tilde{\epsilon}_2$.
    When the event in Eq. (\ref{time_eq_1}) and Eq. (\ref{time_eq_2}) not happen, we get:
        \begin{align}\label{eq25}
            D_s^L(B_L)\geq (1-1/e-\epsilon)D_s^L(B_L^o).
        \end{align}
    Based on Eq. (\ref{eq25}), when when $\theta_1\geq c\theta_a$, Eq. (\ref{time_eq_6}) holds.
    Since $B_L$ is not independent of $\mathbb{C}_1^s$ and there are at most $\binom{n-|S|}{k}$ blocker sets, based on the union bound, the probability that none of the events in Eq. (\ref{time_eq_1})-Eq. (\ref{time_eq_5}) happens is at least:
    \begin{align*}
        1-\left((\frac{\delta}{6})^c\cdot2+(\frac{\delta}{6\binom{n-|S|}{k}})^c\cdot\binom{n-|S|}{k}+i_{\max}\cdot(\frac{\delta}{3i_{\max}})^c\right)\geq 1-\delta^c.
    \end{align*}
    
    And we have:
    \begin{align*}
        &\hat{\epsilon}_1 \leq \sqrt{\frac{2(1-1/e-\epsilon)(\epsilon_1+\epsilon_1\beta-2\beta)^2}{27(1+\beta)^2}}\leq \frac{\epsilon_1+\epsilon_1\beta-2\beta}{3(1+\beta)}\text{,}\\
        &\epsilon_2 \leq \sqrt{\frac{2(1-1/e-\epsilon)D_s^L(B_L^o)(\epsilon_1+\epsilon_1\beta-2\beta)^2}{27(1+\beta)^2 D_s^L(B_L)}}\leq \frac{\epsilon_1+\epsilon_1\beta-2\beta}{3(1+\beta)}\text{,}\\
        &\tilde{\epsilon}_2\leq \sqrt{\frac{(\epsilon_1+\epsilon_1\beta-2\beta)^2(2+2\tilde{\epsilon}_2/3)}{27(1+\beta)^2}}
        \leq \frac{\epsilon_1+\epsilon_1\beta-2\beta}{3(1+\beta)}.
    \end{align*}
    In addition, when the event in Eq. (\ref{time_eq_3}) not happen, we have:
    \begin{align*}
        \left(\sqrt{\frac{Cov_{\mathbb{C}^s_1}(B_L^o)\cdot(1+\beta)}{\hat{I}_G(s)}+\frac{a_3}{2}}+\sqrt{\frac{a_3}{2}}\right)^2\cdot\frac{1}{\theta_1}\geq \frac{D_s^L(B_L^o)}{\mathbb{E}[I_G(s)]}.
    \end{align*}
    Thus, it holds that:
    \begin{align*}
        1-\hat{\epsilon}_1&=1-\sqrt{\frac{2a_3\mathbb{E}[I_G(s)]}{D_s^L(B_L^o)\theta_1}}\\
        &\leq1-\frac{\sqrt{2a_3}}{\sqrt{Cov_{\mathbb{C}^s_1}(B_L^o)\cdot(1+\beta)/\hat{I}_G(s)+\frac{a_3}{2}}+\sqrt{\frac{a_3}{2}}}\\
        &\leq \frac{Cov^u_{\mathbb{C}^s_1}(B_L^o)\cdot(1+\beta)/\hat{I}_G(s)}{(\sqrt{Cov^u_{\mathbb{C}^s_1}(B_L^o)\cdot(1+\beta)/\hat{I}_G(s)+\frac{a_3}{2}}+\sqrt{\frac{a_3}{2}})^2}.
    \end{align*}
    Since $a_2=\ln{(\frac{3i_{\max}}{\delta})} \leq a_3$, based on Line 19 of Algorithm \ref{alg:LSBM}, thus
    \begin{align}
        \notag\sigma^U(B_L^o) &\leq \left(\sqrt{\frac{Cov_{\mathbb{C}^s_1}(B_L^o)\cdot(1+\beta)}{\hat{I}_G(s)}+\frac{a_3}{2}}+\sqrt{\frac{a_3}{2}}\right)^2\cdot\frac{1}{\theta_1}\\
        &\leq \frac{Cov^u_{\mathbb{C}^s_1}(B_L^o)\cdot(1+\beta)/\hat{I}_G(s)}{1-\hat{\epsilon}_1}\cdot\frac{1}{\theta_1}.\label{eq:time_up}
    \end{align}
    When $\theta_2 \geq \theta_d$ and according to Eq. (\ref{time_eq_6}), we have:
    \begin{align*}
        \frac{Cov_{\mathbb{C}^s_2}(B_L)}{\mathbb{E}[I_G(s)]}\geq \frac{\theta_2 \cdot(1-\epsilon_1)D_s^L(B_L)}{\mathbb{E}[I_G(s)]}\geq \frac{5a_1}{18}.
    \end{align*}
    Thus, $f(x)=(\sqrt{x+\frac{2a_1}{9}}-\sqrt{\frac{a_1}{2}})^2-\frac{a_1}{18})$ monotonically increasing. In addition, when the event in Eq. (\ref{time_eq_5}) does not happen, we have:
    \begin{align*}
        \left((\sqrt{\frac{Cov_{\mathbb{C}^s_2}(B_L)\cdot(1-\beta)}{\hat{I}_G(s)}+\frac{2a_3}{9}}-\sqrt{\frac{a_3}{2}})^2-\frac{a_3}{18}\right)\cdot\frac{1}{\theta_2}\leq \frac{D_s^L(B_L)}{\mathbb{E}[I_G(s)]}.
    \end{align*}
    Thus, it holds that:
    \begin{align*}
       &\frac{Cov_{\mathbb{C}^s_2}(B_L)\cdot(1-\beta)}{\hat{I}_G(s)}-\frac{\tilde{\epsilon}_2 D_s^L(B_L)\cdot\theta_2}{\mathbb{E}[I_G(s)]}\\
       =&\frac{Cov_{\mathbb{C}^s_2}(B_L)\cdot(1-\beta)}{\hat{I}_G(s)}-(\sqrt{2a_3\frac{D_s^L(B_L)\cdot \theta_2}{\mathbb{E}[I_G(s)]}+\frac{a_3^2}{9}}+\frac{a_3}{3})\\
       \leq &\frac{Cov_{\mathbb{C}^s_2}(B_L)\cdot(1-\beta)}{\hat{I}_G(s)}-(\sqrt{\frac{2Cov_{\mathbb{C}^s_2}(B_L)(1-\beta)a_3}{\hat{I}_G(s)}+\frac{4a_3^2}{9}}-\frac{2a_3}{3})\\
       =&\left(\sqrt{\frac{Cov_{\mathbb{C}^s_2}(B_L)\cdot(1-\beta)}{\hat{I}_G(s)}+\frac{2a_3}{9}}-\sqrt{\frac{a_3}{2}}\right)^2-\frac{a_3}{18}.
    \end{align*}
    Since $a_1 \leq a_3$, based on the Line 14 of Algorithm \ref{alg:LSBM}, thus
    \begin{align}
        \sigma^L(B_L) \geq \frac{Cov_{\mathbb{C}^s_2}(B_L)\cdot(1-\beta)}{\hat{I}_G(s)\cdot\theta_2}-\frac{\tilde{\epsilon}_2 D_s^L(B_L)}{\mathbb{E}[I_G(s)]}.\label{eq:time_low}
    \end{align}
    Putting Eq. (\ref{eq:time_up}) and Eq. (\ref{eq:time_low}) together, when none of the events in Eq. (\ref{time_eq_1})-Eq. (\ref{time_eq_5}) happens, we have:
    \begin{align*}
        \frac{\sigma^L(B_L)}{\sigma^U(B_L^o)}&\geq \frac{\frac{Cov_{\mathbb{C}^s_2}(B_L)\cdot(1-\beta)}{\hat{I}_G(s)\cdot\theta_2}-\frac{\tilde{\epsilon}_2 D_s^L(B_L)}{\mathbb{E}[I_G(s)]}}{ \frac{Cov^u_{\mathbb{C}^s_1}(B_L^o)\cdot(1+\beta)/\hat{I}_G(s)}{1-\hat{\epsilon}_1}\cdot\frac{1}{\theta_1}}\\
        &\geq \frac{\theta_1\left(\frac{Cov_{\mathbb{C}^s_2}(B_L)\cdot(1-\beta)}{(1+\beta)\cdot\theta_2}-\tilde{\epsilon}_2\cdot D_s^L(B_L)\right)(1-\hat{\epsilon}_1)}{Cov^u_{\mathbb{C}^s_1}(B_L^o)}\\
        &\geq \frac{\theta_1\left(\frac{1-\beta}{1+\beta}\cdot(1-\epsilon_2)-\tilde{\epsilon}_2\right)\cdot D_s^L(B_L) (1-\hat{\epsilon}_1)}{Cov^u_{\mathbb{C}^s_1}(B_L^o)}\\
        &\geq \frac{\theta_1\left(\frac{1-\beta}{1+\beta}\cdot(1-\epsilon_2)-\tilde{\epsilon}_2\right)\cdot D_s^L(B_L) (1-\hat{\epsilon}_1)}{Cov_{\mathbb{C}^s_1}(B_L)}(1-1/e)\\
        &\geq \frac{\theta_1\left(1-\epsilon_2-\frac{2\beta}{1+\beta}-\tilde{\epsilon}_2-\hat{\epsilon}_1\right)\cdot D_s^L(B_L)}{Cov_{\mathbb{C}^s_1}(B_L)}(1-1/e)\\
        &\geq \frac{\theta_1\left(1-\epsilon_1\right)\cdot D_s^L(B_L)}{Cov_{\mathbb{C}^s_1}(B_L)}(1-1/e)\\
        &\geq \frac{\theta_1\left(1-\epsilon_1\right)\cdot \left(Cov_{\mathbb{C}_1^s}(B_L)/\theta_1-\tilde{\epsilon}_1\cdot D_s^L(B_L^o)\right)}{Cov_{\mathbb{C}^s_1}(B_L)}(1-1/e)\\
        &\geq \frac{\left(1-\epsilon_1\right)\cdot \left(Cov_{\mathbb{C}_1^s}(B_L)-\tilde{\epsilon}_1\cdot \frac{Cov_{\mathbb{C}^s_1}(B_L^o)}{(1-\epsilon_1)}\right)}{Cov_{\mathbb{C}^s_1}(B_L)}(1-1/e)\\
        &\geq (1-\epsilon_1)\left(1-\frac{\tilde{\epsilon}_1}{(1-\epsilon_1)(1-1/e)}\right)(1-1/e)\\
        &=1-1/e-\epsilon.
    \end{align*}
    Therefore, when $\theta_1=\theta_2=c\theta^{\prime}$ CP sequences are generated, LSBM does not stop only if at least one of the events in Eq. (\ref{time_eq_1})-Eq. (\ref{time_eq_5}) happens. The probability is at most $\delta^c$.

    Let $j$ be the first iteration in which the number of CP sequences generated by LSBM reaches $\theta^{\prime}$. From this iteration onward, the expected number of CP sequences further generated is at most
    \begin{align*}
        2\cdot \sum_{z \geq j}\theta_0\cdot 2^z \cdot \delta^{2^{z-j}}&=2\cdot2^j\cdot \theta_0 \sum_{z=0}2^z\cdot \delta^{2^z}\\
        &\leq 4 \theta^{\prime}\sum_{z=0}2^{-2^z+z}\\
        &\leq 4 \theta^{\prime}\sum_{z=0} 2^{-z}\leq 8 \theta^{\prime}.
    \end{align*}
    If the algorithm stops before this iteration, there are at most $2\theta^{\prime}$ CP sequences generated. Therefore, the expected number of CP sequences generated is less than $10\theta^{\prime}$, which is 
    \begin{align}
        \mathcal{O}(\frac{(k\ln{(n-|S|)}+\ln{1/\delta})\mathbb{E}[I_G(s)]}{(\epsilon+\epsilon\beta-2\beta)^2 D_s^L(B_L^o)}).
    \end{align}
    We have shown that the expected time required to generate a CP sequence is $\mathcal{O}(m\cdot \alpha(m,n))$. Based on Wald's equation, LSBM requires $\mathcal{O}(\frac{(k\ln{(n-|S|)}+\ln{1/\delta})\mathbb{E}[I_G(s)]m\cdot \alpha(m,n)}{(\epsilon+\epsilon\beta-2\beta)^2 D_s^L(B_L^o)})$ in CP sequences generation. In addition, the total expected time used for executing the Max-Coverage and computing $\sigma^L(B_L)$ and $\sigma^U(B_L^o)$ in all the iterations is
    \begin{align*}
        &\mathcal{O}(k(n-|S|)\cdot i_{\max}+2\mathbb{E}[|\mathbb{C}_1^s\cup\mathbb{C}_2^s|]\cdot\mathbb{E}[|C^s|])\\
        =&\mathcal{O}(\frac{(k\ln{(n-|S|)}+\ln{1/\delta})\mathbb{E}[I_G(s)]}{(\epsilon+\epsilon\beta-2\beta)^2}).
    \end{align*}
    Thus, the theorem holds.
\end{proof}

\begin{lemma}
    Given a blocker set $B$, a seed node $s$ and a fixed number of $\theta$ random LRR sets $\mathbb{L}$. For any $\lambda >0$,
    \begin{align}
        &\Pr[Cov_{\mathbb{L}}(B)-\frac{D^U_s(B)\cdot \theta}{|V_s'|} \geq \lambda] \leq \exp({-\frac{\lambda^2}{\frac{2D_s^U(B)}{|V_s'|}\cdot \theta+\frac{2}{3}\lambda}}),\label{29}\\
        &\Pr[Cov_{\mathbb{L}}(B)-\frac{D^U_s(B)\cdot \theta}{|V_s'|} \leq -\lambda] \leq \exp({-\frac{\lambda^2}{\frac{2D_s^U(B)}{|V_s'|}\cdot \theta}})\label{30}.
    \end{align}
\end{lemma}

\begin{proof}
    Let $L_1,L_2,\ldots,L_{\theta}$ denote $\theta$ random LRR sets in $\mathbb{L}$, $Cov_{L_i}(B)=x_i$ and $\frac{D^U_s(B)}{|V_s'|}=p$. Since each $L_{i}$ is generated from a random realization, for any $i \in [1,\theta]$, we have
    \begin{align*}
        \mathbb{E}[x_i\mid x_1,x_2,\ldots,x_{i-1}]=\mathbb{E}[x_i]=p.
    \end{align*}
    
    Let $M_i=\sum_{j=1}^i(x_j-p)$, thus $\mathbb{E}[M_i]=\sum_{j=1}^i\mathbb{E}[(x_j-p)]=0$, and
    \begin{align*}
        \mathbb{E}[M_i\mid M_1,M_2,\ldots,M_{i-1}]&=\mathbb{E}[M_{i-1}+(x_i-p)\mid M_1,M_2,\ldots,M_{i-1}]\\
        &=M_{i-1}+\mathbb{E}[x_i]-p=M_{i-1}.
    \end{align*}
    Therefore, $M_1,M_2,\ldots,M_{\theta}$ is a martingale.
    Since $x_i,p \in [0,1]$, we have $M_1\leq 1$ and $M_k-M_{k-1} \leq 1$ for any $k \in [2,\theta]$. In addition,  
    \begin{align*}
        &\text{Var}[M_1]+\sum_{k=2}^{\theta}\text{Var}[M_k\mid M_1,M_2,\ldots,M_{k-1}]=\sum_{k=1}^{\theta}\text{Var}[x_k].
    \end{align*}
    We know that $\text{Var}[x_i]=\mathbb{E}[x_i^2]-(\mathbb{E}[x_i])^2$ and $x_i \in [0,1]$, thus we have $\text{Var}[x_i]\leq \mathbb{E}[x_i]-(\mathbb{E}[x_i])^2=p(1-p)$. Therefore,
    \begin{align*}
        &\text{Var}[M_1]+\sum_{k=2}^{\theta}\text{Var}[M_k\mid M_1,M_2,\ldots,M_{k-1}]\leq \theta p(1-p).
    \end{align*}
    Then by Lemma \ref{martingale bound} and $M_{\theta}=\sum_{i=1}^{\theta}(x_i-p)$, we can get that Eq. (\ref{29}) holds. In addition, by applying Lemma \ref{martingale bound} on the martingale $-M_1,-M_2,\ldots,-M_{\theta}$, we can deduce that Eq. (\ref{30}) holds.
\end{proof}

\begin{proof}[\textbf{Proof of Theorem 5.7}]
The proof of Theorem 5.7 is similar to that of Theorem 5.5 and the differences between them are that we use Lemma 8.1 as the concentration bounds and we set:

\begin{align*}
    &\theta_1=\frac{2|V_s'|\ln(6/\delta)}{\epsilon_1^2\cdot D_s^U(B^o_U)},\\
    &\theta_2=\frac{(2-2/e)|V_s'|\ln(\binom{n-|S|}{k}\cdot6/\delta)}{D_s^U(B_U^o)(\epsilon-(1-1/e)\epsilon_1)^2},\epsilon_2=\epsilon-(1-1/e)\cdot \epsilon_1,\\
    &\theta_{\max}=\frac{2|V_s^{\prime}|\left( (1-1/e)\sqrt{\ln\frac{6}{\delta}}+\sqrt{(1-1/e)(\ln\binom{|V_s^{\prime}|-|S|}{k}+\ln\frac{6}{\delta})} \right)^2}{\epsilon^2 \text{OPT}^L}\text{,}\\
    &\theta_0=\theta_{\max}\cdot\epsilon^2 \text{OPT}^L/|V_s^{\prime}|\text{,}\\
    &\sigma^L(B_U)=\left((\sqrt{Cov_{\mathbb{L}_2}(B_U)+\frac{2a_1}{9}}-\sqrt{\frac{a_1}{2}})^2-\frac{a_1}{18}\right)\cdot\frac{|V_s^{\prime}|}{|\mathbb{L}_2|}\text{,}\\   &\sigma^U(B_U^o)\leftarrow\left(\sqrt{Cov_{\mathbb{L}_1}^u(B_U^o)+\frac{a_2}{2}}+\sqrt{\frac{a_2}{2}}\right)^2\cdot\frac{|V_s^{\prime}|}{|\mathbb{L}_1|}.
\end{align*}
Note that, $\sigma^L(B_U)$ is the lower bound of $D_s^U(B_U)$ with at least $1-\frac{\delta}{3i_{\max}}$ probability, and $\sigma^U(B^o_U)$ is the upper bound of $D_s^U(B_U^o)$ with at least $1-\frac{\delta}{3i_{\max}}$ probability.
\end{proof}

\begin{proof}[\textbf{Proof of Theorem 5.8}]
The proof of Theorem 5.8 is similar to that of Theorem 5.6 and the differences between them are that we use Lemma 8.1 as the concentration bounds and we set:

    $\epsilon_1=\epsilon$, $\tilde{\epsilon}_1=\epsilon/e$, 
    $\hat{\epsilon}_1=\sqrt{\frac{2a_3|V_s'|}{D_s^U(B_U^o)\theta_1}}$, $\epsilon_2=\sqrt{\frac{2a_3|V_s'|}{D_s^U(B_U)\theta_2}}$, 
    $\tilde{\epsilon}_2=(\sqrt{\frac{2a_3D_s^U(B_U)\theta_2}{|V_s'|}+\frac{a_3^2}{9}}+\frac{a_3}{3})\cdot\frac{|V_s'|}{D_s^U(B_U)\theta_2}$,
    $a_3=c\ln{(\frac{3i_{\max}}{\delta})}$ for any $c \geq 1$.
    In addition, let 
    $\theta_a=\frac{2|V_s'|\ln{\frac{6}{\delta}}}{\epsilon_1^2D_s^U(B_U^o)}$, $\theta_b=\frac{(2+2\tilde{\epsilon}_1/3)|V_s'|\ln{\frac{6\binom{n-|S|}{k}}{\delta}}}{\tilde{\epsilon}_1^2 D_s^U(B_U^o)}$, 
    $\theta_c=\frac{27|V_s'| \ln{\frac{3i_{\max}}{\delta}}}{(1-1/e-\epsilon)\epsilon_1^2 D_s^U(B_U^o)}$, and $\theta^{\prime}=\max\{\theta_a,\theta_b,\theta_c\}$.
    
    Based on the above setting, the expected number of LRR sets generated by GSBM is proved to be $\mathcal{O}(\frac{(k\ln{(|V_s^{\prime}|-|S|)}+\ln{(1/\delta)) |V_s^{\prime}|}}{\epsilon^2\cdot D_s^U(B_U^o)})$.
    Besides, let $\mathcal{L}$ be the expected number of edges being traversed required to generate one LRR set for a randomly selected node in $|V_s^{\prime}|$.
    Then, $\mathcal{L}=\mathbb{E}[p_{L}\cdot m]$, where the expectation is taken over the random choices of $L$.
    $v^{*}$ is chosen from $|V_s^{\prime}|$ according to a probability that is proportional to its in-degree, and we define the random variable:

    \begin{equation*}
    X(v^{*},L)=\left\{
    	\begin{aligned}
    	1 &\quad\text{if $v^{*} \in L$} \\
    	0 &\quad\text{otherwise} \\
    	\end{aligned}
    	\right
     .
    \end{equation*}
    Then, for any fixed $L$, we have:
    \begin{align*}
        p_{L}=\sum_{v^{*}}\left(Pr[v^{*}]\cdot X(v^{*},L)\right).
    \end{align*}
    Consider the fixed $v^{*}$ and vary $L$, we have:
    \begin{align*}
        p_{v^{*}}=\sum_{L}\left(Pr[L]\cdot X(v^{*},L)\right).
    \end{align*}
    In addition, by Lemma \ref{lem:upper esti}, we have $\mathbb{E}[p_{v^{*}}]=D_s^U(v^{*})/|V_s^{\prime}|$. Therefore, 
    \begin{align*}
        \mathcal{L}/m&=\mathbb{E}[p_{L}]=\sum_{L}(Pr[L]\cdot p_{L})\\
        &=\sum_{L}\left(Pr[L]\cdot \sum_{v^{*}}\left(Pr[v^{*}]\cdot X(v^{*},L)\right)\right)\\
        &=\sum_{v^{*}}\left(Pr[v^{*}]\cdot \sum_{L}\left(Pr[L]\cdot X(v^{*},L)\right)\right)\\
        &=\sum_{v^{*}}\left(Pr[v^{*}]\cdot p_{v^{*}}\right)=\mathbb{E}[p_{v^{*}}]=D_s^U(v^{*})/|V_s^{\prime}|.
    \end{align*}
    Thus, we have $\mathcal{L}=D_s^U(v^{*})\cdot\frac{m}{|V_s^{\prime}|}$. Based on Wald's equation, we can show that GSBM requires $\mathcal{O}(\frac{(k\ln{(|V_s^{\prime}|-|S|)}+\ln{(1/\delta)) m}}{\epsilon^2})$ expected time in LRR set generation.
    Furthermore, the total expected time used for executing the Max-Coverage and computing $\sigma^L(B_U)$ and $\sigma^U(B_U^o)$ in all the iterations is
    \begin{align*}
        &\mathcal{O}(k(n-|S|)\cdot i_{\max}+2\mathbb{E}[|\mathbb{L}_1\cup\mathbb{L}_2|]\cdot\mathbb{E}[|L|])\\
        =&\mathcal{O}(\frac{(k\ln{(n-|S|)}+\ln({1/\delta}))|V_s^{\prime}|}{\epsilon^2}).
    \end{align*}
    Thus, the theorem holds.
    \end{proof}

\end{document}